\documentclass[10pt,a4paper]{article}

\usepackage{amsmath}
\usepackage{amsthm}
\usepackage{amssymb}
\usepackage{amsfonts}
\usepackage{graphicx}
\usepackage{latexsym}
\usepackage{graphicx}
\usepackage{marvosym}
\usepackage{hyperref}
\usepackage{float}
\usepackage{color}
\def\f12{\frac 1 2}

\def\a{\alpha}

\def\hh{\mathcal{H}^{+}}

\def\f12{\frac 1 2}

\newcommand{\nabb}{\mbox{$\nabla \mkern-13mu /$\,}}
\newcommand{\lapp}{\mbox{$\triangle \mkern-13mu /$\,}}

\newtheorem{remark}{Remark}[section]
\newtheorem{lemma}{Lemma}[subsection]
\newtheorem{theorem}{Theorem}[section]
\newtheorem{proposition}{Proposition}[subsection]

\newtheorem{corollary}{Corollary}[subsection]

\newtheorem{mytheo}{Theorem}

\begin{document}
\title{Decay of Axisymmetric Solutions of the Wave Equation on Extreme Kerr Backgrounds}

\author{Stefanos Aretakis\thanks {University of Cambridge,
Department of Pure Mathematics and Mathematical Statistics,
Wilberforce Road, Cambridge, CB3 0WB, United Kingdom}}
\date{October 10, 2011}
\maketitle

\begin{abstract}
We study the Cauchy problem for the  wave equation $\Box_{g}\psi=0$ on extreme Kerr backgrounds under axisymmetry. Specifically, we consider regular axisymmetric initial data prescribed on a Cauchy hypersurface $\Sigma_{0}$ which connects the future event horizon with spacelike or null infinity, and we solve the linear wave equation on the domain of dependence of $\Sigma_{0}$. We show that the spacetime integral of an energy-type density is bounded by the initial conserved flux corresponding to the stationary Killing field $T$, and we derive boundedness of the non-degenerate energy flux corresponding to a globally timelike vector field $N$. Finally, we prove uniform pointwise boundedness and power-law decay for $\psi$. Our estimates hold up to and including the event horizon $\mathcal{H}^{+}$. We remark that these results do not yield decay for the derivatives transversal to $\hh$, and this is suggestive that these derivatives may satisfy instability properties analogous to those shown in our previous work on extreme  Reissner-Nordstr\"{o}m backgrounds.
\end{abstract}
\tableofcontents

\section{Introduction}
\label{sec:Introduction}
The \textit{Kerr family} $(\mathcal{M}_{M,a}, g_{M,a}),$ with $|a|\leq M$, constitutes a two-parameter family of rotating stationary black hole solutions to the Einstein-vacuum equations. The wave equation for the \textit{subextreme} case $(|a|<M)$ has been definitively understood very recently \cite{megalaa}. In this paper, we consider \textit{extreme Kerr} backgrounds $(\mathcal{M}_{M},g_{M})$ corresponding to $|a|=M$, and we study the Cauchy problem for the wave equation
\begin{equation}
\Box_{g}\psi=0
\label{wave}
\end{equation}
with axisymmetric initial data prescribed on a Cauchy hypersurface $\Sigma_{0}$ crossing the future event horizon $\hh$ and terminating at spacelike or null infinity.  For an introduction to the relevant notions see Section \ref{sec:GeometryOfExtremeKerr}.
 \begin{figure}[H]
	\centering
		\includegraphics[scale=0.115]{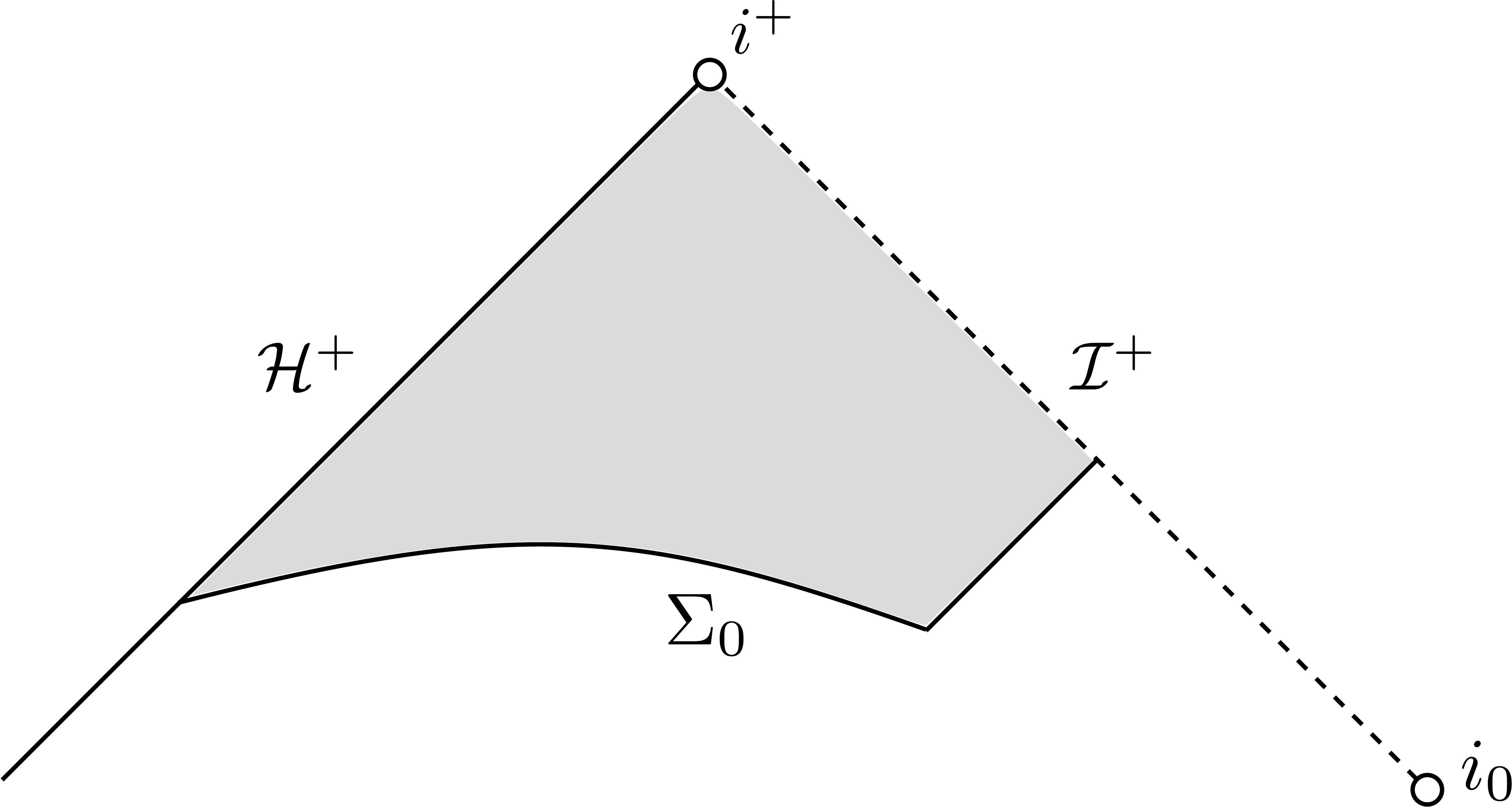}
	\label{fig:kerr00}
\end{figure}
 The main results of this paper include:
\begin{enumerate}
	\item Integrated local energy decay, up to and including the event horizon $\hh$ (Theorems \ref{t2}, \ref{t1}).
	\item Energy and pointwise uniform boundedness of solutions, up to and including $\hh$ (Theorems \ref{t1}, \ref{t4}).
	\item Power-law energy and pointwise decay of solutions, up to and including $\hh$ (Theorems \ref{t3}, \ref{t5}).
\end{enumerate}

Extreme black holes are central objects of study in the high-energy physics community (see \cite{marolf}); however, no results were previously known for the evolution of linear waves on extreme Kerr. In this context, the fundamental aspect of extreme black holes is the degeneracy of the redshift effect on $\hh$. We remark that the assumption of axisymmetry allows us to concentrate on this aspect without having to additionally deal with superradiance (see Section \ref{sec:OverviewOfResultsAndTechniques} for more details). For comments on the non-axisymmetric case, see Section \ref{sec:mnot0}.

 The author has previously studied the wave equation on a simpler model of extreme black holes, namely the spherically symmetric charged solution of Reissner-Nordstr\"{o}m \cite{aretakis1, aretakis2}. Solutions on such backgrounds were shown to exhibit both stability and instability properties.  The results of the present paper are of the form suggested by the stability results of \cite{aretakis1, aretakis2}. The analogues of the instability results will be addressed in a subsequent paper.

\subsection{Previous Work}
\label{sec:PreviousWork}

To put our results into context, we briefly summarise  previous mathematical work on the linear wave equation on black hole spacetimes.

\subsubsection{Schwarzschild and Kerr for $|a|<M$}
\label{sec:SchwarzschildAndKerrForAM}
 Work on the wave equation \eqref{wave} on black holes spacetimes began in 1957 for the Schwarzschild case $(a=0)$ with the pioneering work of Regge and Wheeler \cite{RW}, but the first complete quantitative result (uniform boundedness) was obtained only in 1989 by Kay and Wald \cite{wa1}. During the last decade,  ``$X$ estimates" providing integrated local energy decay for Schwarzschild  were derived   in \cite{blu0,blu3,dr3}. In particular, \cite{dr3} introduced a vector field estimate which captures in a stable manner the so-called  \textit{redshift effect}, which allowed the authors to obtain quantitative pointwise estimates on the horizon $\mathcal{H}^{+}$.  Refinements for Schwarzschild were achieved in \cite{dr5} and \cite{tataru1}. For results on the  wave equation coupled with the Einstein(-Maxwell) system under spherical symmetry see \cite{chrin,price}. See also \cite{dr4,dr5,other1,laul,luk,tataru1,volker,volker1} and for an exhaustive list of references see \cite{md}.

With Schwarzschild understood, interest then shifted to the wave equation on Kerr backgrounds. The first uniform boundedness result for general solutions of the wave equation on slowly rotating Kerr ($\left|a\right|\ll M$) spacetimes was proven in \cite{dr7} and decay results (again for $\left|a\right|\ll M$) were derived in \cite{mikraa,tataru2,bluekerr}. Decay results for general subextreme Kerr  spacetimes ($\left|a\right|<M$) were proven for axisymmetric solutions $\psi$ in \cite{mikraa}  and for general solutions in  \cite{megalaa}.   Two new methods \cite{new,tataru3} were presented recently for obtaining definitive decay estimates for the energy flux as well as pointwise decay given integrated local energy decay bounds. Exponential decay for Kerr-de Sitter (i.e.~Kerr backgrounds with positive cosmological constant) was obtained in \cite{semyon1, semyon2}. The slowly rotating Kerr-AdS was treated in \cite{kostas}. For previous work on mode analysis of the wave equation see \cite{whiting}; see also \cite{finster1}.

\subsubsection{Extreme Reissner-Nordstr\"{o}m}
\label{sec:ExtremeReissnerNordstrOM}
The fundamentally new aspect of extreme black holes is the degeneracy of the redshift effect along the event horizon (see Section \ref{sec:TheRedshiftEffect}). This necessitates developing new methods, since the results available for subextreme black hole backgrounds exploit--in one way or another--this effect.

The tools needed for understanding the properties of solutions to the wave equation on \textit{spherically symmetric} extreme black holes were developed in \cite{aretakis1, aretakis2}, where both stability and instability results were shown for extreme Reissner-Nordstr\"{o}m\footnote{The Reissner-Nordstr\"{o}m family is a 2-parameter family of spherically symmetric asymptotically flat Lorentzian manifolds $(\mathcal{M}_{M,e},g_{M,e})$ which satisfy the Einstein-Maxwell equations. The extreme case corresponds to $|e|=M$. See also \cite{aretakislong}}. Specifically, energy and pointwise decay is shown for $\psi$, whereas  \textbf{non-decay} (and \textbf{blow-up}) is shown for the  transversal to $\mathcal{H}^{+}$ (higher) derivative of $\psi$. Note that these instabilities are in sharp contrast with the subextreme case for which decay holds for all higher order derivatives of $\psi$ along $\mathcal{H}^{+}$.  

We remark that the techniques introduced in  \cite{aretakis1, aretakis2} heavily exploit the spherical symmetry of the background spacetime and thus break down in the case of extreme Kerr. This is then the subject of the present paper.

\subsection{Overview of Results and Techniques}
\label{sec:OverviewOfResultsAndTechniques}

In the present paper, we show stability results for axisymmetric solutions $\psi$ of the form suggested by the results of \cite{aretakis1, aretakis2} for extreme Reissner-Nordstr\"{o}m. The instability results for extreme Kerr will be provided in a subsequent paper.

\subsubsection{Conservation of Degenerate Energy}
\label{sec:ConservationOfDegenerateEnergy}

The Killing vector field $T$ (see Section \ref{sec:GeometryOfExtremeKerr} for details) is spacelike in a region outside $\hh$ known as the \textit{ergoregion}, and therefore, the energy flux corresponding to  $T$  fails to be non-negative definite. This phenomenon is called \textit{superradiance}. Nonetheless,  superradiance is completely absent for axisymmetric solutions $\psi$ to the wave equation: If $\psi$ is axisymmetric then the conserved energy flux corresponding to the Killing field $T$ is non-negative definite, yielding an a priori bound. However, as in the Schwarzschild case, this flux degenerates at the event horizon $\hh$. See Section \ref{sec:AxisymmetryVsSuperradiance}.

\subsubsection{Integrated Local Energy Decay}
\label{sec:LocalIntegratedEnergyDecay}

Having obtained a bound for the degenerate energy associated to $T$, the first problem one would naturally try to address is that of the uniform boundedness of the \textit{non-degenerate} energy flux corresponding to a \textit{globally timelike} vector field $N$. This non-degenerate energy agrees with the energy as measured from a ``local observer's'' point of view. In fact, in \cite{md} it was shown that, for a wide class of non-extreme black holes, the redshift on $\hh$ together with the degenerate estimate of Section \ref{sec:ConservationOfDegenerateEnergy} suffices to yield such bounds without carrying out further analysis  of  dispersive properties. This method applies in particular to axisymmetric solutions on the general subextreme Kerr (see Corollary 7.2 of \cite{md}). However, in the extreme case the degeneracy of the redshift  (see Section \ref{sec:TheRedshiftEffect}) makes understanding of dispersion of $\psi$ essential even for the problem of boundedness. See also the discussion of Section 1.3 of \cite{aretakis1}.

Turning thus to dispersion, one of the main obstructions  is the so-called \textit{trapping effect}. Indeed, in Kerr spacetimes one can easily infer from a continuity argument the existence of a family of null geodesics which neither cross $\hh$ nor terminate at null infinity. Such null geodesics are called trapped and they are seen by the high frequency limit of solutions to the wave equation. In Schwarzschild, all trapped null geodesics approach the hypersurface $r=3M$  known as the \textit{photon sphere}, and dispersion in such backgrounds was proven by employing energy currents associated to vector fields which vanish precisely on this hypersurface.
 
Passing from Schwarzschild to Kerr for $|a|\neq 0$, the structure of trapped null geodesics becomes more complicated. Indeed, in the case of $|a|\neq 0$, there are null geodesics with constant $r$ for an open range of Boyer-Lindquist $r$ values. In fact,  Alinhac \cite{ali0} explicitly showed that classical energy currents cannot yield non-negative definite spacetime estimates for \textit{general} solutions to the wave equation (see, however, the discussion below for the axisymmetric case). Independent recent works have overcome this difficulty based on the separability of the wave equation \cite{mikraa, megalaa}, the complete integrability of the geodesic flow and the use of pseudodifferential calculus \cite{tataru2} and the existence of a non-trivial Killing tensor \cite{bluekerr}. The equivalence, for Ricci flat spacetimes, of these three geometric properties was shown in \cite{carter}. 

We adapt the method used for the first time in \cite{mikraa} (see also \cite{md}). The main insight of \cite{mikraa} is that, although the classical energy method cannot be directly applied to a general solution for obtaining non-negative definite estimates, it could well be the case that the energy method can be used in a more sophisticated form  to individual modes; the separability of the wave equation, as is used in \cite{mikraa}, provides the means for considering such modes (see Section \ref{sec:SeparabilityOfTheWaveEquation}). This approach, though having the disadvantage of taking the Fourier transform, has the advantage that it allows for a clean way to deal with all frequency ranges emphasising  the relevant geometric features, and in particular, it does not require ``fine-tunning'' parameters in the sense of previous delicate constructions for Schwarzschild. However, application of the virial frequency-localised currents constructed in \cite{mikraa} gives rise to error terms that can only be bounded using the redshift effect, and thus, they cannot be readily adapted to extreme Kerr.

In this paper, we construct novel microlocal currents which allow us to completely decouple the integrated local energy decay from the redshift. In other words, we show that the 4-integral of an energy-type density (which degenerates at $\hh$) is bounded by the conserved flux of the `stationary' Killing field $T$ through $\Sigma_{0}$ (see Theorem \ref{t2} of Section \ref{sec:TheMainTheorem}). In fact, we show that it suffices to use this microlocalisation only in a spatially compact region located away from $\hh$.  Note that in order to consider individual modes of a general solution, we need to take the Fourier transform in time, and since, a priori, the solutions might not be $L^{2}(dt)$, we need to cut off in time. The cut-off will create error terms that we control using  auxilliary microlocal currents and the introduction of novel classical vector fields (see Section \ref{sec:FourierLocalisedEstimates}).

As remarked above, trapping affects only the high frequency limit of $\psi$. However, the situation is much more favourable for  axisymmetric solutions $\psi$ to the wave equation. In this case, for the entire range $|a|\leq M$, the behaviour of high frequencies is intimately tied with the structure of the trapped null geodesics  which approach a \textit{unique} hypersurface $r=z_{a,M}$, for a constant $z_{a,M}$ that depends only\footnote{In fact, $z_{a,M}$ is the unique root of $s(r)=r^{3}-3Mr+a^{2}r+a^{2}M$ in the domain of outer communications.} on $a,M$. Geometrically, this is reflected in the fact that trapped null geodesics orthogonal to the axial Killing vector field $\Phi$ must necessarily approach the hypersurface $r=z_{a,M}$. For this reason, we refer to this hypersurface as the ``effective photon sphere''. Note that for extreme Kerr, the effective photon sphere corresponds to $r=(1+\sqrt{2})M$. The trapping can then be quantified from the fact that the integrated energy decay estimate degenerates on this hypersurface. The degeneracy is eliminated by commuting with $T$ (for more details about the trapping on extreme Kerr see the discussions in Sections \ref{sec:PropertiesOfThePotentialV} and \ref{sec:TheTrappedFrequenciesMathcalF21}).

In view of the above discussion, let us explicitly note that the obstruction uncovered by Alinhac \cite{ali0} does not apply to the axisymmetric case, and thus one could in principle expect to derive integrated decay (for $|a|\leq M$) using purely classical currents; this remains, however, an open problem.   Nonetheless, the separability turns out to be extremely useful in view also of its systematic approach to low frequencies as discussed above.

\subsubsection{Uniform Boundedness of Energy}
\label{sec:UniformBoundednessOfEnergy}

Using the above integrated local energy decay and a novel current which captures in a quantitative way the degenerate redshift close to $\hh$, we infer boundedness of the non-degenerate energy. At the same time, this yields the boundedness of the 4-integral of an energy-type density integrated over a neighbourhood of $\hh$ (see Theorem \ref{t1}). This integrated estimate, however,  degenerates with respect to the transversal to $\hh$ derivative.  This degeneracy is related to the \textit{trapping} along $\hh$ and is a characteristic feature of degenerate horizons, first presented in \cite{aretakis1}.

\subsubsection{Energy and Pointwise Decay}
\label{sec:EnergyAndPointwiseDecay}

We obtain now decay of the degenerate energy (see Theorem \ref{t3}) using (a) the integrated local energy decay and the uniform boundedness of energy, (b) an adaptation of the Dafermos-Rodnianski method \cite{new} and (c) the existence of an appropriate causal vector field $P$. We note that in our case, the assumptions of the Dafermos-Rodnianski method are not strictly satisfied in view of the degeneracy at $\mathcal{H}^{+}$. An extension of this method  which covers extreme black holes was presented in \cite{aretakis2} where a virial causal vector field $P$  was introduced. This vector field was used to derive a hierarchy of estimates in a neighbourhood of $\hh$ that parallel the hierarchy of \cite{new} near $\mathcal{I}^{+}$. In Section \ref{sec:Energydecay}, we show that the analogue of $P$ can be constructed in extreme Kerr, and we employ this vector field to prove energy decay.

To obtain pointwise estimates, one needs to derive estimates for non-degenerate higher order energies and then apply appropriate Sobolev inequalities on $\Sigma_{\tau}$. Bounds on such higher order energies were first derived in \cite{dr7} by commuting with suitably chosen timelike vector fields (essentially capturing the higher order redshift effect).  In view, however, of the degeneracy of the redshift in extreme Kerr, one cannot commute with timelike vector fields on $\hh$. In fact, the results of \cite{aretakis2} suggest that non-degenerate higher order energies on $\Sigma_{\tau}$ generically blow-up!

Nonetheless, by an interpolation argument, we prove non-degenerate $L^{2}$ bounds for $\psi$ on the spheres $\mathbb{S}^{2}(r)$. These can  be used to derive non-degenerate bounds for higher order energies controlling the derivatives of $\psi$ which are tangential to $\mathbb{S}^{2}(r)$. Indeed, although commutation with the Killing vector fields $T,\Phi$ is not enough for controlling all derivatives tangential to the sphere,  Kerr possesses a third differential operator $Q$ which commutes with $\Box_{g}$, which unlike $T$ and $\Phi$, is of second order (see Section \ref{sec:TheCarterOperatorAndSymmetries}).  This operator can then be used to bound an elliptic operator on the spheres, and by  a spherical Sobolev  embedding we infer the required pointwise results (see Theorems \ref{t4}, \ref{t5}). This technique was introduced by Andersson and Blue \cite{bluekerr} in the $|a|\ll M$ case and requires higher regularity than the method of \cite{dr7} described above. The fact that  this loss of regularity is necessary in extreme Kerr reveals another characteristic feature of extreme black holes. See Section \ref{sec:PointwiseEstimates}.

\subsection{The Non-Axisymmetric Case}
\label{sec:mnot0}
We shall briefly describe several additional issues concerning the non-axisymmetric case which are not present in the axisymmetric case considered in this paper.

The main  features which emerge in the non-axisymmetric case are the problems of superradiance  (see Sections \ref{sec:ConservationOfDegenerateEnergy}, \ref{sec:AxisymmetryVsSuperradiance}) and the more complicated trapping. In  particular, there are trapped null geodesics which never leave the ergoregion, and thus the previous two difficulties are in some sense coupled, at least in physical space. Regarding the subextreme case, one of the main insights of \cite{megalaa} is that  for all $|a|<M$, the superradiant frequencies\footnote{The superradiant frequency range corresponds to $0\leq m\omega< \omega_{+}m^{2}$, where $\omega_{+}=\frac{a}{2M(M+\sqrt{M^{2}-a^{2}})}$ is the angular velocity of the event horizon (for the definition of the frequencies $\omega,m$ see Section \ref{sec:SeparabilityOfTheWaveEquation}).}  are not trapped; thus these difficulties uncouple when viewed with respect to this microlocalisation. In the extreme case, however,  the upper limit of the superradiant frequencies is in some sense marginally trapped. This may be related to the existence of ``essentially undamped'' quasinormal modes located in this frequency regime. Heuristics based on the existence of such modes and numerical analysis of  Andersson and Glampedakis  \cite{mhighinsta} suggest that solutions $\psi$ are subject to an `instability' which forces them to decay much more slowly. Proving, however, mathematically this `instability' remains an open problem and we hope that the methods of the present paper in conjunction with those of \cite{megalaa} will be useful towards this direction.   We remark that since superradiance is absent in extreme Reissner-Nordstr\"{o}m (as in the case considered in the present paper),  this `instability' is not present in such backgrounds. 

\subsection{The Uniqueness Problem and Extremality}
\label{uni}

We end this introduction with a brief discussion of a related problem, namely the uniqueness of Kerr black holes. 

The ``no-hair'' conjecture states that the domains of outer communication of smooth, stationary, four dimensional, vacuum, connected black hole solutions are isometrically diffeomorphic to those of the Kerr family of black holes.  It is a well-known result (see \cite{cuni,c1,haw}) that if the event horizon is \textit{non-degenerate} and if the metric in the domain of outer communications is \textit{real analytic}, then this conjecture holds. Regarding the degenerate case, Chru\'{s}ciel and Nguyen \cite{chru} showed that the domains of outer communication of analytic, stationary, electrovacuum spacetimes with connected, non-empty, rotating, degenerate future event horizons are isometrically diffeomorphic to the domain of outer communications of \textit{extreme} Kerr-Newman black holes. 

The assumption of analyticity is quite restrictive since, a priori, there is no reason that  general stationary solutions to the Einstein-vacuum equations be analytic in the ergoregion\footnote{Note that, in view of standard elliptic theory, stationary solutions are indeed analytic in the exterior of the ergoregion.}. The program to remove the analyticity assumption was initiated in \cite{uni1} by Ionescu and Klainerman.  Subsequently, the authors of \cite{alexakisiokl} showed how to bypass the analyticity assumption in the case where the stationary vacuum spacetime is a small perturbation of a given \textit{subextreme} Kerr background.

The works on removing the analyticity assumption require in a fundamental way the event horizons to be bifurcate,  i.e.~the future event horizon $\hh$ and the past event horizon $\mathcal{H}^{-}$ must intersect on a 2-surface $S$ homeomorphic to the sphere. Indeed, the existence of this sphere is crucial in the arguments for the unique continuation and extension of the Hawking Killing vector field in a neighbourhood of the event horizon. Moreover, Ionescu and Klainerman \cite{uni2} have very recently constructed local, stationary, vacuum extensions of a subextreme Kerr solution in a future neighbourhood of a point $p\in\mathcal{H}^{-}$ (with $p$ not on the bifurcation sphere) which admit \textit{no} extension of the associated Hawking vector field, emphasising, in particular,  the importance of the bifurcation sphere in the context of the uniqueness problem. Recall that in extreme Kerr there is no bifurcation sphere, and hence, the uniqueness of extreme Kerr without analyticity may be a very challenging problem.

\section{Geometry of Extreme Kerr}
\label{sec:GeometryOfExtremeKerr}

For the convenience of the reader, we briefly recall the geometric features of extreme Kerr that are relevant to the considerations of this paper. In Section \ref{sec:TheMetric} we introduce the metric in Boyer-Lindquist and Eddington-Finkelstein coordinates and then, in Section \ref{sec:TheBlackHoleRegion}, we define the differential structure of a manifold $\mathcal{N}$ on which the metric with respect to the latter coordinates is regular. The wave equation will be considered in the region $\mathcal{R}\subset\mathcal{N}$ defined in Section \ref{sec:TheInitialHypersurfaceSigma0}.

\subsection{The Metric}
\label{sec:TheMetric}

The  Kerr metric with respect to the \textit{Boyer-Lindquist coordinates} $(t, r,\theta, \phi)$ is given by
\begin{equation*}
g=g_{tt}dt^{2}+g_{rr}dr^{2}+g_{\phi\phi}d\phi^{2}+g_{\theta\theta}d\theta^{2}+2g_{t\phi}dtd\phi,
\end{equation*}
where
\begin{equation*}
\begin{split}
&g_{tt}=-\frac{\Delta-a^{2}\sin^{2}\theta}{\rho^{2}}, \ \ \  g_{rr}=\frac{\rho^{2}}{\Delta},\  \ \   g_{t\phi}=-\frac{2Mar\sin^{2}\theta}{\rho^{2}},\\  
&\ \ \ \  \ \   g_{\phi\phi}=\frac{(r^{2}+a^{2})^{2}-a^{2}\Delta\sin^{2}\theta}{\rho^{2}}\sin^{2}\theta, \ \ \   g_{\theta\theta}=\rho^{2}
\end{split}
\end{equation*}
with
\begin{equation}
\Delta=r^{2}-2Mr+a^{2}, \ \ \ \ \ \rho^{2}=r^{2}+a^{2}\cos^{2}\theta.
\label{basic}
\end{equation}
Schwarzschild corresponds to the case $a=0$, subextreme Kerr to $|a|<M$ and extreme Kerr to $|a|=M$.

Note that the metric component $g_{rr}$ is singular precisely at the points where $\Delta=0$. To overcome this coordinate singularity  we introduce the following functions $r^{*}(r), \phi^{*}(\phi, r)$ and $v(t, r^{*})$ such that
\begin{equation*}
r^{*}=\int\frac{r^{2}+a^{2}}{\Delta}, \ \ \ \phi^{*}=\phi +\int\frac{a}{\Delta}, \ \ \ v=t+r^{*}
\end{equation*}
In the \textit{ingoing Eddington-Finkelstein  coordinates} $(v,r,\theta, \phi^{*})$  the metric takes the form
\begin{equation*}
g=g_{vv}dv^{2}+g_{rr}dr^{2}+g_{\phi^{*}\phi^{*}}(d\phi^{*})^{2}+g_{\theta\theta}d\theta^{2}+2g_{vr}dvdr+2g_{v\phi^{*}}dvd\phi^{*}+2g_{r\phi^{*}}drd\phi^{*},
\end{equation*}
where
\begin{equation}
\begin{split}
&g_{vv}=-\left(1-\frac{2Mr}{\rho^{2}}\right), \ \ \   g_{rr}=0,\  \ \  g_{\phi^{*}\phi^{*}}=g_{\phi\phi}, \  \ \  g_{\theta\theta}=\rho^{2}\\
\\& \ \ \ \  \ \  g_{vr}=1, \ \ \   g_{v\phi^{*}}=-\frac{2Mar\sin^{2}\theta}{\rho^{2}}, \  \ \  g_{r\phi^{*}}=-a\sin^{2}\theta.
\end{split}
\label{edi}
\end{equation}
For completeness, we include the computation for the inverse of the metric in $(v,r,\theta,\phi^{*})$ coordinates:
\begin{equation*}
\begin{split}
&g^{vv}=\frac{a^{2}\sin^{2}\theta}{\rho^{2}}, \ \ \   g^{rr}=\frac{\Delta}{\rho^{2}},\ \ \   g^{\phi^{*}\phi^{*}}=\frac{1}{\rho^{2}\sin^{2}\theta}, \  \ \  g^{\theta\theta}=\frac{1}{\rho^{2}}\\
\\& \ \ \ \ \ \ \  \ \   g^{vr}=\frac{r^{2}+a^{2}}{\rho^{2}} ,\  \ \   g^{v\phi^{*}}=\frac{a}{\rho^{2}}, \ \ \   g^{r\phi^{*}}=\frac{a}{\rho^{2}}.
\end{split}
\end{equation*}

\subsection{The Differential Structure}
\label{sec:TheBlackHoleRegion}

Clearly, the metric expression \eqref{edi} does not break down at the points where $\Delta=0$, and in fact, it turns out (see \cite{haw}) that this expression is regular even for $r<0$. On the other hand, the curvature would blow-up at $\rho^{2}=0$, i.e.~the equatorial points of $r=0$.  This motivates the following definition of the underlying differential structure of the Kerr spacetime.

Let $(\theta,\phi^{*})$ represent standard global\footnote{modulo the standard degeneration at $\theta=0,\pi$...} spherical coordinates on the sphere $\mathbb{S}^{2}$ and $S_{\text{eq}}$ denote the equator, i.e.~$S_{\text{eq}}=\mathbb{S}^{2}\cap\left\{\theta=\pi/2\right\}$. Let also $(v,r)$ be a global coordinate system on $\mathbb{R}\times\mathbb{R}$. We define the differential structure of the manifold $\mathcal{N}$ to be
\begin{equation*}
\mathcal{N}=\Bigg\{\big(v,r,\theta,\phi^{*}\big)\in\bigg\{\Big\{\mathbb{R}\times\mathbb{R}\times\mathbb{S}^{2}\Big\}\setminus\Big\{\mathbb{R}\times\left\{0\right\}\times S_{\text{eq}}\Big\}\bigg\}\Bigg\}.
\end{equation*}
On this manifold, given now parameters $|a|\leq M$, the expression \eqref{edi} defines a regular Lorentzian metric.

From now on, we restrict our attention to extreme Kerr $|a|=M$, unless otherwise stated.  The event horizon $\hh$ is defined by
\begin{equation*}
\mathcal{H}^{+}=\mathcal{N}\cap\left\{r=M\right\}.
\end{equation*}
The \textit{black hole region} $\mathcal{N}_{\text{BH}}$ corresponds to
\begin{equation*}
\mathcal{N}_{\text{BH}}=\mathcal{N}\cap \left\{r<M\right\};
\end{equation*}
it is characterised by the fact that observers in the black hole region cannot send signals to observers located at points with $r>M$. The exterior region $\mathcal{D}$ given by
\begin{equation*}
\mathcal{D}=\mathcal{N}\cap \left\{r>M\right\}
\end{equation*}
is the so-called \textit{domain of outer communications}. This is precisely the region covered by  the Boyer-Lindquist coordinates.  Note that we shall be interested in studying the solutions to the wave equation in the region $\mathcal{D}\cup\hh$.

\subsection{The Penrose Diagram}
\label{sec:ThePenroseDiagram}

A convenient graphic representation of the previously mentioned regions can be achieved using Penrose diagrams (for relevant definitions see \cite{haw}).  The Penrose diagram\footnote{The Penrose diagrams of the Kerr-Newman family are treated in detail in \cite{carter2}.} of the extended region $\mathcal{N}$  is
 \begin{figure}[H]
	\centering
		\includegraphics[scale=0.135]{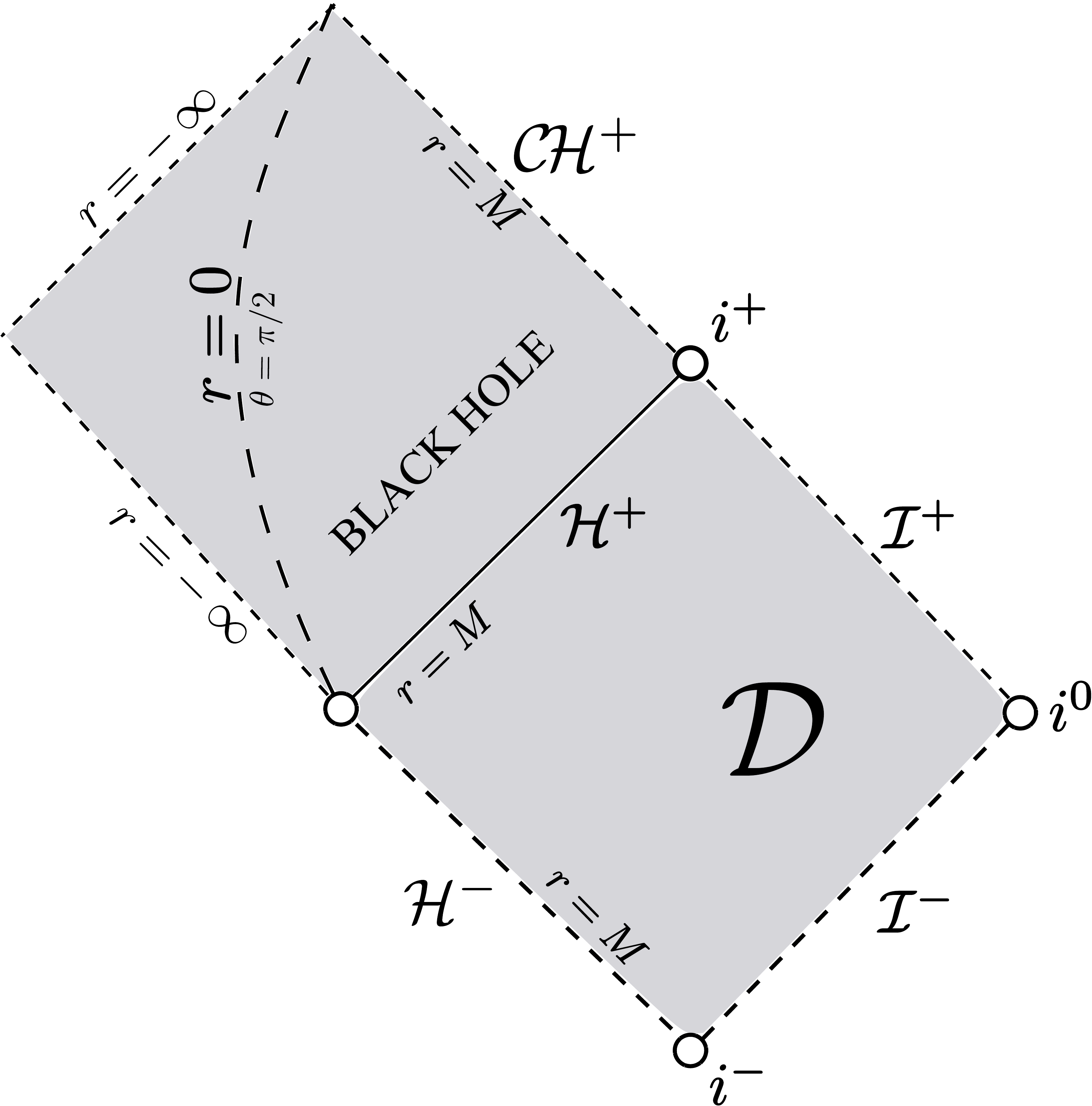}
	\label{fig:ekerr2}
\end{figure}
This diagram represents pictorially the following fact: with respect to suitable definitions for the asymptotic structure corresponding to future (past) null infinity $\mathcal{I}^{+}$  $(\mathcal{I}^{-})$, the region $\mathcal{D}$ can be characterised by
\begin{equation*}
\mathcal{D}=J^{-}(\mathcal{I}^{+})\cap J^{+}(\mathcal{I}^{-}).
\end{equation*} 
We will not rely on such constructions in this paper, but we will often depict spacetime regions in the above form.

\subsection{Useful Computations in Kerr}
\label{sec:UsefulComputationsInKerr}

We denote by $T= \partial_{v}, Y=\partial_{r}$ and $\Phi=\partial_{\phi^{*}}$ the coordinate vector fields with respect to the system $(v, r, \theta, \phi^{*})$. Then $T=\partial_{t}$ and $\Phi=\partial_{\phi}$, where the coordinate vector fields are taken with respect to the Boyer-Lindquist system. These two vector fields are manifestly Killing. Note that we may sometimes use the system $(t, r^{*}, \theta, \phi)$ in the region $r\geq r_{e}$ for some fixed $r_{e}>M$. From now on the coordinate vector field $\partial_{r^{*}}$ refers exclusively to the system $(t, r^{*},\theta,\phi)$.

The determinant of the metric with respect to the system $(t,r,\theta,\phi)$ and $(v,r,\theta,\phi^{*})$ is $\text{det}(g)=-\rho^{4}\sin^{2}\theta$. Therefore, one can easily show that the volume form $dg_{\text{vol}}$ takes the following form in the various systems:
\begin{equation*}
\begin{split}
&(t,r,\theta,\phi):\  dg_{\text{vol}}=\rho^{2}\sin\theta\, dt\, dr\, d\theta\, d\phi,\\
&(t,r^{*},\theta,\phi):\  dg_{\text{vol}}=\rho^{2}\frac{\Delta}{r^{2}+M^{2}}\sin\theta\, dt\, dr^{*}\, d\theta\, d\phi,\\
&(t,r^{*},\theta,\phi^{*}):\  dg_{\text{vol}}=\rho^{2}\frac{\Delta}{r^{2}+M^{2}}\sin\theta\, dt\, dr^{*}\, d\theta\, d\phi^{*},\\
&(v,r,\theta,\phi^{*}):\  dg_{\text{vol}}=\rho^{2}\sin\theta\, dv\, dr\, d\theta\, d\phi^{*}.
\end{split}
\end{equation*}
It is important to remark that for $r\geq r_{e}$ (where $r_{e}>M$) we have $dg_{\text{vol}}=\nu dtdr^{*}d\theta d\phi$ where $\nu\sim r^{2}$ and $\sim$ depends only on $r_{e},a, M$. 

The wave operator in $(v,r,\theta, \phi^{*})$ coordinates is (recall that $\partial_{v}=T$, $\partial_{r}=Y$ and $\partial_{\phi^{*}}=\Phi$):
\begin{equation*}
\begin{split}
\Box_{g}\psi=&\frac{a^{2}}{\rho^{2}}\sin^{2}\theta\left(TT\psi\right)+\frac{2(r^{2}+M^{2})}{\rho^{2}}\left(TY\psi\right)+\frac{\Delta}{\rho^{2}}(YY\psi)\\&+\frac{2a^{2}}{\rho^{2}}(T\Phi\psi)+\frac{2a}{\rho^{2}}(Y\Phi\psi)+\frac{2r}{\rho^{2}}(T\psi)+\frac{\Delta'}{\rho^{2}}(Y\psi) +\frac{1}{\rho^{2}}\lapp_{(\theta,\phi^{*})}\psi,
\end{split}
\end{equation*}
where $\lapp_{(\theta,\phi^{*})}\psi=\frac{1}{\sin\theta}\left(\partial_{\theta}\left[\sin\theta\cdot\partial_{\theta}\psi\right]\right)+\frac{1}{\sin^{2}\theta}\partial_{\phi^{*}}\partial_{\phi^{*}}\psi$ denotes the standard Laplacian on $\mathbb{S}^{2}$ with respect to $(\theta,\phi^{*})$. 

Regarding the tortoise coordinate $r^{*}$, in extreme Kerr  we have 

\begin{equation*}
\begin{split}
\frac{dr^{*}}{dr}=\frac{r^{2}+M^{2}}{(r-M)^{2}}=1+\frac{2Mr}{(r-M)^{2}}=1+M\frac{2(r-M)}{(r-M)^{2}}+\frac{2M^{2}}{(r-M)^{2}},
\end{split}
\end{equation*}
and therefore, 
\begin{equation}
\begin{split}
r^{*}(r)=(r-M)+2M\log(r-M)-\frac{2M^{2}}{r-M}-2M\log(\sqrt{2}M).
\end{split}
\label{rstar}
\end{equation}
Clearly, we have $r^{*}(M+\sqrt{2}M)=0$. Note that for large $r$ we have $r\leq r^{*}\leq 2r$. The coordinate $r^{*}$ will be used throughout the paper. There will be several times where we will define a function with respect to $r^{*}$ but depict its graph with respect to the $r$ variable. Note also that $r^{*}\rightarrow -\infty$ as $r\rightarrow M$ in an inverse linear way.

\subsection{The Foliation $\Sigma_{\tau}$}
\label{sec:TheInitialHypersurfaceSigma0}

Fix a sufficiently large constant $R$ (in particular, let $R>R_{e}$, where $R_{e}$ is as defined in Proposition \ref{larger}).  Let $H_{0}$ be a closed, connected, axisymmetric, spacelike hypersurface in $\big(\mathcal{D}\cup\hh\big)\cap\left\{r\leq R\right\}$ such that $\partial H_{0}=S_{1}\cup S_{2}$, where $S_{1}$ and $S_{2}$ are $(\theta, \phi^{*})$ spheres on $\hh$ and $\left\{r=R\right\}$, respectively. (The prototype for such a hypersurface is $\left\{t^{*}=0,M\leq r\leq R\right\}$, where $t^{*}$ is defined in Section 2.4 of \cite{mikraa}). Let also $\tilde{N}_{0}=\partial J^{+}(H_{0})\cap\mathcal{D}$. Then, we define the hypersurface $\Sigma_{0}=H_{0}\cup\tilde{N}_{0}$. Note that $\Sigma_{0}$  crosses the event horizon $\hh$ and ``terminates at a sphere on null infinity $\mathcal{I}^{+}$''.  
 \begin{figure}[H]
	\centering
		\includegraphics[scale=0.11]{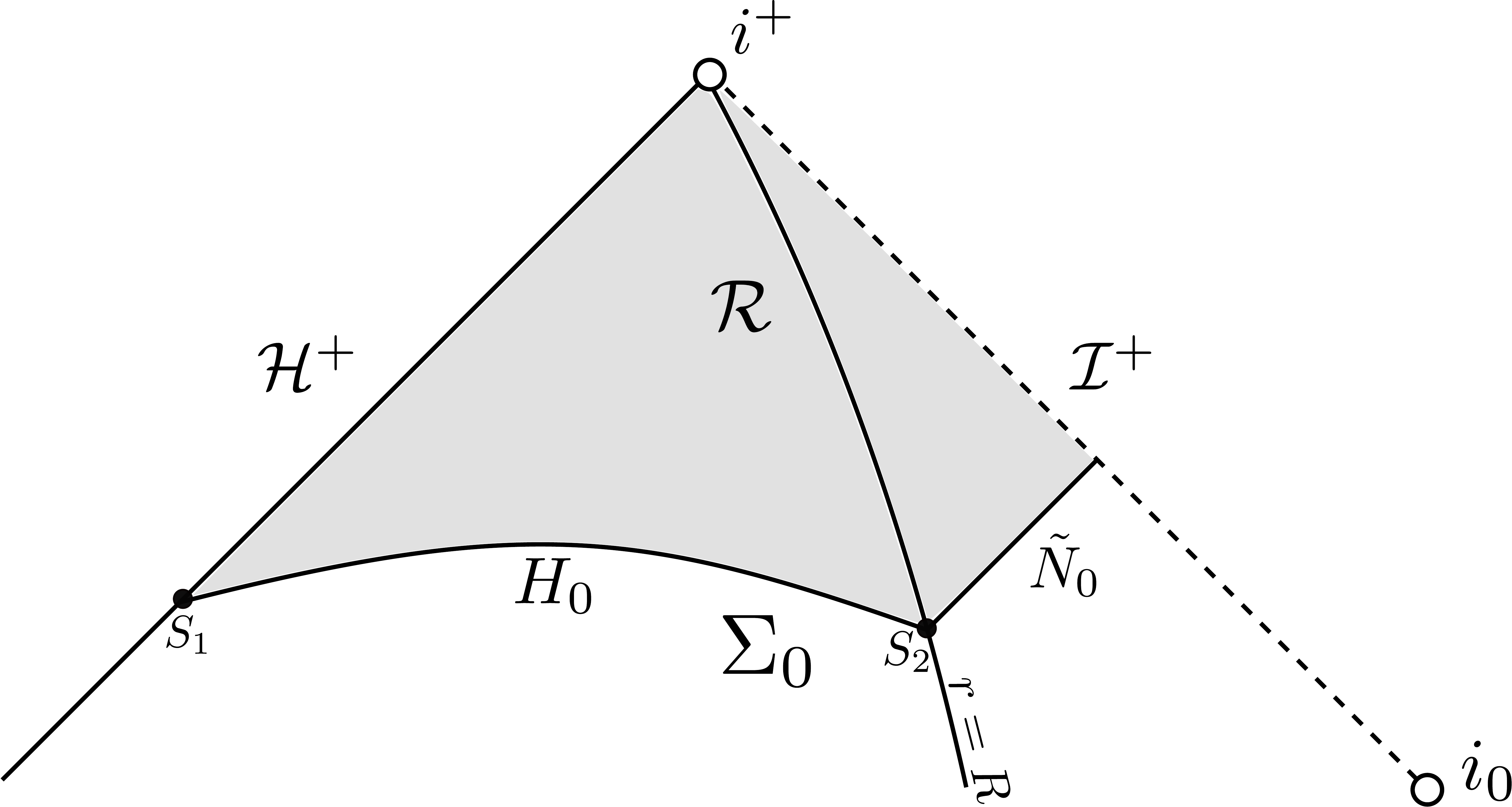}
	\label{fig:ernt1}
\end{figure}
Note that we will not explicitly use the exact formula that defines $\Sigma_{0}$  but rather the geometric properties of $\Sigma_{0}$ which we describe below. In fact, all we need is $\Sigma_{0}$ to be an admissible hypersurface of the second kind in the sense of \cite{mikraa}.

We define the region
\begin{equation*}
\mathcal{R}=J^{+}(\Sigma_{0})\cap\Big(\mathcal{D}\cup\hh\Big).
\end{equation*}
Note that $\mathcal{R}$ includes the event horizon $\hh$. We consider now the foliation $\Sigma_{\tau}=\varphi_{\tau}^{T}(\Sigma_{0}),\tau\geq 0$, where $\varphi_{\tau}^{T}$ is the flow of $T$. Since $T$ is Killing, the hypersurfaces $\Sigma_{\tau}$ are all isometric to $\Sigma_{0}$. We denote by $n_{\Sigma_{\tau}}$ the future directed (unit) vector field normal to $\Sigma_{\tau}$. We define the regions $\mathcal{R}(0,\tau)=\cup_{0\leq \tilde{\tau}\leq \tau}\Sigma_{\tilde{\tau}},\ \hh(0,\tau)=\hh\cap \mathcal{R}(0,\tau)$ and $\mathcal{I}^{+}(0,\tau)=\mathcal{I}^{+}\cap\mathcal{R}^{+}(0,\tau)$.

On $\Sigma_{\tau}$ we have an induced Lie propagated coordinate system $(p,\omega)$ such that $p\in[\left.\right.\!\! M,+\infty)$ and $\omega\in\mathbb{S}^{2}$. These coordinates are defined such that if $Z\in\Sigma_{\tau}$ and $Z=(v_{Z},r_{Z},\omega_{Z})$  then $p=r_{Z}$ and $\omega=\omega_{Z}$. There exist bounded functions $h_{i},i=1,2,3$ such that 
\begin{equation}
\partial_{p}=h_{1}T+h_{2}Y+h_{3}\Phi.
\label{eq:rho}
\end{equation}
The volume form of $\Sigma_{\tau}$ satisfies d$g_{\scriptstyle\Sigma_{\tau}}\sim r^{2}dp d\omega$. This observation allows us to use the same Hardy inequalities that first appeared in \cite{aretakis1} (see Section \ref{sec:HardyInequalities}). 
 \begin{figure}[H]
	\centering
		\includegraphics[scale=0.11]{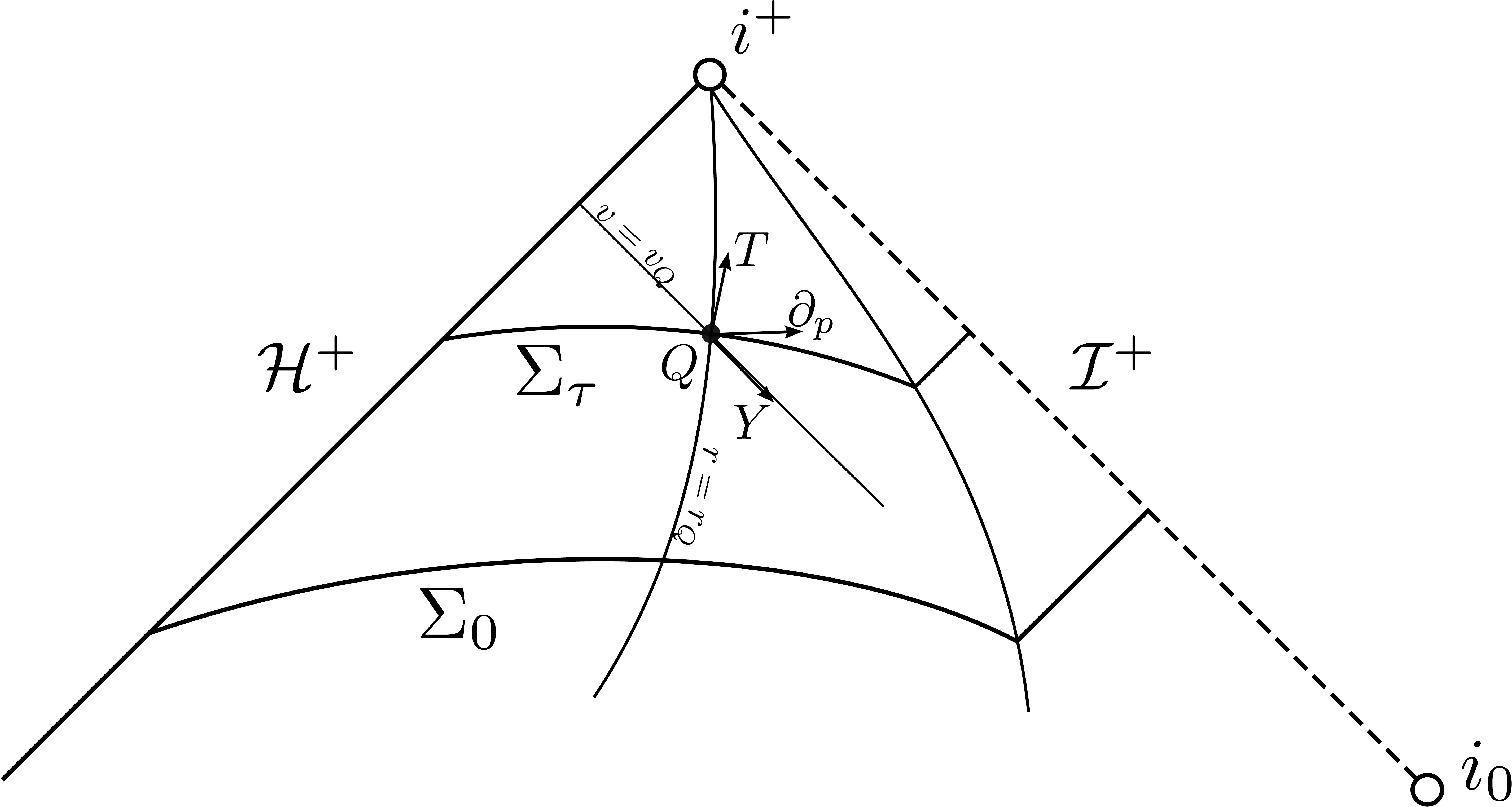}
	\label{fig:ernt2}
\end{figure}
The importance of this foliation lies in the fact that its leaves follow the waves to the future and thus capture the energy that is radiated away through $\mathcal{I}^{+}$.

\section{The Main Theorems}
\label{sec:TheMainTheorem}
We briefly explain all the  notations which are necessary for stating the main theorems.

The foliation $\Sigma_{\tau}$ and the associated tangent vector field $\partial_{p}$ and the region $\mathcal{R}$ are defined in Section \ref{sec:TheInitialHypersurfaceSigma0}.  We consider the Cauchy problem\footnote{For a well-posedness statement see Section \ref{sec:WellPosedness}.} for the wave equation in $\mathcal{R}$ with initial data prescribed on $\Sigma_{0}$.  

The coordinate systems $(t,r)$, $(t,r^{*})$ and $(v,r)$ are  described in Section \ref{sec:GeometryOfExtremeKerr}. We denote $T=\partial_{v}$, $Y=\partial_{r}$ and $\Phi=\partial_{\phi^{*}}$, where $\partial_{v},\partial_{r},\partial_{\phi^{*}}$ correspond to the system $(v,r,\theta,\phi^{*})$.   Note that $T,\Phi$ are Killing vector fields and $Y$ is transversal to $\mathcal{H}^{+}$. Let also $N$ be a $\varphi_{\tau}^{T}$-invariant future directed timelike vector field which coincides with $T$ far away from $\hh$.

We denote  $\left|\nabb\psi\right|^{2}=\frac{1}{r^{2}}\left|\nabb_{\mathbb{S}^{2}}\psi\right|^{2}$,  where $\left|\nabb_{\mathbb{S}^{2}}\psi\right|^{2}$ is the norm of the gradient of $\psi$ on the unit sphere $\mathbb{S}^{2}$ with respect the standard metric.

The current $J^{V}$ associated to the vector field $V$ is defined in Section \ref{sec:EnergyCurrents}. For reference, we remark that  
\begin{equation*}
J_{\mu}^{T}[\psi]n^{\mu}_{\Sigma}\sim \, (T\psi)^{2}+\left(1-\frac{M}{r}\right)^{2}(Y\psi)^{2}+\left|\nabb\psi\right|^{2},
\end{equation*}
which degenerates on $\mathcal{H}^{+}$, whereas
\begin{equation*}
J_{\mu}^{N}[\psi]n^{\mu}_{\Sigma}\sim \, (T\psi)^{2}+(Y\psi)^{2}+\left|\nabb\psi\right|^{2},
\end{equation*}
which does not degenerate on $\mathcal{H}^{+}$. Note that $\sim$ depends on $M$ and upper bound for $r$. 

The Carter operator $Q$ is defined in Section \ref{sec:TheCarterOperatorAndSymmetries}. We consider the following symmetry operators of up to second order of Kerr
\begin{equation*}
\mathbb{S}_{0}=\left\{id\right\}, \ \ \mathbb{S}_{1}=\left\{T, \Phi\right\}\ \ \mathbb{S}_{2}=\left\{T^{2}, \Phi^{2}, T\Phi, Q\right\},
\end{equation*}
and we denote $\left|S^{k}\psi\right|^{2}=\sum_{S\in \mathbb{S}_{k}}\left|S\psi\right|^{2}$.

 Finally, the initial data are assumed to be as in Section \ref{sec:WellPosedness} and moreover sufficiently regular  such that the right hand side of the estimates below are all finite.   All the integrals are considered with respect to the induced volume form. Then we have the following

\begin{mytheo}\textbf{\textbf{(Integrated Local Energy Decay)}} Let $\delta>0$ and $r_{e}>M$. There exists a constant $C_{\delta}$ which depends only on $M,r_{e}$ and $\delta$ such that for all axisymmetric solutions $\psi$ to the wave equation we have 
\begin{equation*}
\begin{split}
&\int_{\left\{r\geq r_{e}\right\}}{\left[\frac{1}{r^{3+\delta}}\psi^{2}+\frac{1}{r^{1+\delta}}(\partial_{r^{*}}\psi)^{2}+\frac{\left(r-(1+\sqrt{2})M\right)^{2}}{r^{3+\delta}}\left((T\psi)^{2}+\left|\nabb\psi\right|^{2}\right)\right]}\\&\ \ \ \ \ \ \ \ \ \ \ \ \ \ \ \ \ \ \ \ \ \ \ \leq C(r_{e},\delta)\int_{\Sigma_{0}}{J_{\mu}^{T}[\psi]n^{\mu}_{\Sigma_{0}}}.
\end{split}
\end{equation*}
\begin{remark}
The degeneracy on the hypersurface $\left\{r=(1+\sqrt{2})M\right\}$ can be removed by commuting with the Killing field $T$. This commutation is related to the trapping on this hypersurface. See also the discussion in Section \ref{sec:LocalIntegratedEnergyDecay}. 
\label{r1}
\end{remark}
\label{t2}
\end{mytheo}

\begin{mytheo}\textbf{(Uniform Boundedness of Non-Degenerate Energy)} There exists a constant $C$ which depends only on $M$ such that for all axisymmetric  solutions $\psi$ to the wave equation we have 
\begin{equation*}
\begin{split}
\int_{\Sigma_{\tau}}{J_{\mu}^{N}[\psi]n^{\mu}_{\Sigma_{\tau}}}+\int_{\mathcal{H}^{+}}{J_{\mu}^{N}[\psi]n^{\mu}_{\mathcal{H}^{+}}}+&\int_{\left\{r\leq \frac{23}{21}M\right\}}{\Big[\psi^{2}+(T\psi)^{2}+(r-M)(Y\psi)^{2}+\left|\nabb\psi\right|^{2}\Big]}\\& \leq C\int_{\Sigma_{0}}{J_{\mu}^{N}[\psi]n^{\mu}_{\Sigma_{0}}}.
\end{split}
\end{equation*}
\label{t1}
\end{mytheo}
\begin{remark}
We remark that the factor of $Y\psi$ degenerates at the event horizon. This degeneracy is related to the trapping phenomenon along the event horizon of extreme Kerr.  See Section \ref{sec:UniformBoundednessOfEnergy}.  Note also that $\frac{23}{21}M<(1+\sqrt{2})M$.
\label{r2}
\end{remark}
\begin{mytheo}\textbf{(Decay of Degenerate Energy)}
 Let
\begin{equation*}
\begin{split}
I^{T}_{\Sigma_{\tau}}[\psi]=&\int_{\Sigma_{\tau}}{J^{N}_{\mu}[\psi]n^{\mu}_{\Sigma_{\tau}}}+
\int_{\Sigma_{\tau}}{J^{T}_{\mu}[T\psi]n^{\mu}_{\Sigma_{\tau}}}+\int_{\Sigma_{\tau}}{r^{-1}\left(\partial_{p}(r\psi)\right)^{2}}
\end{split}
\end{equation*}
and
\begin{equation*}
E_{1}[\psi]=I^{T}_{\Sigma_{0}}[T\psi]+\displaystyle\int_{\Sigma_{0}}{J^{N}_{\mu}[\psi]n^{\mu}_{\Sigma_{0}}}+\displaystyle\int_{\Sigma_{0}}{\left(\partial_{p}(r\psi)\right)^{2}}.
\end{equation*}
Then, for all axisymmetric solutions $\psi$ of the wave equation we have 
\begin{equation*}
\displaystyle\int_{\Sigma_{\tau}}J^{T}_{\mu}[\psi]n_{\Sigma_{\tau}}^{\mu}\leq CE_{1}[\psi]\frac{1}{\tau^{2}}.
\end{equation*}
\label{t3}
\end{mytheo}

\begin{mytheo}\textbf{(Pointwise Boundedness)} There exists a constant $C$ which depends only on $M$  such that for all axisymmetric solutions $\psi$ to the wave equation we have 
\begin{equation*}
\begin{split}
\left|\psi\right|\leq C\cdot \sqrt{E_{2}[\psi]},
\end{split}
\end{equation*}
everywhere in $\mathcal{R}$, where 
\begin{equation*}
E_{2}[\psi]=\sum_{k\leq 2}\int_{\Sigma_{0}}J_{\mu}^{N}\big[S^{k}\psi\big]n^{\mu}_{\Sigma_{0}}.
\end{equation*}
\label{t4}
\end{mytheo}

\begin{mytheo}\textbf{(Pointwise Decay)} Fix $R>M$ and let $\tau\geq 1$. There exists a constant $C$ that depends on $M$ and $R$ such that 
\begin{itemize}
	\item For all axisymmetric solutions $\psi$ to the wave equation we have 
	\begin{equation*}
	\left|\psi\right|\leq C_{R}\sqrt{E_{3}[\psi]}\frac{1}{\sqrt{r}\cdot \tau},\ \ \ \  \left|\psi\right|\leq C_{R}\sqrt{E_{3}[\psi]}\frac{1}{r\cdot \sqrt{\tau}}
	\end{equation*}
	in $\left\{r\geq R\right\}$, where 
	\begin{equation*}
E_{3}[\psi]=\sum_{\left|k\right|\leq 2}{E_{1}\big[S^{k}\psi\big]}.
	\end{equation*}
	
	\item For all axisymmetric solutions $\psi$ to the wave equation we have 
	\begin{equation*}
	\left|\psi\right|\leq C\sqrt{E_{3}[\psi]}\frac{1}{\tau^{\frac{1}{2}}}
	\end{equation*}
	in $\left\{r\geq M\right\}$.	
\end{itemize}
\label{t5}
\end{mytheo}

\section{Preliminaries}
\label{sec:Preliminaries}

\subsection{Conventions}
\label{sec:Conventions}

We summarise some conventions that we follow in the paper. 

 We use $b$ and $C$ to denote  potentially small and large constants, respectively, which depend  only on $M$.  Some (undetermined) constants may depend on other undetermined constants (say $\omega_{1}$) and that will be stated explicitly at each instance by writing, for example,  $C(\omega_{1})$. Note that all constants will eventually depend only on $M$.

 All the integrals should be considered with respect to the induced volume form, unless the measure of integration is explicitly stated.  Regarding the event horizon, we choose the normal to be $n_{\hh}=T+\frac{1}{2M}\Phi$ and thus the volume form is taken respectively.
 
 The functions $\Delta,\rho$ are defined by \eqref{basic}. The spherical gradient will be involved in various computations. By convention, we denote by $\left|\nabb_{\mathbb{S}^{2}}\psi\right|^{2}$ the norm of the gradient of $\psi$ on the unit sphere $\mathbb{S}^{2}$ with respect to the standard metric $g_{\mathbb{S}^{2}}$. Moreover, we denote by $\left|\nabb\psi\right|^{2}$ and $\big|\nabb_{g}\psi\big|^{2}$ the norm of the gradient on $\mathbb{S}^{2}(r)$ equipped with the standard metric and the (induced from the) Kerr metric, respectively. In other words,
 \[  \left|\nabb\psi\right|^{2}=\frac{1}{r^{2}}\left|\nabb_{\mathbb{S}^{2}}\psi\right|^{2}, \ \ \ \ \ \ \ \  \big|\nabb_{g}\psi\big|^{2}=\frac{1}{\rho^{2}}\left|\nabb_{\mathbb{S}^{2}}\psi\right|^{2}.\]

Given any $r$-parameter such as $R$, we will denote by $R^{*}$ the value $r^{*}(R)$. We will often refer to $r^{*}$-ranges by their corresponding $r$-ranges. For example, we will describe using the coordinate $r$ the region over which we integrate and at the same time the variable of integration will be $r^{*}$. Moreover, $f'$ always denotes differentiation with respect to $r^{*}$, i.e.~$f'=\frac{df}{dr^{*}}$.

Finally, for simplicitly, we will  write $\hh$ instead of $\hh\cap\mathcal{R}$ etc. unless otherwise stated.

\subsection{The Cauchy Problem; Well-posedness}
\label{sec:WellPosedness}

We consider solutions of the Cauchy problem for the wave equation \eqref{wave} with axisymmetric initial data 
\begin{equation}
\left.\psi\right|_{\Sigma_{0}}=\psi_{0}\in H^{k}_{\operatorname{loc}}\left(\Sigma_{0}\right), \left.n_{\Sigma_{0}}\psi\right|_{\Sigma_{0}}=\psi_{1}\in H^{k-1}_{\operatorname{loc}}\left(\Sigma_{0}\right),
\label{cd}
\end{equation}
where $k\geq 1$ and the hypersurface $\Sigma_{0}$ is as defined in Section \ref{sec:TheInitialHypersurfaceSigma0} and $n_{\Sigma_{0}}$ denotes the future unit normal of $\Sigma_{0}$.  In view of the global hyperbolicity of $\mathcal{R}$, there exists a unique solution to the above equation in $\mathcal{R}$ such that  for any spacelike hypersurface  $S$,
\begin{equation*}
\left.\psi\right|_{S}\in H^{k}_{\operatorname{loc}}\left(S\right), \left.n_{S}\psi\right|_{S}\in H^{k-1}_{\operatorname{loc}}\left(S\right).
\end{equation*}
Moreover, solutions depend continuously on initial data.

In this paper, we will be interested in the case where $k\geq 2$ and  assume that 
$\lim_{x\rightarrow \mathcal{I}^{+}}r\psi^{2}(x)=0.$ For simplicity, from now on, \textbf{when we say ``for all axisymmetric solutions $\psi$ of the wave equation" we will assume that $\psi$ satisfies the above conditions}. Note that for obtaining sharp decay results, we will have to consider even higher regularity for $\psi$.

\subsection{Energy Currents}
\label{sec:EnergyCurrents}

The vector field method is a robust method for deriving $L^{2}$ estimates. One applies Stokes's theorem 
\begin{equation}
\int_{\Sigma_{0}}{P_{\mu}n^{\mu}_{\Sigma_{0}}}=\int_{\Sigma_{\tau}}{P_{\mu}n^{\mu}_{\Sigma_{\tau}}}+\int_{\mathcal{H}^{+}\left(0,\tau\right)}{P_{\mu}n^{\mu}_{\mathcal{H}^{+}}}+\int_{\mathcal{I}^{+}\left(0,\tau\right)}{P_{\mu}n^{\mu}_{\mathcal{I}^{+}}}+\int_{\mathcal{R}\left(0,\tau\right)}{\nabla^{\mu}P_{\mu}}
\label{sto}
\end{equation}
for suitable currents $P_{\mu}$, where all the integrals are with respect to the \textit{induced volume form} and the unit normals $n_{\scriptstyle\Sigma_{\tau}}$ are future directed. By convention, along the (null) event horizon $\hh$ we choose $n_{\hh}=T+\frac{1}{2M}\Phi$ (see also Section \ref{sec:TheRedshiftEffect}). An important class of currents are the so-called energy currents, which are produced by contracting the energy-momentum tensor
\begin{equation*}
\textbf{T}_{\mu\nu}\left[\psi\right]=\left(\partial_{\mu}\psi\right)\left(\partial_{\nu}\psi\right)-\frac{1}{2}g_{\mu\nu}(\partial^{\a}\psi)(\partial_{\a}\psi),
\end{equation*}
which, in case $\psi$ satisfies $\Box_{g}\psi=0$, is a symmetric divergence free (0,2) tensor. We will in fact consider this tensor for general functions $\psi:\mathcal{D}\rightarrow\mathbb{R}$ in which case we have $\text{Div}\textbf{T}\left[\psi\right]=\left(\Box_{g}\psi\right)d\psi.$  Given  a vector field $V$,  we define the $J^{V}$ current by $J^{V}_{\mu}[\psi]=\textbf{T}_{\mu\nu}[\psi]V^{\nu}.$ The divergence of this current is $\operatorname{Div}(J)=\operatorname{Div}\left(\textbf{T}\right)V+\textbf{T}\left(\nabla V\right)$. We also define the currents
\begin{equation*}
K^{V}\left[\psi\right]=\textbf{T}\left[\psi\right]\left(\nabla V\right),\ \ \ \ \ \mathcal{E}^{V}[\psi]=\text{Div}\left(\textbf{T}\right)V=\left(\Box_{g}\psi\right)\left(V\psi\right).
\end{equation*}
Note that from the symmetry of  $\textbf{T}$ we have
$K^{V}\left[\psi\right]=\textbf{T}_{\mu\nu}\left[\psi\right]\pi^{\mu\nu}_{V},$
where $\pi^{\mu\nu}_{V}=(\mathcal{L}_{V}g)^{\mu\nu}$ is the deformation tensor of $V$. Clearly if $\psi$ satisfies the wave equation then
$K^{V}[\psi]=\nabla^{\mu}J^{V}_{\mu}[\psi],$ which is an expression of the 1-jet of $\psi$.
Thus, if we use Killing vector fields as multipliers then the divergence vanishes and so we obtain a conservation law. This is partly the content of a deep theorem of Noether. For generalisations see \cite{christodoulou_actionprinciple}.

The following proposition is important and concerns the boundary terms that would arise from applying \eqref{sto} with $P_{\mu}=J_{\mu}^{V}$.
\begin{proposition}
Let $V_{1},V_{2}$ be two future directed timelike vectors. Then the quadratic expression $\textbf{T}\left(V_{1},V_{2}\right)$ is positive definite in $d\psi$. By continuity, if one of these vectors is null then $\textbf{T}\left(V_{1},V_{2}\right)$ is non-negative definite in $d\psi$.
\label{hypwave}
\end{proposition}

We finally remark that in Section \ref{sec:FourierLocalisedEstimates} we make more general use of the vector field method. See also Section \ref{sec:MicrolocalEnergyCurrents}. 

\subsection{Hardy Inequalities}
\label{sec:HardyInequalities}
The Hardy inequalities presented in \cite{aretakis1} for extreme Reissner-Nordstr\"{o}m spacetime immediately generalise to extreme Kerr. Recall that in extreme Kerr $\Delta=(r-M)^{2}$. For the convenience of the reader we state these inequalities below.

\begin{proposition}(\textbf{First Hardy Inequality}) For all functions $\psi$ which satisfy the regularity assumptions of Section \ref{sec:WellPosedness}  we have
\begin{equation*}
\int_{\Sigma_{\tau}}{\frac{1}{r^{2}}\psi^{2}}\leq C\int_{\Sigma_{\tau}}{\frac{\Delta}{r^{2}}[(T\psi)^{2}+(Y\psi)^{2}]},
\end{equation*}
where the constant  $C$ depends only on $M$ and $\Sigma_{0}$.
\label{firsthardy}
\end{proposition}

\begin{proposition}(\textbf{Second Hardy Inequality})  Let  $r_{0}>M$. Then for all functions $\psi$ which satisfy the regularity assumptions of Section \ref{sec:WellPosedness}   and for any positive number $\epsilon$ we have
\begin{equation*}
\int_{\mathcal{H}^{+}\cap\Sigma_{\tau}}{\psi^{2}}\leq \epsilon\int_{\Sigma_{\tau}\cap\left\{r\leq r_{0}\right\}}{\!\!(T\psi)^{2}+(Y\psi)^{2}}\ \,+C_{\epsilon}\int_{\Sigma_{\tau}\cap\left\{r\leq r_{0}\right\}}{\psi^{2}},
\end{equation*}
where the constant $C_{\epsilon}$ depends on $M$, $\epsilon$, $r_{0}$ and $\Sigma_{0}$. 
\label{secondhardy}
\end{proposition}

\begin{proposition} (\textbf{Third Hardy Inequality})  Let  $r_{0},r_{1}$ be such that  $M<r_{0}<r_{1}$. We define the regions
$\mathcal{A}=\mathcal{R}(0,\tau)\cap\left\{M\leq r\leq r_{0}\right\},\ \mathcal{B}=\mathcal{R}(0,\tau)\cap\left\{r_{0}\leq r\leq r_{1}\right\}.$ Then for all functions $\psi$ which satisfy the regularity assumptions of Section \ref{sec:WellPosedness}   we have
\begin{equation*}
\int_{\mathcal{A}}{\psi^{2}}\leq C\int_{\mathcal{B}}{\psi^{2}}\ \,+C\int_{\mathcal{A}\cup\mathcal{B}}{\frac{\Delta}{r^{2}}[(T\psi)^{2}+(Y\psi)^{2}]},
\end{equation*}
where the constant $C$ depends on $M$, $r_{0}$, $r_{1}$  and $\Sigma_{0}$.
\label{thirdhardy}
\end{proposition}

\section{Extreme Kerr and Linear Waves}
\label{sec:ExtremeKerrAndLinearWaves}

\subsection{Axisymmetry vs Superradiance}
\label{sec:AxisymmetryVsSuperradiance}

One of the main new features which emerges from passing from Schwarzschild to Kerr is that the `stationary' Killing field $T$ fails to be everywhere timelike  in the exterior region. In general, the region $\mathcal{E}\subset\mathcal{D}$ where $T$ is spacelike is called the \textit{ergoregion} and the boundary $\partial\mathcal{E}$ is called the \textit{ergosphere}. The ergoregion is well-known for enabling the extraction of energy out of a black hole. This ``process'' was discovered first by Penrose \cite{penroseergo} and remains the subject of intense research in the high energy physics community.

 In the context of the wave equation, if $T$ is spacelike then the energy flux $J_{\mu}^{T}[\psi]n^{\mu}$, where $n^{\mu}$ is timelike, fails to be non-negative definite and thus the associated conservation law does not imply uniform boundedness of the flux of $T$ through $\Sigma_{\tau}$. In particular, the energy that is radiated away through null infinity may be larger than the initial energy. This phenomenon is called \textit{superradiance} and is the main reason that obtaining any boundedness result (even away from $\hh$) is very difficult. In the last few years, this phenomenon was successfully treated in a series of papers for the general subextreme Kerr. See Section \ref{sec:PreviousWork} for relevant references. 

 Note that $\hh\subset\mathcal{E}$ and moreover $\text{sup}_{\mathcal{E}}r=2M$. However, the crucial observation is that \textbf{superradiance is absent if $\psi$ is assumed to be axisymmetric}.  Indeed, we have the following 
 
 \begin{proposition}
 Let $\Sigma$ be an axisymmetric spacelike hypersurface and $n_{\Sigma}$ be its future directed timelike unit normal. Let $\psi$ be an arbitrary axisymmetric function, i.e. such that $\Phi\psi=0$. Then
  \begin{equation*}
\begin{split}
 J_{\mu}^{T}[\psi]n^{\mu}_{\Sigma}\geq 0.
 \end{split}
\end{equation*}
 \end{proposition}
\begin{proof}
 We first show the following lemma
\begin{lemma}
 Let $n$ be a vector orthogonal to $\Phi$ and $\psi$ be an axisymmetric function. Then
  \begin{equation*}
\begin{split}
J_{\mu}^{\Phi}[\psi]n^{\mu}=0.
\end{split}
\end{equation*}
\end{lemma}
\begin{proof}
We have 
  \begin{equation*}
\begin{split}
J_{\mu}^{\Phi}[\psi]n^{\mu}=\textbf{T}_{\mu\nu}[\psi]\Phi^{\mu}n^{\nu}=(\Phi\psi)(\partial_{\nu}\psi)n^{\nu}-\frac{1}{2}g(\Phi, n)\left|\nabla\psi\right|^{2}=0.
\end{split}
\end{equation*}
\end{proof}
Then, since for every point $p$ there is a constant $\omega(p)$ such that $T+\omega(p)\Phi$ is future directed causal vector, from Proposition \ref{hypwave} we have
\begin{equation*}
\begin{split}
J_{\mu}^{T}[\psi]n^{\mu}_{\Sigma}=J_{\mu}^{T+\omega(p)\Phi}[\psi]n^{\mu}_{\Sigma}\geq 0.
\end{split}
\end{equation*}
\end{proof}
The assumption of axisymmetry for $\Sigma$ is in fact necessary, at least on the event horizon. Indeed, if $J_{\mu}^{T}[\psi]n^{\mu}_{\Sigma}\geq 0$  on $\mathcal{H}^{+}$, then, working in the $(v, r,\theta, \phi^{*})$ coordinates, if $n_{\Sigma}=n^{v}T+n^{r}Y+n^{\phi}\Phi$ then 
\begin{equation*}
\begin{split}
J_{\mu}^{T}[\psi]n^{\mu}_{\Sigma}=&(T\psi)^{2}n^{v}+(T\psi)(Y\psi)n^{r}+(T\psi)(\Phi\psi)n^{\phi}-\\
&-\frac{1}{2}\left(g(T,T)n^{v}+g(T,Y)n^{r}+g(T,\Phi)n^{\phi}\right)\left|\nabla\psi\right|^{2},
\end{split}
\end{equation*}
where $\left|\nabla\psi\right|^{2}=\frac{1}{\rho^{2}}\left[(a^{2}\sin^{2}\theta)(T\psi)^{2}+\Delta(Y\psi)^{2}+2(r^{2}+a^{2})(T\psi)(Y\psi)\right]+\left|\nabb\psi\right|^{2}$. Since the coefficient of $(Y\psi)^{2}$ vanishes on $\mathcal{H}^{+}$, for general positivity we must  have that the coefficient of $(T\psi)(Y\psi)$  vanishes there too. The vanishing of this coefficient implies the desired result. The above discussion applies for all $|a|\leq M$; by restricting to the extreme case we, in fact, obtain 
\begin{equation*}
J_{\mu}^{T}[\psi]n^{\mu}_{\Sigma_{0}}\sim (T\psi)^{2}+\left(1-\frac{M}{r}\right)^{2}(Y\psi)^{2}+\left|\nabb\psi\right|^{2}
\end{equation*}
for $r\leq R$ where the constants in $\sim$ depend  on $M$ and $R$.  Similarly, we obtain $\int_{\hh}J^{T}_{\mu}[\psi]n^{\mu}_{\hh}\geq 0$, and thus, since $K^{T}[\psi]=0$, we obtain the following 
\begin{proposition}
Let $\Sigma_{\tau}$ denote a foliation of spacelike axisymmetric hypersurfaces. If $\psi$ is also axisymmetric then
\begin{equation*}
\int_{\Sigma_{\tau}}J^{T}_{\mu}[\psi]n^{\mu}_{\Sigma_{\tau}}\leq \int_{\Sigma_{0}}J^{T}_{\mu}[\psi]n^{\mu}_{\Sigma_{0}}.
\end{equation*}
\label{tprop}
\end{proposition}

\subsection{The Degeneracy of the Redshift Effect}
\label{sec:TheRedshiftEffect}

Even though $T$ is not  everywhere null on $\hh$, the vector field 
\begin{equation*}
V=T+\frac{1}{2M}\Phi
\end{equation*}
is Killing and normal to the horizon\footnote{It is worth mentioning that its null conjugate with respect to the $(\theta,\phi^{*})$ foliation (of spheres) of the event horizon is $\bar{V}=Y+\left(\frac{\sin^{2}\theta}{4}\right)T+\left(\frac{3+\cos^{2}\theta}{8M}\right)\Phi.$}. In general, if there exists a Killing vector field $V$ which is normal to a null hupersurface then
\begin{equation*}
\nabla_{V}V=\kappa V
\end{equation*}
on the hypersurface. The quantity $\kappa$ is the so-called \textit{surface gravity}. For the general subextreme Kerr we have that $\kappa>0$ and depends only on $a,M$ (and in particular is constant on $\hh$). For the extreme Kerr we have $\kappa=0$ on $\hh$, and therefore, the integral curves of $V$ are affinely parametrised. This implies that the \textit{redshift} that takes place along $\hh$ degenerates in the extreme Kerr. In subextreme Kerr, this celebrated effect was first used in the $a=0$ case in \cite{dr3} and later in \cite{dr7, md, mikraa, megalaa}. In that case, the redshift effect implies the existence of a $\varphi_{\tau}^{T}$-invariant timelike future directed vector field $N$ such that 
\begin{equation}
K^{N}[\psi]\sim J_{\mu}^{N}[\psi]N^{\mu}
\label{red}
\end{equation}
on $\hh$. In Section \ref{sec:TheVectorFieldN} we show that not only is there no such vector field in extreme Kerr, but  there is in fact no $\varphi_{\tau}^{T}$-invariant timelike vector field $N$  satisfying the weaker condition $K^{N}[\psi]\geq 0$. However, we can still quantitatively capture the degenerate redshift close to $\hh$ by introducing a novel current. Such a current appeared for the first time in \cite{aretakis1}.

 We remark that in \cite{aretakis2}, it was shown  that the degeneracy of redshift on the event horizon of extreme Reissner-Nordstr\"{o}m gives rise to a hierarchy of conservation laws. These laws  were the source of all instability results for the solutions to the wave equation on such backgrounds; they, however, require fixing angular frequency, and thus, in view of the lack of spherical symmetry, they do not apply in the case of extreme Kerr.

\subsection{The Carter Operator and Hidden Symmetries }
\label{sec:TheCarterOperatorAndSymmetries}

In Schwarzschild, on top of the stationary Killing field $T$, one has a complete set of spherical  Killing fields traditionally denoted $\Omega_{1}, \Omega_{2}, \Omega_{3}$. These Killing fields provide enough conserved quantities to deduce that the Hamilton-Jacobi equations separate. In Kerr, although the only Killing vector fields are $T$ and  $\Phi$, the Hamilton-Jacobi equations still separate in view of a third non-trivial conserved quantity discovered by Carter \cite{carter2}. Penrose and Walker \cite{penrose} showed that the complete integrability of the geodesic flow has its fundamental origin in the existence of an \textit{irreducible Killing tensor}. 

A Killing 2-tensor is a symmetric 2-tensor $K$ which satisfies $\nabla_{\left(\alpha\right.}K_{\left.\beta\gamma\right)}$$=0$. For example, the metric is always a Killing tensor. A Killing tensor will be called \textit{irreducible} if it can not be constructed from the metric and other Killing vector fields. 

A Killing tensor $K$ yields a conserved quantity for geodesics. Indeed, if $\gamma$ is a geodesic then $K^{\alpha\beta}\overset{.}{\gamma}^{\alpha}\overset{.}{\gamma}^{\beta}$ is a constant of the motion. It turns out that  in Ricci flat spacetimes  the separability of Hamilton-Jacobi, the separability of the wave equation and the existence of an irreducible Killing tensor are equivalent. Moverover, a Killing tensor gives rise to a symmetry operator $K=\nabla_{\a}\left(K^{\a\beta}\nabla_{\beta}\right)$ with the property $\left[K, \Box_{g}\right]=0$. In the Kerr spacetime the  symmetry operator $Q$ associated to Carter's irreducible Killing 2-tensor  takes the form
\begin{equation}
Q\psi=\lapp_{\mathbb{S}^{2}}\psi-\Phi^{2}\psi+\left(a^{2}\sin^{2}\theta\right) T^{2}\psi.
\label{eq:q}
\end{equation}

The differential operator $Q$ was first used in the context of estimating solutions to the wave equation by Andersson and Blue \cite{bluekerr}. The authors commuted the wave equation with $Q$ before applying suitable vector field multipliers in order to obtain integrated local energy decay for slowly rotating Kerr backgrounds ($|a|\ll M$). The symmetry operator $Q$ was additionally used as a commutator for obtaining pointwise estimates. Note that in the subextreme case, this commutation was in fact unnecessary for obtaining pointwise estimates in view of the redshift commutation presented in \cite{dr7} (see the discussion in Section \ref{sec:EnergyAndPointwiseDecay}). However, in extreme Kerr, in view of the degeneracy of redshift,  commuting with $Q$ turns out to be useful (see Section \ref{sec:PointwiseEstimates}).

\section{The Vector Field $X$}
\label{sec:TheVectorFieldX}

The vector fields of the form $X=f(r^{*})\partial_{r^{*}}$ will be useful for constructing currents with non-negative definite divergence. Using $(t,r^{*},\theta,\phi^{*})$ coordinates we obtain for axisymmetric functions $\psi$
\begin{equation*}
\begin{split}
K^{X}[\psi]=&\left(\left[\frac{1}{2\rho^{2}}\frac{(r^{2}+M^{2})^{2}}{\Delta}-\frac{a^{2}\sin^{2}\theta}{2\rho^{2}}\right]f'+\left[\frac{r}{\rho^{2}}-\frac{a^{2}\sin^{2}\theta}{2\rho^{2}}(\partial_{r}D)\right]f\right)(T\psi)^{2}\\
& \left(\left[\frac{1}{2\rho^{2}}\frac{(r^{2}+M^{2})^{2}}{\Delta}\right]f'+\left[-\frac{r}{\rho^{2}}\right]f\right)(\partial_{r^{*}}\psi)^{2}+\left(-\frac{f'+f(\partial_{r}D)}{2}\right)\left|\nabb_{g}\psi\right|^{2},
\end{split}
\end{equation*}
where $D=\frac{\Delta}{r^{2}+M^{2}}$ and $\left|\nabb_{g}\psi\right|^{2}=\frac{1}{\rho^{2}}\left|\nabb_{\mathbb{S}^{2}}\psi\right|^{2}$. Note that $f'=\frac{df}{dr^{*}}$.

\subsection{Non-Negative Definite Currents Near and Away $\hh$}
\label{sec:NonNegativeDefiniteCurrentsNearAndAwayHh}

Let $\psi$ satisfy $\Box_{g}\psi=0$ and $\Phi\psi=0$. Consider $X=\partial_{r^{*}}$ (i.e. $f=1$). Consider also the current
\begin{equation*}
J_{\mu}^{\partial_{r^{*}},G}[\psi]=J_{\mu}^{\partial_{r^{*}}}[\psi]+2G\psi\nabla_{\mu}\psi-(\nabla_{\mu}G)\psi^{2}.
\end{equation*}
Then 
\begin{equation*}
K^{\partial_{r^{*}},G}[\psi]=K^{\partial_{r^{*}}}[\psi]+2G\left|\nabla\psi\right|^{2}-(\Box_{g}G)\psi^{2},
\end{equation*}
where 
\begin{equation*}
\left|\nabla\psi\right|^{2}=\left[-\frac{(r^{2}+M^{2})^{2}}{\Delta\rho^{2}}+\frac{a^{2}\sin^{2}\theta}{\rho^{2}}\right](T\psi)^{2}+\frac{(r^{2}+M^{2})^{2}}{\Delta\rho^{2}}(\partial_{r^{*}}\psi)^{2}+\left|\nabb_{g}\psi\right|^{2}.
\end{equation*}
Therefore, if we choose $G$ such that $G=\frac{\Delta\cdot r}{2(r^{2}+M^{2})^{2}}$, then 
\begin{equation*}
\begin{split}
K^{\partial_{r^{*}},G}[\psi]=&\left[\frac{a^{2}\sin^{2}\theta}{\rho^{2}}\left(2G-\frac{\partial_{r}D}{2}\right)\right](T\psi)^{2}+\left[-\frac{\partial_{r}D}{2}+2G\right]\left|\nabb_{g}\psi\right|^{2}-(\Box_{g}G)\psi^{2}\\
=&\frac{1}{\rho^{2}}\frac{(r-M)(r^{2}-2rM-M^{2})}{(r^{2}+M^{2})^{2}}\Bigg[a^{2}\sin^{2}\theta (T\psi)^{2}+\left|\nabb_{\mathbb{S}^{2}}\psi\right|^{2}+\left(\frac{3M(r-M)(r^{2}+2rM-M^{2})}{(r^{2}+M^{2})^{2}}\right)\psi^{2}\Bigg].
\end{split}
\end{equation*}
\textbf{The importance of this current lies in the observation that $K^{\partial_{r^{*}},G}[\psi]\geq 0$ for $r\geq (1+\sqrt{2})M$ and that $K^{-\partial_{r^{*}},-G}[\psi]\geq 0$ for $r\leq (1+\sqrt{2})M$.}

Let us now look at the boundary term
\begin{equation*}
\begin{split}
\int_{\left\{r=c\right\}}\!\!J_{\mu}^{\partial_{r}^{*},G}&[\psi]n^{\mu}\sqrt{\rho^{2}\Delta}\sin\theta\, dt\,d\theta\, d\phi
\end{split}
\end{equation*}
that arises when we apply the divergence identity at $r=c$ hypersurfaces ($c$ is a constant). Note  that the unit normal to such hypersurfaces pointing towards infinity is $n=\frac{r^{2}+M^{2}}{\sqrt{\rho^{2}\Delta}}\partial_{r^{*}}$. We have
\begin{equation}
\begin{split}
J_{\mu}^{\partial_{r}^{*},G}&[\psi]n^{\mu}\sqrt{\rho^{2}\Delta} =J_{\mu}^{\partial_{r}^{*},G}[\psi]\partial_{r^{*}}^{\mu}\cdot (r^{2}+M^{2})=
\Big[\textbf{T}_{r^{*}r^{*}}[\psi]+2G\psi (\partial_{r^{*}}\psi)-(\partial_{r^{*}}G)\cdot \psi^{2}\Big](r^{2}+M^{2})\\
=&\Big[(\partial_{r^{*}}\psi)^{2}-\frac{1}{2}g_{r^{*}r^{*}}\left|\nabla\psi\right|^{2}+2G\psi (\partial_{r^{*}}\psi)-(\partial_{r^{*}}G)\cdot \psi^{2}\Big](r^{2}+M^{2})\\
=&\left[\frac{r^{2}+M^{2}}{2}\right](\partial_{r^{*}}\psi)^{2}+\left[\frac{r^{2}+M^{2}}{2}\right](T\psi)^{2}-\left[\frac{\Delta}{2(r^{2}+M^{2})}\right]\Bigg[a^{2}\sin^{2}\theta (T\psi)^{2}+\left|\nabb_{\mathbb{S}^{2}}\psi\right|^{2}\Bigg]\\&+\left[\frac{\Delta\cdot r}{(r^{2}+M^{2})}\right]\psi(\partial_{r^{*}}\psi)-\Bigg[\frac{(r-M)^{3}}{2(r^{2}+M^{2})^{4}}\left(r^{3}+M^{3}-3r^{2}M-3rM^{2}\right)\Bigg]\psi^{2}.
\end{split}
\label{boundaryphys}
\end{equation}
For future use (see Section \ref{sec:PhysicalSpaceEstimates}), let us define
\begin{equation}
\begin{split}
\textbf{\textit{J}}[\psi]=& 2J_{\mu}^{\partial_{r}^{*},G}[\psi]n^{\mu}\sqrt{\rho^{2}\Delta}\\
=&\left[r^{2}+M^{2}\right](\partial_{r^{*}}\psi)^{2}+\left[r^{2}+M^{2}\right](T\psi)^{2}-\left[\frac{\Delta}{(r^{2}+M^{2})}\right]\Bigg[a^{2}\sin^{2}\theta (T\psi)^{2}+\left|\nabb_{\mathbb{S}^{2}}\psi\right|^{2}\Bigg]\\&+\left[\frac{\Delta\cdot 2r}{(r^{2}+M^{2})}\right]\psi(\partial_{r^{*}}\psi)-\Bigg[\frac{(r-M)^{3}}{(r^{2}+M^{2})^{4}}\left(r^{3}+M^{3}-3r^{2}M-3rM^{2}\right)\Bigg]\psi^{2}.
\end{split}
\label{defboundaryphys}
\end{equation}
Note that \textbf{\textit{J}}$[\psi]$ has in general no sign.

\section{Estimates near $\mathcal{H}^{+}$ and $\mathcal{I}^{+}$}
\label{sec:EstimatesNearMathcalHAndMathcalI}

\subsection{A Weighted Positive Definite Current near $\mathcal{I}^{+}$}
\label{sec:AWeightedPositiveDefiniteCurrentFarAwayFromMathcalH}

Regarding the neighbourhoods of $\mathcal{I}^{+}$, in \cite{mikraa} the following it is shown 

\begin{proposition}
Let $\delta>0$.  There exists a value $R_{e}$ which depends only on $M$ and a constant $C_{\delta}$ which depends  on $M,\delta$ such that for all solutions $\psi$ to the wave equation we have 
\begin{equation*}
\begin{split}
&\int_{\left\{r\geq R_{e}\right\}}\left[\frac{1}{r^{3+\delta}}\psi^{2}+\frac{1}{r^{1+\delta}}(\partial_{r^{*}}\psi)^{2}+\frac{1}{r^{1+\delta}}(T\psi)^{2}+\frac{1}{r}\left|\nabb\psi\right|^{2}\right]\\& \ \ \ \ \ \ \ \ \leq C_{\delta}\int_{\left\{R_{e}-1\leq r\leq R_{e}\right\}}\left[\psi^{2}+(\partial_{r^{*}}\psi)^{2}+(T\psi)^{2}+\left|\nabb\psi\right|^{2}\right]+C_{\delta}\int_{\Sigma_{0}}J_{\mu}^{T}[\psi]n^{\mu}_{\Sigma_{0}}
\end{split}
\end{equation*}
\label{larger}
\end{proposition}

\subsection{The Vector Field $N$}
\label{sec:TheVectorFieldN}

We will construct an appropriate current which will yield the desired estimate for a neigbourhood of $\hh$. The construction of such currents first appeared in \cite{aretakis2}.

We first note that there is no timelike vector field $N$ of the form $N=N^{T}(r)T+N^{Y}(r)Y$ such that $K^{N}[\psi]\geq 0$ on $\mathcal{H}^{+}$. Indeed,  for axisymmetric $\psi$ we have
\begin{equation*}
K^{N}[\psi]=F_{TT}(T\psi)^{2}+F_{YY}(Y\psi)^{2}+F_{TY}(T\psi)(Y\psi)+F_{\scriptsize\nabb}\left|\nabb\psi\right|^{2},
\end{equation*}
where 
\begin{equation*}
\begin{split}
&F_{TT}=\frac{r^{2}+a^{2}}{\rho^{2}}\frac{d}{dr}\!\!\left(N^{T}\right)-\frac{a^{2}\sin^{2}\theta}{2\rho^{2}}\frac{d}{dr}\!\!\left(N^{Y}\right),\\
&F_{YY}=\frac{\Delta}{2\rho^{2}}\frac{d}{dr}\!\!\left(N^{Y}\right)-\frac{\frac{d\Delta}{dr}}{2\rho^{2}}\left(N^{Y}\right),\\
&F_{\scriptsize\nabb}=-\frac{1}{2}\frac{d}{dr}\!\!\left(N^{Y}\right),\\
&F_{TY}=\frac{\Delta}{\rho^{2}}\frac{d}{dr}\!\!\left(N^{T}\right)-\frac{2r}{\rho^{2}}\left(N^{Y}\right).
\end{split}
\end{equation*}
Recall that in extreme Kerr we have $\Delta=(r-M)^{2}$, and thus, $\Delta$ vanishes on $\hh$ to second order. Note that the $F_{YY}$ vanishes on $\mathcal{H}^{+}$ whereas the coefficient $F_{TY}$ does not vanish, since $N^{Y}\neq 0$ on $\hh$. Hence, we can not obtain a positive definite current. However, we can still obtain a non-negative current by appropriately modifying the energy current $J_{\mu}^{N}[\psi]$. 

We first define in the region $\mathcal{A}_{N}=\left\{M\leq r\leq \frac{23}{21}M\right\}$ the vector field $N$ to be  such that
\begin{equation}
N^{T}(r)=18r-\frac{35}{2}M, \ \ \ \ \ \  N^{Y}(r)=-2r+M.
\label{defn}
\end{equation}
It is an easy computation to see that $N$ is a future directed timelike vector field in $\mathcal{A}_{N}$. We define the  current
\begin{equation}
J_{\mu}^{N,-\frac{1}{2}}[\psi]=J_{\mu}^{N}[\psi]-\frac{1}{2}\psi\nabla_{\mu}\psi.
\label{defJn}
\end{equation}
We have the following proposition
\begin{proposition}
For all axisymmetric functions $\psi$, the divergence $K^{N,-\frac{1}{2}}[\psi]$ of the current $J_{\mu}^{N,-\frac{1}{2}}[\psi]$ defined by \eqref{defJn} is non-negative definite in $\mathcal{A}_{N}$ and, in particular, there is a positive constant $C$ which depends only on $M$ such that 
\begin{equation}
K^{N,-\frac{1}{2}}[\psi]\geq C\left((T\psi)^{2}+\left(1-\frac{M}{r}\right)(Y\psi)^{2}+\left|\nabb\psi\right|^{2}\right).
\label{pkn}
\end{equation}
\label{nprop}
\end{proposition}
\begin{proof}
We have
\begin{equation*}
\begin{split}
K^{N,-\frac{1}{2}}[\psi]=K^{N}[\psi]-\frac{1}{2}\left|\nabb\psi\right|^{2}=G_{TT}(T\psi)^{2}+G_{YY}(Y\psi)^{2}+G_{TY}(T\psi)(Y\psi)+G_{\scriptsize\nabb}\left|\nabb\psi\right|^{2},
\end{split}
\end{equation*}
where
\begin{equation*}
\begin{split}
&G_{TT}=\frac{r^{2}+a^{2}}{\rho^{2}}\frac{d}{dr}\!\!\left(N^{T}\right)-\frac{a^{2}\sin^{2}\theta}{2\rho^{2}}\left(\frac{d}{dr}\!\!\left(N^{Y}\right)+1\right),\\
&G_{YY}=\frac{\Delta}{2\rho^{2}}\frac{d}{dr}\!\!\left(N^{Y}\right)-\frac{\frac{d\Delta}{dr}}{2\rho^{2}}N^{Y}-\frac{\Delta}{2\rho^{2}},\\
&G_{\scriptsize\nabb}=-\frac{1}{2}\frac{d}{dr}\!\!\left(N^{Y}\right)-\frac{1}{2},\\
&G_{TY}=\frac{\Delta}{\rho^{2}}\frac{d}{dr}\!\!\left(N^{T}\right)-\frac{2r}{\rho^{2}}\left(N^{Y}\right)-\frac{r^{2}+a^{2}}{\rho^{2}}.
\end{split}
\end{equation*}
Observe first that $G_{\scriptsize\nabb}=\frac{1}{2}>0$ and $G_{TT}\geq 18$. Moreover, $G_{YY}\sim \frac{\frac{d\Delta}{dr}}{\rho^{2}}$ since the dominant term $-\frac{\frac{d\Delta}{dr}}{2\rho^{2}}N^{Y}$ has the right sign. Regarding now the coefficient $G_{TY}$ of the mixed term we have 
\begin{equation*}
\begin{split}
G_{TY}=\frac{(r-M)}{\rho^{2}}\left[(r-M)\cdot \frac{d}{dr}\!\!\left(N^{T}\right)+(3r+M)\right]=\epsilon_{1}\cdot \epsilon_{2},
\end{split}
\end{equation*}
 where 
\begin{equation*}
\begin{split}
\epsilon_{1}=\frac{r-M}{\rho^{2}}, \ \ \epsilon_{2}=\frac{1}{\rho}\left[(r-M)\cdot \frac{d}{dr}\!\!\left(N^{T}\right)+(3r+M)\right].
\end{split}
\end{equation*}
However, 
\begin{equation*}
\begin{split}
\epsilon_{1}^{2}\leq G_{YY}\Leftrightarrow r \leq \frac{4}{3}M 
\end{split}
\end{equation*}
and 
\begin{equation*}
\begin{split}
\epsilon_{2}^{2}\leq G_{TT}\Leftrightarrow \frac{1}{\rho^{2}}\left[18(r-M)+(3r+M)\right]^{2}\leq \frac{1}{\rho^{2}}(r^{2}+M^{2})18\Leftrightarrow r\leq \frac{23}{21}M.
\end{split}
\end{equation*}
The proposition now follows from the spatial compactness of $\mathcal{A}_{N}$ and the inequality $\left|ab\right|\leq \frac{1}{2}(a^{2}+b^{2})$. 
\end{proof}

We will extend globally the vector field $N$ in Section \ref{sec:EnergyEstimates}, where it is used to show uniform boundedness of the non-degenerate energy.

Let $r_{e}<\frac{21}{23}M$. If $\mathcal{M}=\left\{r_{e}\leq r\leq R_{e}\right\}$, where $R_{e}$ is as defined in Proposition \ref{larger}, then, in view of the results of this section, it suffices to derive a non-negative definite estimate in region $\mathcal{M}$. 
\begin{figure}[H]
	\centering
		\includegraphics[scale=0.14]{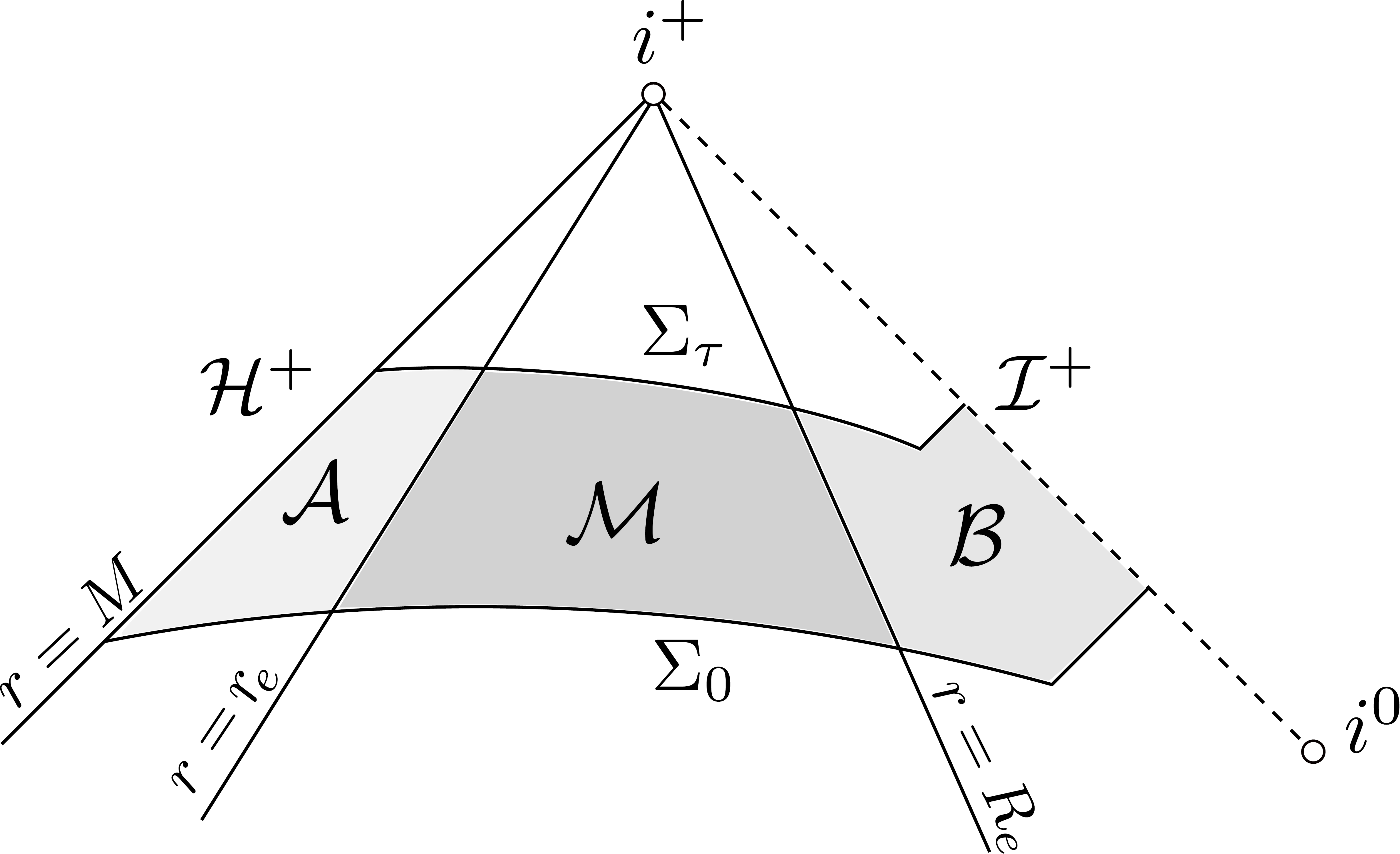}
	\label{fig:mregion}
\end{figure}

\section{Carter's Separability of the Wave Equation}
\label{sec:SeparabilityOfTheWaveEquation}

Separating the wave equation involves taking the Fourier transform in time. Since, a priori, we do not know that the solution is  $L^{2}(dt)$ we must first cut off in time.

\subsection{The Cut-off $\xi_{\tau}$}
\label{sec:TheCutOffXiVarepsilonTau}

Let $\xi_{\tau}$ be a cut-off function such that $\xi_{\tau}(\tilde{\tau})=0$ for $\tilde{\tau}\leq 0$ and $\tilde{\tau}\geq \tau$ and $\xi_{\tau}(\tilde{\tau})=1$ for $1\leq \tau\leq \tau-1$. 
\begin{figure}[H]
	\centering
		\includegraphics[scale=0.14]{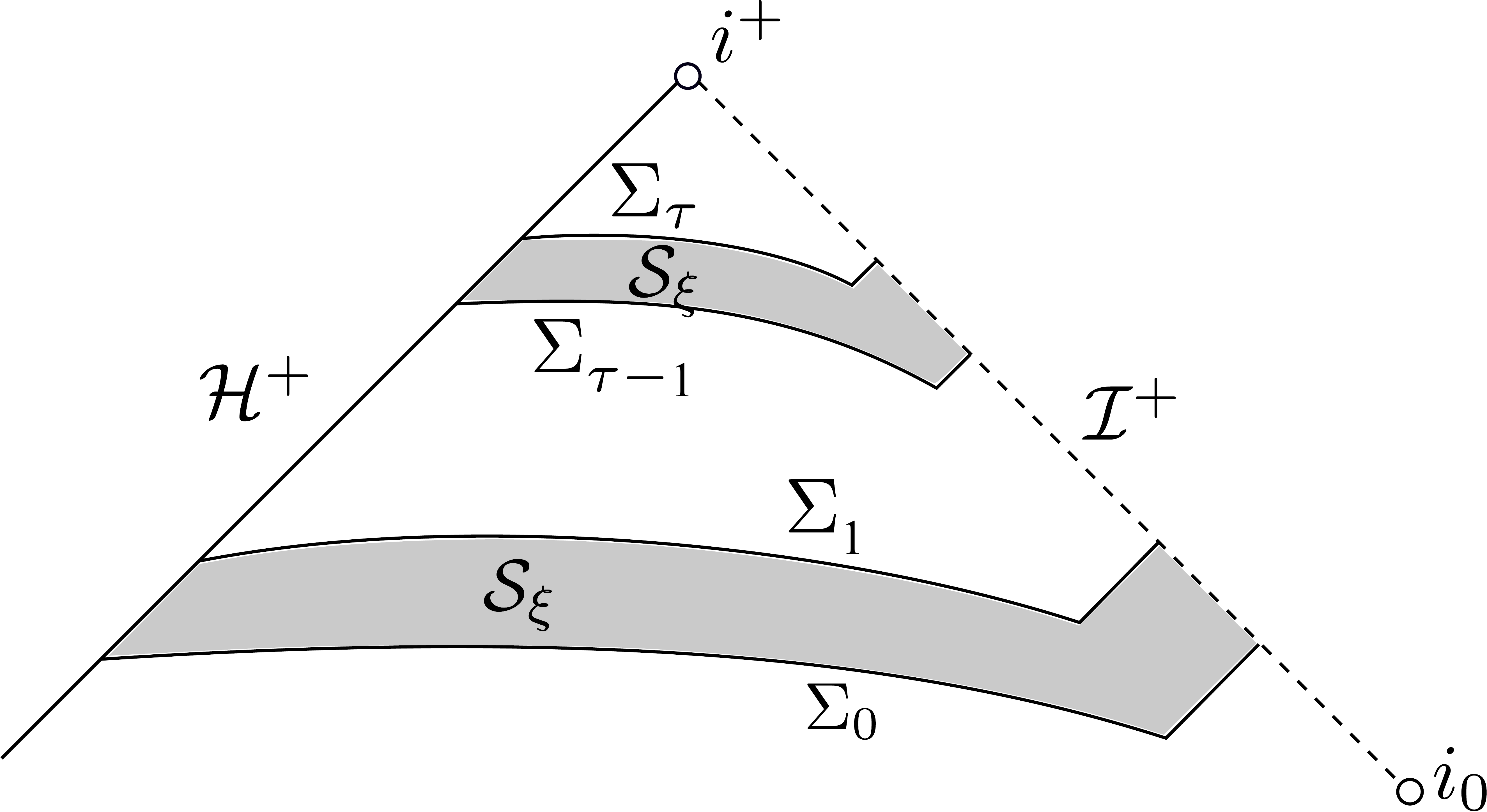}
	\label{fig:xicut}
\end{figure}
Let $\psi_{\hbox{\Rightscissors}}=\xi_{\tau}\psi$ be the compactly supported in time function which arises from the cut-off $\xi_{\tau}$ multiplied to the solution $\psi$ of the wave equation. This cut-off version of $\psi$ then satisfies the following inhomogeneous wave equation
\[
\Box_{g}\psi_{\hbox{\Rightscissors}}=F,
\]where 
\begin{equation}
F=2\nabla^{\mu}\xi_{\tau}\nabla_{\mu}\psi+\left(\Box\xi_{\tau}\right)\psi.
\label{eq:f}
\end{equation}

\subsection{Separability of the Wave Equation}
\label{sec:SeparabilityOfTheWaveEquation}

The complete separability of the wave equation on Kerr was used for the first time in the context of energy currents and  $L^{2}$ estimates in \cite{md, mikraa}. This separability requires decomposing in oblate spheroidal harmonics. For this reason we briefly describe this decomposition for the general case $|a|\leq M$ ($\psi$ is not assumed to be axisymmetric in this subsection).

Let $\xi\in\mathbb{R}$ and $P_{\xi}f$ denote the following elliptic operator acting on a suitable dense subset of $L^{2}(\mathbb{S}^{2})$
\begin{equation}
P_{\xi}f=-\lapp_{\mathbb{S}^{2}}f-\big(\xi^{2}\cos^{2}\theta\big) f.
\label{pxi}
\end{equation}
Clearly $P_{0}$ is the standard spherical Laplacian and more generally $P_{\xi}$  is the Laplacian on an oblate spheroid that is described by the parameter $\xi$. In view of standard elliptic theory, we can infer the existence of a complete orthonormal system $S_{m\ell}^{(\xi)}(\theta, \phi)$ with $m,\ell\in\mathbb{Z}, \ell\geq  \left|m\right|$ of $L^{2}(\mathbb{S}^{2})$ of eigenfunctions of $P_{\xi}$ with real eigenvalues $\lambda_{m\ell}^{(\xi)}$:
\begin{equation*}
P_{\xi}S_{m\ell}^{(\xi)}(\theta,\phi)=\lambda_{m\ell}^{(\xi)}\cdot S_{m\ell}^{(\xi)}(\theta,\phi).
\end{equation*}
The functions $S_{m\ell}^{(\xi)}$ are known as \textit{oblate spheroidal harmonics}.  Note that for $\xi=0$ these reduce to the standard spherical harmonics $Y_{m\ell}$ and $\lambda_{m\ell}^{(0)}=\ell(\ell+1)$. Given $\xi\in\mathbb{R}$,  any function $f\in L^{2}(\mathbb{S}^{2})$ can be decomposed as follows
\begin{equation*}
f(\theta,\phi)=\sum_{m,\ell}f^{(\xi)}_{m\ell}\cdot S_{m\ell}^{(\xi)}(\theta,\phi),
\end{equation*}
where 
\begin{equation*}
f^{(\xi)}_{m\ell}=\int_{\mathbb{S}^{2}}f(\theta,\phi)\cdot \overline{S_{m\ell}^{(\xi)}(\theta,\phi)}\,dg_{\mathbb{S}^{2}}. 
\end{equation*}
The following Parseval identity will be useful
\begin{equation*}
\int_{\mathbb{S}^{2}}\left|f\right|^{2}\,dg_{\mathbb{S}^{2}}=\sum_{m,\ell}\left|f^{(\xi)}_{m\ell}\right|^{2}.
\end{equation*}
All we will need about the eigenvalues $\lambda_{m\ell}^{(\xi)}$ is the following
\begin{proposition}
Let $\lambda^{(\xi)}_{m\ell}$ with $m,\ell\in\mathbb{Z}, \ell\geq  \left|m\right|$  be the  eigenvalues of $P_{\xi}$ defined by $\eqref{pxi}$. If we define  $\Lambda_{m\ell}^{(\xi)}=\lambda_{m\ell}^{(\xi)}+\xi^{2}$ then
\[\Lambda^{(\xi)}_{m\ell}\geq |m|\big(|m|+1\big). \]
\label{lambdap}
\end{proposition}
\begin{proof}
Let $f(\theta, \phi)\in H^{2}(\mathbb{S}^{2})$.  Then, by expanding in the $\phi$ variable we obtain
\begin{equation*}
f(\theta,\phi)=\sum_{m\in\mathbb{Z}}f_{m}(\theta)\cdot e^{im\phi},
\end{equation*}
where $f_{m}(\theta)=\frac{1}{2\pi}\int_{0}^{2\pi}f(\theta,\phi)\cdot \overline{e^{im\phi}} \, d\phi $. Consider now the operator
\begin{equation*}
P_{\xi, m}R=-\frac{1}{\sin\theta}\frac{d}{d\theta}\Bigg(\sin\theta\frac{d}{d\theta }R\Bigg)+\frac{m^{2}}{\sin^{2}\theta}R-\big(\xi^{2}\cos^{2}\theta\big)R.
\end{equation*}
Since $f$ and  $f_{m}$ are regular at the poles $\theta=0,\pi$, in view of the Sturm-Liouville theory, we have that there exists a complete orthonormal system $R^{(\xi)}_{m\ell}$, with $m,\ell\in\mathbb{Z}, l\geq |m|$ of $L^{2}(\sin\theta\, d\theta)$ of $P_{\xi,m}$ with real eigenvalues $\lambda_{m\ell}^{(\xi)}$. Then,
\begin{equation*}
f(\theta,\phi)=\sum_{m\in\mathbb{Z}}\sum_{\ell\geq |m|}f^{(\xi)}_{m\ell}\cdot R^{(\xi)}_{m\ell}(\theta)\cdot e^{im\phi},
\end{equation*}
where $f^{(\xi)}_{m\ell}=\int_{0}^{\pi}f_{m}(\theta)\cdot \overline{R^{(\xi)}_{m\ell}}\sin\theta\, d\theta$. Hence, 
\begin{equation}
S_{m\ell}^{(\xi)}(\theta,\phi)=R^{(\xi)}_{m\ell}(\theta)\cdot e^{im\phi}.
\label{sform}
\end{equation} 
Consider the standard spherical decomposition  of the function $S_{m\ell}^{(\xi)}(\theta,\phi)$:
\begin{equation*}
S_{m\ell}^{(\xi)}(\theta,\phi)=\sum_{m'}\sum_{L\geq |m'|}\big(S_{m\ell}^{(\xi)}\big)_{m'L}\cdot Y_{m'L}(\theta,\phi),
\end{equation*}
where $Y_{m'L}$ are the standard spherical harmonics and $\big(S_{m\ell}^{(\xi)}\big)_{m'L}=\int_{\mathbb{S}^{2}}\big(S_{m\ell}^{(\xi)}\big)\cdot Y_{m'L} dg_{\mathbb{S}^{2}}$. Hence, in view of \eqref{sform}, if $m'\neq m$ then $\big(S_{m\ell}^{(\xi)}\big)_{m'L}=0$. Therefore, 
\begin{equation*}
S_{m\ell}^{(\xi)}(\theta,\phi)=\sum_{L\geq |m|}\big(S_{m\ell}^{(\xi)}\big)_{mL}\cdot Y_{mL}(\theta,\phi),
\end{equation*}
which shows that $S_{m\ell}^{(\xi)}(\theta,\phi)$ is supported on (standard) angular frequencies greater or equal than $|m|$. By Poincar\'{e} inequality (see \cite{aretakis2}) we obtain
\begin{equation*}
\begin{split}
\Lambda_{m\ell}^{(\xi)}=&\Lambda_{m\ell}^{(\xi)}\int_{\mathbb{S}^{2}}\left|S_{m\ell}^{(\xi)}\right|^{2}=\int_{\mathbb{S}^{2}}\big(\lambda_{m\ell}^{(\xi)}+\xi^{2}\big)S_{m\ell}^{(\xi)}\cdot
\overline{S_{m\ell}^{(\xi)}}\\
=&\int_{\mathbb{S}^{2}}\big(P_{\xi}+\xi^{2}\big)S_{m\ell}^{(\xi)}\cdot \overline{S_{m\ell}^{(\xi)}}=\int_{\mathbb{S}^{2}}\big(-\lapp_{\mathbb{S}^{2}}+\xi^{2}\sin^{2}\theta\big)S_{m\ell}^{(\xi)}\cdot \overline{S_{m\ell}^{(\xi)}}\geq \int_{\mathbb{S}^{2}}-\lapp_{\mathbb{S}^{2}}S_{m\ell}^{(\xi)}\cdot \overline{S_{m\ell}^{(\xi)}}\\=&\int_{\mathbb{S}^{2}}\left|\nabb_{\mathbb{S}^{2}}S_{m\ell}^{(\xi)}\right|^{2}
\geq |m|\big(|m|+1\big)\int_{\mathbb{S}^{2}}\left|S_{m\ell}^{(\xi)}\right|^{2}= |m|\big(|m|+1\big).
\end{split}
\end{equation*}
\end{proof}

Let us return now to the Kerr geometry. We use the Boyer-Lindquist coordinate system $(t,r,\theta,\phi)$. Let $\widehat{\psi_{\hbox{\Rightscissors}}}(\omega, r,\theta, \phi)$ denote the Fourier transform with respect to $t$, i.e.
\begin{equation*}
\widehat{\psi_{\hbox{\Rightscissors}}}(\omega, r,\theta, \phi)=\frac{1}{\sqrt{2\pi}}\int_{-\infty}^{+\infty}\psi_{\hbox{\Rightscissors}}(t,r,\theta,\phi)\cdot e^{i\omega t}dt.
\end{equation*}
Since $\psi_{\hbox{\Rightscissors}}$ is compactly supported in $t$ we have that $\widehat{\psi_{\hbox{\Rightscissors}}}$ is of Schwartz class in $\omega\in\mathbb{R}$ and, moreover, smooth in $r>M$ and smooth on the $(\theta,\phi)$ spheres. For each $\omega\in\mathbb{R}$, we can further decompose $\widehat{\psi_{\hbox{\Rightscissors}}}(\omega, r,\theta, \phi)$ in oblate spheroidal harmonics
\begin{equation*}
\widehat{\psi_{\hbox{\Rightscissors}}}(\omega, r,\theta, \phi)=\sum_{m,\ell}\widehat{\psi_{\hbox{\Rightscissors}}}^{(a\omega)}_{m\ell}(\omega, r)\cdot S_{m\ell}^{(a\omega)}(\theta,\phi).
\end{equation*}
In other words, for the temporal frequency $\omega$, we take the spheroidal parameter to be $\xi=a\omega$. We are thus led to the following decomposition 
\begin{equation*}
\psi(t,r,\theta,\phi)=\frac{1}{\sqrt{2\pi}}\int_{-\infty}^{+\infty}\sum_{m,\ell}\widehat{\psi_{\hbox{\Rightscissors}}}^{(a\omega)}_{m\ell}(r)\cdot S_{m\ell}^{(a\omega)}(\theta,\phi)\cdot e^{-i\omega t}\,d\omega.
\end{equation*}
Then the function  \[u^{(a\omega)}_{m\ell}(r^{*})=(r^{2}+a^{2})^{1/2}\cdot \widehat{\psi_{\hbox{\Rightscissors}}}^{(a\omega)}_{m\ell}(r)\]  satisfies the following ODE
\begin{equation}
\label{e3iswsntouu}
\frac{d^2}{(dr^*)^2}u^{(a\omega)}_{m\ell}+\Big(\omega^2 - V^{(a\omega)}_{m\ell }(r)\Big)u^{(a\omega)}_{m\ell} = 
H^{(a\omega)}_{m\ell},
\end{equation}
where
\begin{equation}
 H^{(a\omega)}_{m\ell}(r)=\frac{\Delta F^{(a\omega)}_{m\ell}(r)}{(r^2+a^2)^{1/2}}
\label{hf}
\end{equation}
and
\[
V^{(a\omega)}_{m \ell}(r)= \frac{4Mram\omega-a^2m^2+\Delta\cdot\Lambda_{m\ell}^{(a\omega)}}{(r^2+a^2)^2}
+\frac{\Delta(3r^2-4Mr+a^2)}{(r^2+a^2)^3}
-\frac{3\Delta^2 r^2}{(r^2+a^2)^4},
\]
where $\Lambda_{m\ell}^{(a\omega)}=\lambda_{m\ell}^{(a\omega)}+a^{2}\omega^{2}$.  Although $u$ is a complex-valued function of $r$, the potential $V$ is real\footnote{It is easy to verify that the same remark would fail to hold when we separate with respect to $(v, r, \theta, \phi^{*})$.}!

\subsection{Properties of the Potential $V$}
\label{sec:PropertiesOfThePotentialV}

By returing to the extreme case $|a|=M$, considering axisymmetric $\psi$ (i.e.~$m=0$) and by dropping the indices we obtain
\begin{equation*}
\begin{split}
V=\frac{(r-M)^{2}}{(r^{2}+M^{2})^{2}}\Lambda+\frac{(r-M)^{3}M}{(r^{2}+M^{2})^{4}}\left(2r^{2}+3rM-M^{2}\right).
\end{split}
\end{equation*}
\begin{figure}[H]
	\centering
		\includegraphics[scale=0.14]{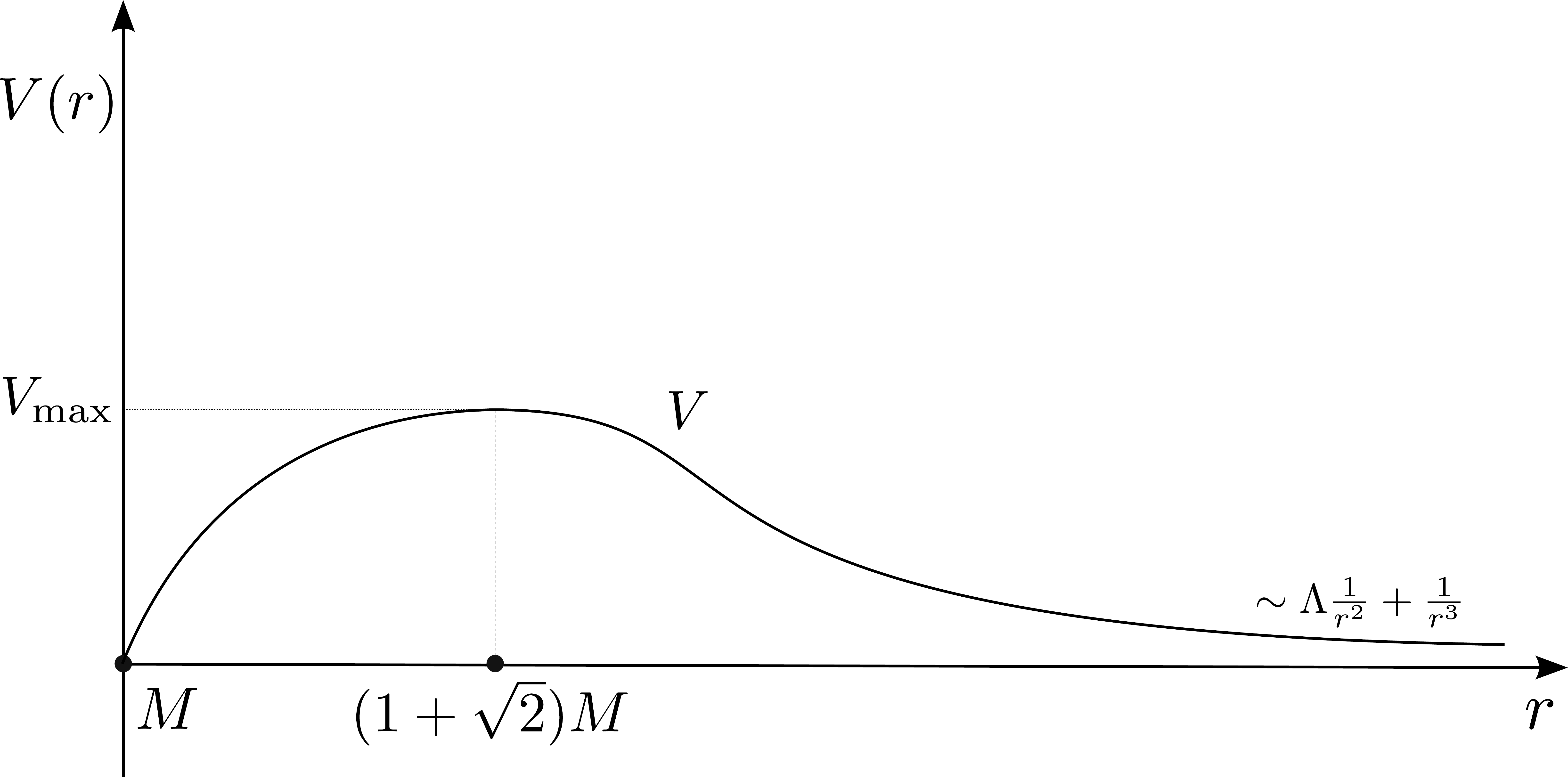}
	\label{fig:vgp1}
\end{figure}
Regarding the maximum value of $V$ we have the following
\begin{lemma}
For all frequencies $(\omega,\Lambda)$ we have 
\begin{equation*}
V_{\text{max}}\leq \frac{1}{4M^{2}}\Lambda+\frac{5}{M^{4}}
\end{equation*}
\label{vmax}
\end{lemma}
\begin{proof}
It suffices to note that 
\begin{equation*}
\frac{(r-M)^{2}}{(r^{2}+M^{2})^{2}}\leq \frac{1}{4M^{2}}
\end{equation*}
for all $r\geq M$. In fact, as expected, this expression attains its maximum at $r=(1+\sqrt{2})$ and so the constant on the right hand side could be further improved. Moreover, we have the following  bound
\begin{equation*}
\frac{(r-M)^{3}M}{(r^{2}+M^{2})^{4}}(2r^{2}+3rM-M^{2})\leq \frac{5r^{3}M}{(r^{2}+M^{2})^{4}}\leq \frac{5r^{4}}{r^{8}}\leq \frac{5}{M^{4}}.
\end{equation*}

\end{proof}

It is clear that for all frequencies $\Lambda, \omega$ we have 
\begin{equation}
\begin{split}
V\sim \frac{\Delta}{r^{4}}\Lambda+\frac{\Delta^{3/2}}{r^{6}},
\end{split}
\label{V}
\end{equation}
for all $r\geq M$. The $\sim$ depends only on $M$. In  particular, we have
\begin{equation}
\begin{split}
&V\sim \frac{1}{r^{2}}\Lambda+\frac{1}{r^{3}}, \text{ for all } r\geq R,\\
&V\sim \frac{1}{(r^{*})^{2}}\Lambda-\frac{1}{(r^{*})^{3}}, \text{ for all } r\leq r_{e}.
\end{split}
\label{V1}
\end{equation}
The above shows that \textbf{in extreme Kerr the potential $V$ exhibits a symmetric asymptotic behaviour towards the event horizon and infinity.} This observation will be crucial for our constructions. 

As regards the derivative with respect to $r^{*}$ of $V$ we have in this case
\begin{equation*}
\begin{split}
V'=\frac{2(r-M)^{3}}{(r^{2}+M^{2})^{4}}\left[r-(1+\sqrt{2})M\right]\!\left[(r-M)+\sqrt{2}M\right]\!\left[-\Lambda-\frac{3M(r-M)}{(r^{2}+M^{2})^{2}}(r^{2}+2rM-M^{2})\right].
\end{split}
\end{equation*}

\begin{figure}[H]
	\centering
		\includegraphics[scale=0.14]{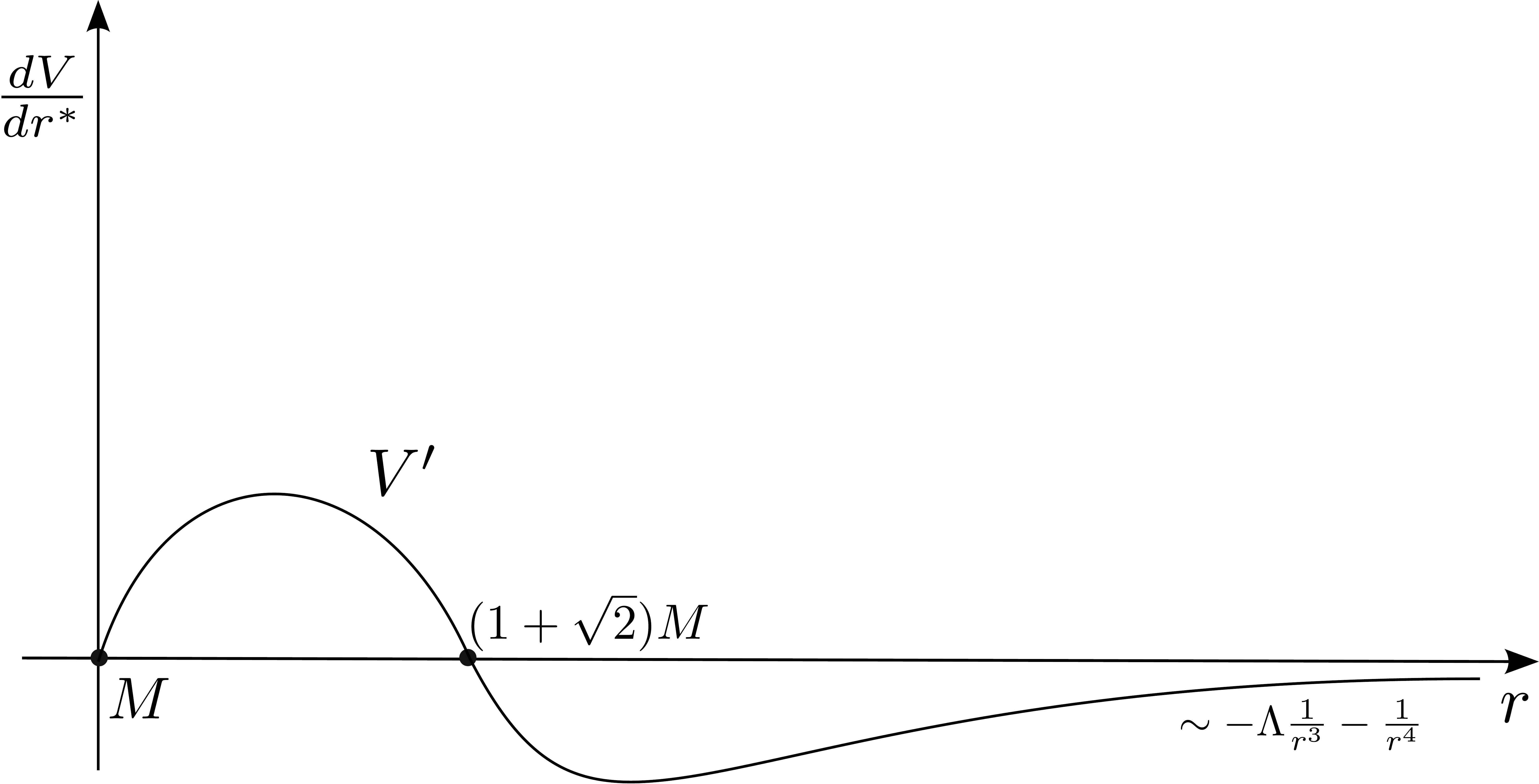}
	\label{fig:vtonos1}
\end{figure}
 Note that for all frequencies, $V'$ vanishes exactly at $r=M$ and $r=(1+\sqrt{2})M$.  This property of trapping for extreme black holes was already captured in \cite{aretakis2} using physical space techniques. One should probably contrast at this point this behaviour of $V'$ with the subextreme case. If $|a|<M$ and $m=0$ then $V'$ vanishes at $r=r_{\mathcal{H}^{+}}$ and $r=r_{\Lambda}$, where $r_{\Lambda}$ depends on the frequency $\Lambda$ and the limit $\lim_{\Lambda\rightarrow+\infty}r_{\Lambda}$ is a value of $r$ which does not depend on $\omega, \Lambda, m$.  See also the discussion in Section \ref{sec:TheTrappedFrequenciesMathcalF21}.

 Recall now the value $r_{e}=\frac{23}{21}M<(1+\sqrt{2})M$ related to the modified redshift of Section \ref{sec:TheVectorFieldN}. Then in region $\mathcal{A}=\left\{M\leq r\leq r_{e}\right\}$ we have
\begin{equation}
\begin{split}
V'\sim\Delta^{3/2}\Lambda+\Delta^{2}
\end{split}
\label{vtonoskonta}
\end{equation}
and, if $R_{e}$ is as defined in  Proposition \ref{larger}, then  in region $\mathcal{B}=\left\{R_{e}\leq r\right\}$
\begin{equation}
\begin{split}
V'\sim-\frac{1}{r^{3}}\Lambda-\frac{1}{r^{4}}.
\end{split}
\label{vtonosmarkua}
\end{equation}
Finally,  in the region $\mathcal{M}=\left\{r_{e}\leq r\leq R_{e}\right\}$ we have 
 \begin{equation}
\begin{split}
V'\sim\left[(1+\sqrt{2})M-r\right]\left(\Lambda+1\right)
\end{split}
\label{v'}
\end{equation}
In all the above instances, the $\sim$ depends only on $M$.

\section{Physical space - Fourier space Estimates}
\label{sec:PhysicalSpaceFourierSpace}

In view of the fundamental properties of the Fourier transform we have
\begin{equation*}
\begin{split}
\int_{-\infty}^{+\infty}\int_{\mathbb{S}^{2}}\left(\psi_{\hbox{\Rightscissors}}\right)^{2}r^{2}dg_{\mathbb{S}^{2}}dt\leq
\int_{-\infty}^{+\infty}\sum_{m,l}\left|u^{(a\omega)}_{ml}(r)\right|^{2}d\omega, 
\end{split}
\end{equation*}
\begin{equation*}
\begin{split}
\int_{-\infty}^{+\infty}\int_{\mathbb{S}^{2}}\left(T\psi_{\hbox{\Rightscissors}}\right)^{2}r^{2}dg_{\mathbb{S}^{2}}dt\leq
\int_{-\infty}^{+\infty}\sum_{m,l}\omega^{2}\left|u^{(a\omega)}_{ml}(r)\right|^{2}d\omega, 
\end{split}
\end{equation*}
\begin{equation*}
\begin{split}
\int_{-\infty}^{+\infty}\int_{\mathbb{S}^{2}}\left|\nabb_{g}\psi_{\hbox{\Rightscissors}}\right|^{2}r^{2}dg_{\mathbb{S}^{2}}dt\leq\int_{-\infty}^{+\infty}\sum_{m,l}\Lambda_{m\ell}^{(a\omega)}\left|u^{(a\omega)}_{ml}(r)\right|^{2}d\omega,
\end{split}
\end{equation*}
\begin{equation*}
\begin{split}
\int_{-\infty}^{+\infty}\int_{\mathbb{S}^{2}}\left(\partial_{r^{*}}\psi_{\hbox{\Rightscissors}}\right)^{2}r^{2}dg_{\mathbb{S}^{2}}dt\leq \int_{-\infty}^{+\infty}\sum_{m,l}\left[2\left|\frac{d}{dr^{*}}u^{(a\omega)}_{ml}(r)\right|^{2}+\frac{2}{M^{2}}\left|u^{(a\omega)}_{ml}(r)\right|^{2}\right]d\omega,
\end{split} 
\end{equation*}
where we have used that 
\begin{equation*}
\begin{split}
r^{2}(\partial_{r^{*}}\psi_{\hbox{\Rightscissors}})^{2}\leq 2\left[\partial_{r^{*}}\left(\sqrt{r^{2}+a^{2}}\psi_{\hbox{\Rightscissors}}\right)\right]^{2}+\frac{2}{M^{2}}\left(\psi_{\hbox{\Rightscissors}}\right)^{2}r^{2}.
\end{split}
\end{equation*}
Recall that $\nabb_{g}$ denotes the gradient on $\mathbb{S}^{2}$ with respect to the induced metric from the Kerr metric. 
In fact, if we suppress the indices we obtain  the following identities:
\begin{equation*}
\begin{split}
\int_{-\infty}^{+\infty}\sum_{m,l}\left|u\right|^{2}d\omega = \int_{-\infty}^{+\infty}\int_{\mathbb{S}^{2}}\left(\psi_{\hbox{\Rightscissors}}\right)^{2}\cdot(r^{2}+M^{2})\,dt\, dg_{\mathbb{S}^{2}},
\end{split}
\end{equation*}
\begin{equation*}
\begin{split}
\int_{-\infty}^{+\infty}\sum_{m,l}\omega^{2}\left|u\right|^{2}d\omega = \int_{-\infty}^{+\infty}\int_{\mathbb{S}^{2}}\left(T\psi_{\hbox{\Rightscissors}}\right)^{2}\cdot(r^{2}+M^{2})\,dt\, dg_{\mathbb{S}^{2}},
\end{split}
\end{equation*}
\begin{equation*}
\begin{split}
\int_{-\infty}^{+\infty}\sum_{m,l}\Lambda\left|u\right|^{2}d\omega = \int_{-\infty}^{+\infty}\int_{\mathbb{S}^{2}}\left[\left|\nabb_{\mathbb{S}^{2}}\psi_{\hbox{\Rightscissors}}\right|^{2}+a^{2}\sin^{2}\theta\cdot\left(T\psi_{\hbox{\Rightscissors}}\right)^{2}\right]\cdot(r^{2}+M^{2})\,dt\, dg_{\mathbb{S}^{2}},
\end{split}
\end{equation*}
\begin{equation*}
\begin{split}
\int_{-\infty}^{+\infty}\sum_{m,l}\left|u'\right|^{2}d\omega = \int_{-\infty}^{+\infty}\int_{\mathbb{S}^{2}}\left(\partial_{r^{*}}\left(\sqrt{r^{2}+M^{2}}\cdot \psi_{\hbox{\Rightscissors}}\right)\right)^{2}dt\, dg_{\mathbb{S}^{2}},
\end{split}
\end{equation*}
where $\nabb_{\mathbb{S}^{2}}\psi$ denotes the gradient of $\psi$ on the unit sphere with respect to the standard metric and $dg_{\mathbb{S}^{2}}=\sin\theta \,d\theta\, d\phi$.

\section{Microlocal Energy Currents}
\label{sec:MicrolocalEnergyCurrents}
We next introduce the Fourier-localised energy currents that will allow us to derive $L^{2}(dr^{*})$ estimates for each $u_{m\ell}^{(a\omega)}$. Then by futher summing over $\lambda_{m\ell},m$,  integrating in $\omega$ and using the above estimates, we will obtain the required result in physical space.

It is clear from the estimates of Section \ref{sec:PhysicalSpaceFourierSpace} that we have to derive estimates for the quantity
\begin{equation}
\Pi_{m\ell}^{(a\omega)}(\omega^{2},\Lambda)= |u'|^{2}+|u|^{2}+\omega^{2}|u|^{2}+\Lambda|u|^{2}
\label{pi}
\end{equation}
for all frequencies $\omega^{2}$ and $\Lambda$.

 Without further delay, for arbitrary functions $y(r^{*}),h(r^{*}),f(r^{*})$, we define the currents (see also \cite{mikraa})
\begin{equation}
\begin{split}
&\mathcal{J}^{y}_{1}[u]=y\left[\left|u'\right|^{2}+(\omega^{2}-V)\left|u\right|^{2}\right],\\
&\mathcal{J}_{2}^{h}[u]=h\text{Re}(u'\overline{u})-\frac{1}{2}h'\left|u\right|^{2},\\
&\mathcal{J}_{3}^{f}[u]=f\left[\left|u'\right|^{2}+(\omega^{2}-V)\left|u\right|^{2}\right]+f'\text{Re}(u'\overline{u})-\frac{1}{2}f''\left|u\right|^{2}.\\
\end{split}
\label{currents}
\end{equation}
We shall construct several combinations of these currents with appropriate multiplier functions $y,h,f$ such that the derivatives of the combined currents are non-negative definite (possibly modulo small error terms). For convenience, we include the computations
\begin{equation}
\begin{split}
&\Big(\mathcal{J}^{y}_{1}[u]\Big)'=y'\left[\left|u'\right|^{2}+(\omega^{2}-V)\left|u\right|^{2}\right]-yV'\left|u\right|^{2}+2y\text{Re}(u'\overline{H}),\\
&\Big(\mathcal{J}_{2}^{h}[u]\Big)'=h\left[\left|u'\right|^{2}+(V-\omega^{2})\left|u\right|^{2}\right]-\frac{1}{2}h''\left|u\right|^{2}+h\text{Re}(u\overline{H}),\\
&\left(\mathcal{J}_{3}^{f}[u]\right)'=2f'\left|u'\right|^{2}-fV'\left|u\right|^{2}-\frac{1}{2}f'''\left|u\right|^{2}+2f\text{Re}(u'\overline{H})+f'\text{Re}(u\overline{H}).
\end{split}
\label{derivatives}
\end{equation}
Note that these currents are well-defined whenever $y\in C^{0}[M,+\infty]$ and $h\in C^{1}[M,+\infty]$ and the higher order derivatives are just integrable.
The expressions for the derivatives of these currents involve only the real-valued functions $|u|^{2}$ and $|u'|^{2}$ (modulo the error terms from the cut-off). This is precisely the motivation for these currents. 

We shall  require the multipliers to remain bounded as $r\rightarrow M$ and $r\rightarrow+\infty$.  In order to control the boundary terms we  require the multipliers of the form $y,f$ to be 1 for sufficiently large  values of $r^{*}$ and $-1$ for sufficiently small values of $r^{*}$  and $h$ to compactly supported in  $r^{*}$.  For simplicity, all the terms which involve the inhomegeneous term $H$ will be denoted $E(H)$. These terms are controlled in Section \ref{sec:AuxilliaryCurrents}.

\section{Fourier Localised Estimates}
\label{sec:FourierLocalisedEstimates}

In constructing currents which yield $L^{2}(dr^{*})$ estimates for $u_{m\ell}^{(a\omega)}$ (and $\frac{d}{dr^{*}}u_{m\ell}^{(a\omega)}$), it is always convenient to have frequency parameters which are either very small, or very large (unbounded) or bounded but away from zero. In our case, the frequencies are $\Lambda, \omega$. Note that in view of Proposition \ref{lambdap} we have $\Lambda\geq 0$.  It is useful to mention that the potentially small frequency parameter $\omega_{0}$ and the potentially large parameters $\omega_{1}, \lambda_{1},\lambda_{2}$ (once determined) depend only on $M$.  We thus consider the following frequency ranges
 
\begin{enumerate}
	\item  Bounded frequency range \[\mathcal{F}_{1}=\left\{(\omega, \Lambda):\left|\omega\right|\leq\omega_{1}, 0\leq \Lambda\leq \lambda_{1}\right\}\] 
\begin{itemize}
	\item The near-stationary frequencies
		\[\mathcal{F}_{1,1}=\left\{(\omega, \Lambda):\left|\omega\right|\leq\omega_{0}\ll 1,  \Lambda\leq \lambda_{1}\right\}\]
	
	\item The non-stationary frequencies
\end{itemize}

		\[\mathcal{F}_{1,2}=\left\{(\omega, \Lambda):\omega_{0}<\left|\omega\right|\leq\omega_{1},  \Lambda\leq \lambda_{1}\right\}\]

\item 
Unbounded frequency range  \[\mathcal{F}_{2}=\left\{(\omega, \Lambda):\left|\omega\right|>\omega_{1} \text{ or }\Lambda> \lambda_{1}\right\}\]
\begin{itemize}
	\item Trapped frequencies
	\[\mathcal{F}_{2,1}=\left\{(\omega, \Lambda):\left|\omega\right|>\omega_{1}, 2M^{2}\omega^{2}\leq \Lambda\leq \lambda_{2}\omega^{2}\right\}\]
	
	\item Angular dominated frequencies
		\[\mathcal{F}_{2,2}=\left\{(\omega, \Lambda):\left|\omega\right|\leq \omega_{1}, \Lambda >\lambda_{1} \right\}\cup\left\{(\omega, \Lambda):\left|\omega\right|> \omega_{1}, \Lambda>\lambda_{2}\omega^{2} \right\}\]
	
	\item Time dominated frequencies
		\[\mathcal{F}_{2,3}=\left\{(\omega, \Lambda):\left|\omega\right|> \omega_{1}, \Lambda \leq 2M^{2}\omega^{2} \right\}\]

\end{itemize}

\end{enumerate}
 
Note that we are only looking for positive definite ``bulk'' terms in the region $r_{e}\leq r\leq R_{e}$, where  $R_{e}$ is as defined in  Proposition \ref{larger}. 

\subsection{The Bounded Frequency Range $\mathcal{F}_{1}$}
\label{sec:TheBoundedFrequencyRangeF1}

First observe that it suffices to derive an integrated estimate for the quantity 
 \begin{equation*}
\begin{split}
\Pi_{1}=|u'|^{2}+|u|^{2}.
 \end{split}
\end{equation*}
Indeed, in view of the boundedness of $\omega^{2}$ and $\Lambda$, the quantity $\Pi$ defined by $\eqref{pi}$ is dominated by the above expression.

The parameters $\omega^{2}$ and $\Lambda$ are bounded. As we shall see, it is convenient to split this range in two subranges. The first range corresponds to the case where the frequency $\omega^{2}$ is very small (i.e.~`almost' zero) and the second range concerns $\omega^{2}$ which are bounded but bigger than a strictly positive constant.

\subsubsection{The Near-Stationary Range $\mathcal{F}_{1,1}$}
\label{sec:TheNearStationaryCase}
 This frequency range is defined by
 \begin{equation*}
\begin{split}
\omega^{2}\leq \omega^{2}_{0}, \ \ \  0\leq\Lambda\leq \lambda_{1},
 \end{split}
\end{equation*}
where $\omega_{0}$ will be chosen very small.

First observe that if $\omega=0$ exactly, then  for $h=1$
 \begin{equation*}
\begin{split}
\left(\mathcal{J}_{2}^{h=1}[u]\right)'=|u'|^{2}+V|u|^{2}+E(H),
 \end{split}
\end{equation*}
for all $r\leq M$, which in view of the positivity of $V$ suffices to control the quantity $\Pi_{1}$ in region $r_{e}\leq r\leq R_{e}$. However, although $\omega$ is very small in $\mathcal{F}_{1,1}$, it may not be exactly zero. Therefore, the term $V-\omega^{2}$ (which would appear in the very same current) will cease to be positive for large $r$ and for $r$ close to $M$ (in view of the degeneracy of $V$ at these limits). However, for sufficiently small $\omega_{0}$, the term $V-\omega^{2}$ is positive in the region $\mathcal{M}=\left\{r_{e}\leq r\leq R_{e}\right\}$. That is to say, we need to cut off the function $h$ in regions $\mathcal{A}=\left\{r\leq r_{e}\right\}$ and $\mathcal{B}=\left\{r\geq R_{e}\right\}$.  In order to absorb the error terms that would arise, we  add the current $J_{1}^{y}[u]$, for an appropriate function $y$ which is supported only in regions $\mathcal{A},\mathcal{B}$. 

Note first that 
 \begin{equation*}
\begin{split}
\left(\mathcal{J}_{2}^{h}[u]+\mathcal{J}^{y}_{1}[u]\right)'&=(h+y')|u'|^{2}+(h-y')(V-\omega^{2})|u|^{2}+\left(-\frac{1}{2}h''-yV'\right)|u|^{2}+E(H)\\
&=(h+y')|u'|^{2}+h(V-\omega^{2})|u|^{2}+y'\omega^{2}|u|^{2}+\left(-\frac{1}{2}h''-(yV)'\right)|u|^{2}+E(H).
 \end{split}
\end{equation*}
\begin{figure}[H]
	\centering
		\includegraphics[scale=0.1]{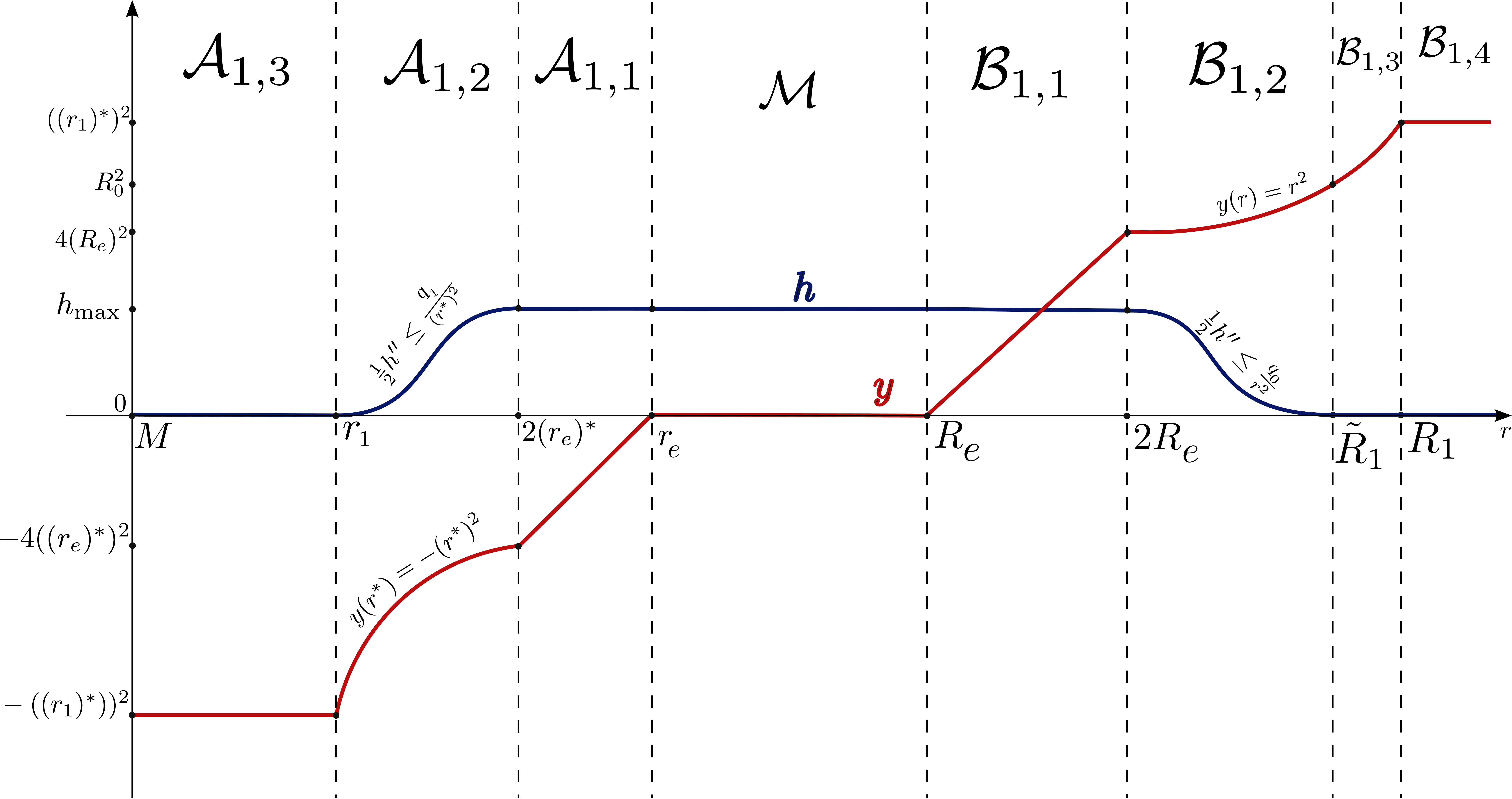}
	\label{fig:f11}
\end{figure}
Recall that $r^{*}\rightarrow -\infty$ ar $r\rightarrow M$ and since $r_{e}<(1+\sqrt{2})M$ we have $(r_{e})^{*}<0$.  We next present the details of this construction for each of the above regions.

In region $\mathcal{M}$, $h=4\text{max}\left\{-4(r_{e})^{*}, 4R_{e}\right\}>0$ (constant), $y=0$  and so, as we have already shown, by taking $\omega_{0}$ sufficiently small, these choices give us the positivity that we want.

In region $\mathcal{A}_{1,1}=\left\{2(r_{e})^{*}\leq r^{*}\leq (r_{e})^{*}\right\}$, $h=4\text{max}\left\{-4(r_{e})^{*}, 4R_{e}\right\}$, $y(r_{e})=0, y(2(r_{e})^{*})=-(2(r_{e})^{*})^{2}$ and $y$ is linear with respect to the $r^{*}$ variable. Therefore, $y'=\frac{dy}{dr^{*}}=-4(r_{e})^{*}$. Note now that if we choose $\omega_{0}$ such that $V-\omega_{0}^{2}\geq \frac{V}{2}$ in region $\mathcal{A}_{1,1}$, then 
 \begin{equation*}
\begin{split}
h\cdot(V-\omega^{2})-yV'-y'V\geq h\frac{V}{2}-y'V=\left(\frac{h}{2}-y'\right)\cdot V> 0,
\end{split}
\end{equation*}
where we have used that $V'\geq 0$ and $y\leq 0$ in this region.

In region $\mathcal{A}_{1,2}=\left\{(r_{1})^{*}\leq r^{*}\leq 2(r_{e})^{*}\right\}$ (where $r_{1}$ is sufficiently small, but not yet determined), we have $y(r^{*})=-(r^{*})^{2}$ and so $y'>0$  and, in view of \eqref{V1}, $-yV=(r^{*})^{2}V\sim \Lambda \frac{1}{(r^{*})^{2}}(r^{*})^{2}-\frac{1}{(r^{*})^{3}}(r^{*})^{2}= \Lambda -\frac{1}{r^{*}}$. Therefore,  in region $\mathcal{A}_{1,2}$ we have $-(yV)'\sim\frac{1}{(r^{*})^{2}}$ or, in other words, there exists a constant $c_{1}>0$ which depends only on $M$ such that $-(yV)'\geq \frac{c_{1}}{(r^{*})^{2}}$ (note that we could have instead chosen $y=\frac{1}{V\cdot r^{*}}$). The only term that remains to be understood is $-\frac{1}{2}h''$. For this reason we construct $h$ such that $h''$ is small in an appropriate sense. Recall that $r=M$ is located at $r^{*}=-\infty$. Note that for any fixed constant $q>0$ the integral $\int_{-\infty}^{2(r_{e})^{*}}\!-\frac{q}{r^{*}}=+\infty$ (recall that $r^{*}$ is negative in $\mathcal{A}$). Therefore, there exists a value $(r_{1})^{*}$ of $r^{*}$ which depends only on $M$ and $q$ such that $\int_{(r_{1})^{*}}^{2(r_{e})^{*}}-\frac{q}{r^{*}}>h_{\text{max}}$. Hence, we can take $h$ such that $h=0$ for $r^{*}\leq (r_{1})^{*}$ and $h'\sim -\frac{q}{r^{*}}$ and $\frac{1}{2}h''\leq \frac{q}{(r^{*})^{2}}$ for $r^{*}\geq (r_{1})^{*}$ (and $h$ smooth in $\mathcal{A}$). Clearly, it suffices to choose $q=\frac{c_{1}}{4}$ and thus $(r_{1})^{*}$ (and thus $r_{1}$) is also now determined (and depends only on $M$). Then 
 \begin{equation*}
\begin{split}
-\frac{1}{2}h''-(yV)'\geq c_{1}\frac{3}{4(r^{*})^{2}}> 0.
\end{split}
\end{equation*}
 As regards the term $h\cdot (V-\omega^{2})$, since $h\geq 0$ it suffices to consider $\omega_{0}$ such that $V-\omega_{0}^{2}\geq 0$ in region $\mathcal{A}_{1,2}$.

In region $\mathcal{A}_{1,3}=\left\{M\leq r\leq r_{1}\right\}$, we have $h=0$ and $y=-((r_{1})^{*})^{2}$. Then, $-yV'\geq 0$.

In region $\mathcal{B}_{1,1}=\left\{R_{e}\leq r\leq 2R_{e}\right\}$, $h=h_{\text{max}}$, $y(R_{e})=0, y(2R_{e})=(2R_{e})^{2}$ and $y$ is linear with respect to the $r$ variable. Therefore, $y'=\frac{dr}{dr^{*}}\frac{dy}{dr}\leq \frac{dy}{dr}=4R_{e}$. Note now that if we choose $\omega_{0}$ such that $V-\omega_{0}^{2}\geq \frac{V}{2}$ in region $\mathcal{B}_{1,1}$, then 
 \begin{equation*}
\begin{split}
h\cdot(V-\omega^{2})-yV'-y'V\geq h\frac{V}{2}-y'V=\left(\frac{h}{2}-y'\right)\cdot V\geq 0,
\end{split}
\end{equation*}
where we have used that $V'\leq 0$ in this region.

In region $\mathcal{B}_{1,2}=\left\{2R_{e}\leq r\leq \tilde{R}_{1}\right\}$ (where $\tilde{R}_{1}$ is sufficiently large, but not yet determined), we have $y(r)=r^{2}$ and so $y'>0$  and  $yV=r^{2}V\sim \Lambda \frac{1}{r^{2}}r^{2}+\frac{1}{r^{3}}r^{2}\sim \Lambda +\frac{1}{r}$. Therefore, in $\mathcal{B}_{1,2}$ we have $(yV)'\sim-\frac{1}{r^{2}}$ or, in other words, there exists a constant $c_{0}>0$  which depends only on $M$ such that $-(yV)'\geq \frac{c_{0}}{r^{2}}$. Now we can apply exactly the same argument as we did in region $\mathcal{A}_{1,2}$ to construct  $h$ in $\mathcal{B}_{1,2}$ such that $h\geq 0$ and $\frac{1}{2}h''\leq \frac{q_{0}}{(r^{*})^{2}}$ and $h(\tilde{R}_{1})=h'(\tilde{R}_{1})=0$ for any $q_{0}>0$, where $\tilde{R}_{1}$ depends on $q_{0}$. Clearly, we will choose $q_{0}=\frac{c_{0}}{4}$ and thus $\tilde{R}_{1}$ is also now determined. Then 
 \begin{equation*}
\begin{split}
-\frac{1}{2}h''-(yV)'\geq c_{0}\left(\frac{1}{r^{2}}-\frac{1}{4(r^{*})^{2}}\right)> 0,
\end{split}
\end{equation*}
where we use that for $r\geq R_{e}$ we have $r^{*}\geq r$. As regards the term $h\cdot (V-\omega^{2})$, since $h\geq 0$ it suffices to consider $\omega_{0}$ such that $V-\omega_{0}^{2}\geq 0$.

Without loss of generality, we may assume that $\left(\tilde{R}_{1}\right)^{2}<((r_{1})^{*})^{2}$. Hence, there exists a constant $R_{1}$  such that $R_{1}^{2}=((r_{1})^{*})^{2}$. Consider now the region $\mathcal{B}_{1,3}=\left\{\tilde{R}_{1}\leq r\leq R_{1}\right\}$. In this region, we take $y(r)=r^{2}$ and $h=0$.  The positivity of our current is preserved under the previous choices in region $\mathcal{B}_{1,3}$.

In region $\mathcal{B}_{1,4}=\left\{r\geq R_{1}\right\}$, simply take $h=0$ and $y=((r_{1})^{*})^{2}$. Then the non-negativity of this current is a consequence of the fact that $V'\leq 0$ in this region.

Regarding the precise choice for $\omega_{0}$, it suffices to take it such that $V-\omega_{0}^{2}\geq \frac{V}{2}$ for $r_{1}\leq r\leq R_{1}$. Note that both $r_{1}$ and $R_{1}$ depend only on $M$ and thus $\omega_{0}$ depends only on $M$ too. Let $\omega_{0}$ be now chosen.

By \textbf{rescaling} the functions $y,h$, we may assume that $y=-1$ in region $\mathcal{A}_{1,3}$ and $y=+1$ in region $\mathcal{B}_{1,4}$.

 By integrating $\left(\mathcal{J}_{2}^{h}[u]+\mathcal{J}^{y}_{1}[u]\right)'$ we obtain 
 \begin{equation}
\begin{split}
b(\lambda_{1})\int_{\mathcal{M}}\bigg[|u'|^{2}+(1+\omega^{2}+\Lambda)|u|^{2}\bigg]dr^{*}\leq & \bigg(\mathcal{J}^{y=1}_{1}[u]\bigg)(B)-\bigg(\mathcal{J}^{y=-1}_{1}[u]\bigg)(A)+\int_{A}^{B}E(H)dr^{*}.
 \end{split}
 \label{ef11}
\end{equation}
The above estimate holds for all $A\leq (r_{1})^{*}$ and $B\geq (R_{1})^{*}$.

\begin{remark}
The above multipliers $h,y$ are chosen to be independent of $\omega^{2}, \Lambda$ and, in fact, they do not even depend on $\omega_{0}, \lambda_{1}$. In fact, the dependence of $b$ on $\lambda_{1}$ in \eqref{ef11} may be removed. 

\label{remarkbounded1}
\end{remark}

\subsubsection{The Non-Stationary Range $\mathcal{F}_{1,2}$}
\label{sec:TheNonStationaryCase}

This frequency range is defined by 
 \begin{equation*}
\begin{split}
\omega_{0}^{2}<\omega^{2}\leq \omega_{1}^{2}, \ \ \ \Lambda\leq \lambda_{1}.
 \end{split}
\end{equation*}
The construction we will present is quite general and depends only on the asymptotic behaviour of $V$ and its positivity. Observe that 
 \begin{equation*}
\begin{split}
\left(\mathcal{J}_{1}^{y}[u]\right)'=y'|u'|^{2}+\omega^{2}y'|u|^{2}-(yV)'|u|^{2}+E(H).
 \end{split}
\end{equation*}
If we take $y$ such that $y'>0$ and $y$ bounded at the ends, then it remains to estimate the term $-(yV)'|u|^{2}$. In view of the fact that we do not require $\left(\mathcal{J}_{1}^{y}[u]\right)'$ to be positive-definite pointwise but only after integration, we can in fact apply integration by parts to this term. Then by Cauchy-Schwarz and an appropriate choice for $y$, we can indeed bound this term using the remaining two ``good'' terms. However, this $y$ would not be sufficiently flat for large $r$ (something required later on). For this reason, we can decompose $V$ into a flat and non-flat part for large $r$ and apply the above integration by parts for the flat part. This approach, however, will generate error terms for large $r$ that require coupling with Proposition \ref{larger}. In order to avoid this we apply the above integration by parts \textit{in a finite $r$-interval} for $V$ itself. The main problem will then be to estimate the arising boundary terms.

First note that we  have
 \begin{equation}
\begin{split}
\frac{\Delta^{3/2}}{r^{6}}\leq V\leq C\frac{\Delta}{r^{2}}\frac{\lambda_{1}}{r^{2}},
 \end{split}
 \label{v2}
\end{equation}
for all frequencies in $\mathcal{F}_{1,2}$ and all $r\geq M$.  Moreover, 
 \begin{equation*}
\begin{split}
0\leq \int_{r^{*}_{0}}^{+\infty}Vdr^{*}\leq  \int_{r_{0}}^{+\infty}V\frac{2r^{2}}{\Delta}dr\leq\int_{r_{0}}^{+\infty}C\frac{\lambda_{1}}{r^{2}}dr=B\frac{\lambda_{1}}{r_{0}}\leq C\frac{\lambda_{1}}{M},
 \end{split}
\end{equation*}
for all $r_{0}\geq M$. 

The main idea is to consider general constants $A_{1},R_{2}$ and apply the following integration by parts
 \begin{equation*}
\begin{split}
\int_{(A_{1})^{*}}^{(R_{2})^{*}}-\left(yV\right)'|u|^{2}dr^{*}&=-yV|u|^{2}((R_{2})^{*})+yV|u|^{2}((A_{1})^{*})+\int_{(A_{1})^{*}}^{(R_{2})^{*}}yV\left[u\bar{u}'+u'\bar{u}\right]\\
&\geq-yV|u|^{2}((R_{2})^{*})+yV|u|^{2}((A_{1})^{*})+ \int_{(A_{1})^{*}}^{(R_{2})^{*}}\bigg[-\frac{1}{2}y'|u'|^{2}-\frac{4y^{2}V^{2}}{y'}|u|^{2}\bigg].\\
 \end{split}
\end{equation*}
 In view of the lower bound $\omega^{2}>\omega_{0}^{2}$, we can take $y$ such that  $y=y_{\text{exp}}$ where
 \begin{equation*}
\begin{split}
\frac{9y_{\text{exp}}^{2}V^{2}}{y_{\text{exp}}'}=\omega_{0}^{2}y_{\text{exp}}' \Longrightarrow y_{\text{exp}}=\frac{1}{2}e^{-3\omega_{0}^{-1}\int_{r^{*}}^{+\infty}Vdr^{*}}
 \end{split}
\end{equation*}
in region $\left\{A_{1}\leq r\leq R_{2}\right\}$.
\begin{figure}[H]
	\centering
		\includegraphics[scale=0.14]{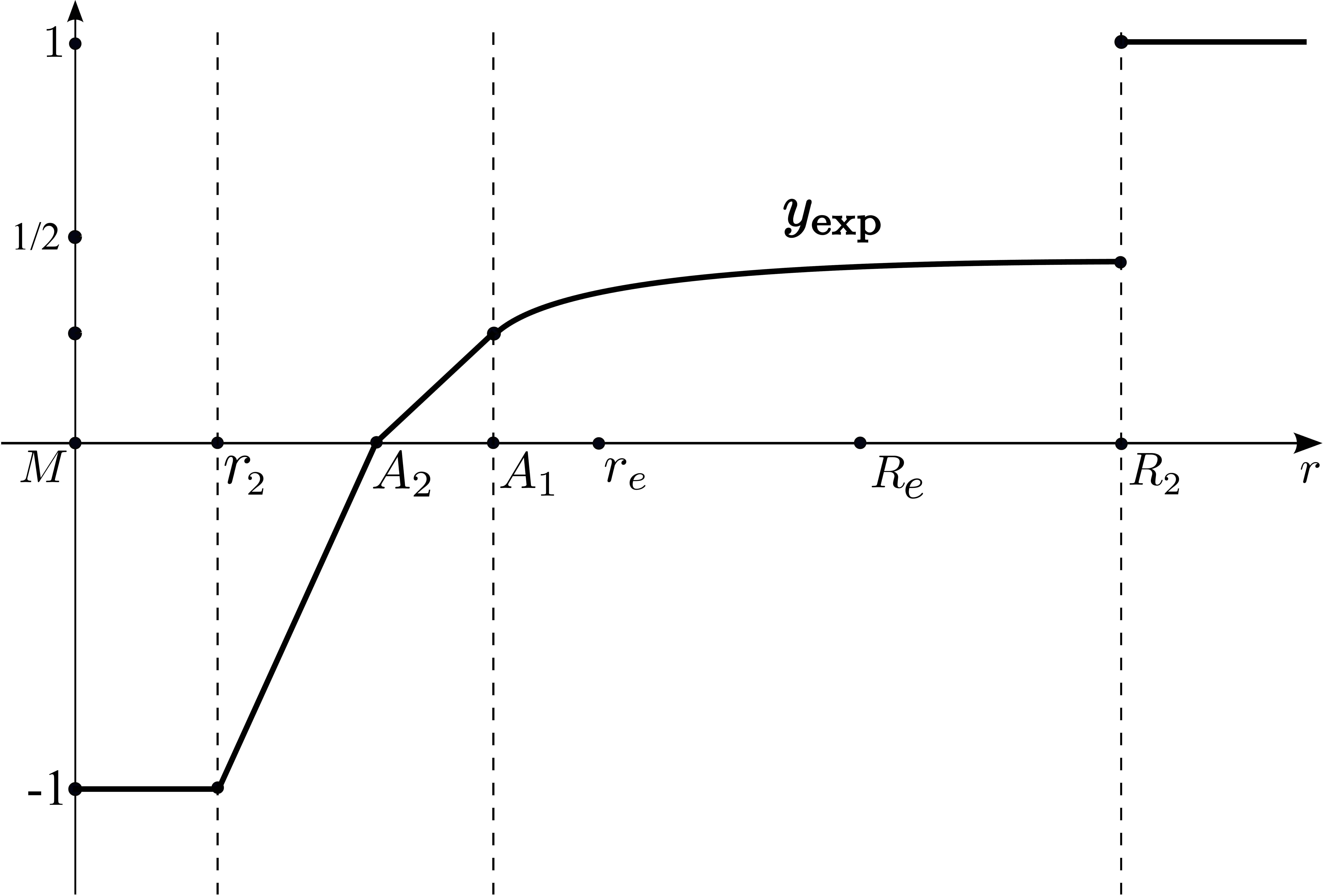}
	\label{fig:exp2}
\end{figure}

In view of the previous estimate for the integral of $V$, we have 
 \begin{equation*}
\begin{split}
\frac{1}{2}e^{-3\omega_{0}^{-1}C\frac{\lambda_{1}}{M}}\leq y_{\text{exp}}\leq \frac{1}{2},
 \end{split}
\end{equation*}
for all $A_{1}\leq r\leq R_{2}$. Moreover, in the same region we have
 \begin{equation*}
\begin{split}
y_{\text{exp}}'=\frac{3}{2}\omega_{0}^{-1}Ve^{-3\omega_{0}^{-1}\int_{r^{*}}^{+\infty}Vdr^{*}}.
 \end{split}
\end{equation*}
Therefore, 
 \begin{equation*}
\begin{split}
& \bigg(\mathcal{J}_{1}^{y_{\text{exp}}}[u]\bigg)((R_{2})^{*})-\bigg(\mathcal{J}_{1}^{y_{\text{exp}}}[u]\bigg)((A_{1})^{*})=\int_{(A_{1})^{*}}^{(R_{2})^{*}}\bigg(\mathcal{J}_{1}^{y_{\text{exp}}}[u]\bigg)'dr^{*}\\
=&\int_{(A_{1})^{*}}^{(R_{2})^{*}}E(H)+y_{\text{exp}}'\left|u'\right|^{2}+y_{\text{exp}}'\omega^{2}\left|u\right|^{2}+\int_{(A_{1})^{*}}^{(R_{2})^{*}}-(y_{\text{exp}}V)'\left|u\right|^{2}\\
\geq  &\int_{(A_{1})^{*}}^{(R_{2})^{*}}\bigg[E(H)+y_{\text{exp}}'\left|u'\right|^{2}+y_{\text{exp}}'\omega^{2}\left|u\right|^{2}\bigg]+\int_{(A_{1})^{*}}^{(R_{2})^{*}}\bigg[-\frac{1}{2}y_{\text{exp}}'\left|u'\right|^{2}-\frac{4}{9}y_{\text{exp}}'\omega_{0}^{2}|u|^{2}\bigg]\\&-y_{\text{exp}}V|u|^{2}((R_{2})^{*})+y_{\text{exp}}V|u|^{2}((A_{1})^{*})\\
\geq & \int_{(A_{1})^{*}}^{(R_{2})^{*}}\bigg[E(H)+\frac{1}{2}y_{\text{exp}}'\left|u'\right|^{2}+\frac{5}{9}y_{\text{exp}}'\omega_{0}^{2}\left|u\right|^{2}\bigg]-y_{\text{exp}}V|u|^{2}((R_{2})^{*}).
 \end{split}
\end{equation*}
Recalling  the expression for the current $\mathcal{J}_{1}^{y_{\text{exp}}}[u]$ we obtain
 \begin{equation*}
\begin{split}
\int_{(A_{1})^{*}}^{(R_{2})^{*}}&\bigg[E(H)+\frac{1}{2}y_{\text{exp}}'\left|u'\right|^{2}+\frac{5}{9}y_{\text{exp}}'\omega_{0}^{2}\left|u\right|^{2}\bigg]\\&\ \ \ \  \leq y_{\text{exp}}\bigg[|u'|^{2}+\omega^{2}|u|^{2}-V|u|^{2}\bigg]((R_{2})^{*})+y_{\text{exp}}V|u|^{2}((R_{2})^{*})-\bigg(\mathcal{J}_{1}^{y_{\text{exp}}}[u]\bigg)((A_{1})^{*})\\&\ \ \ \ = y_{\text{exp}}\bigg[|u'|^{2}+\omega^{2}|u|^{2}\bigg]((R_{2})^{*})-\bigg(\mathcal{J}_{1}^{y_{\text{exp}}}[u]\bigg)((A_{1})^{*}).
 \end{split}
\end{equation*}
Since $y_{\text{exp}}\leq \frac{1}{2}$, we have
\begin{equation}
\begin{split}
\int_{(A_{1})^{*}}^{(R_{2})^{*}}&\bigg[E(H)+\frac{1}{2}y_{\text{exp}}'\left|u'\right|^{2}+\frac{5}{9}y_{\text{exp}}'\omega_{0}^{2}\left|u\right|^{2}\bigg]\\&\ \ \ \  \leq \bigg[\frac{1}{2}|u'|^{2}+\frac{1}{2}\omega^{2}|u|^{2}\bigg]((R_{2})^{*})-\bigg(\mathcal{J}_{1}^{y_{\text{exp}}}[u]\bigg)((A_{1})^{*}).
 \end{split}
\label{na1ineq}
\end{equation}

Extend now the multiplier $y_{\text{exp}}=1$  for $r\geq R_{2}$. (Note that $y_{\text{exp}}$ is not continuous at $r=R_{2}$).  Then for any $B\geq R_{2}$ we have
 \begin{equation*}
\begin{split}
\bigg(\mathcal{J}_{1}^{y=1}[u]\bigg)(B^{*})-\bigg(\mathcal{J}_{1}^{y=1}[u]\bigg)((R_{2})^{*})=&\int_{(R_{2})^{*}}^{B^{*}}\bigg(\mathcal{J}_{1}^{y=1}[u]\bigg)'dr^{*}\\ =&
\int_{(R_{2})^{*}}^{B^{*}}\bigg[E(H)-V'|u|^{2}\bigg].
 \end{split}
\end{equation*}
Recall that $V'\leq 0$ for $r\geq R_{2}$. Therefore, 
\begin{equation}
\int_{(R_{2})^{*}}^{B^{*}}\bigg[E(H)-V'|u|^{2}\bigg]=\bigg(\mathcal{J}_{1}^{y=1}[u]\bigg)(B^{*})-\bigg[|u'|^{2}+\omega^{2}|u|^{2}-V|u|^{2}\bigg]((R_{2})^{*}).
\label{na2ineq}
\end{equation}
By adding \eqref{na1ineq} and \eqref{na2ineq} we obtain
 \begin{equation*}
\begin{split}
\int_{(A_{1})^{*}}^{(R_{2})^{*}}&\bigg[E(H)+\frac{1}{2}y_{\text{exp}}'\left|u'\right|^{2}+\frac{5}{9}y_{\text{exp}}'\omega_{0}^{2}\left|u\right|^{2}\bigg]+\int_{(R_{2})^{*}}^{B^{*}}\bigg[E(H)-V'|u|^{2}\bigg]\\
\ \ \leq & \bigg(\mathcal{J}_{1}^{y=1}[u]\bigg)(B^{*})-\bigg(\mathcal{J}_{1}^{y_{\text{exp}}}[u]\bigg)((A_{1})^{*})-\bigg[\frac{1}{2}|u'|^{2}+\frac{1}{2}\omega^{2}|u|^{2}-V|u|^{2}\bigg]((R_{2})^{*}) \\
\ \ \leq & \bigg(\mathcal{J}_{1}^{y=1}[u]\bigg)(B^{*})-\bigg(\mathcal{J}_{1}^{y_{\text{exp}}}[u]\bigg)((A_{1})^{*})-\bigg[\frac{1}{2}|u'|^{2}+\left(\frac{1}{2}\omega_{0}^{2}-V\right)|u|^{2}\bigg]((R_{2})^{*}).
 \end{split}
\end{equation*}
In view of \eqref{v2}, we can chose $R_{2}$ such that $\left(\frac{1}{2}\omega_{0}^{2}-V\right)((R_{2})^{*})>0$. Note that $R_{2}$ depends only on $M$ and $\lambda_{1}$ (recall that $\omega_{0}$ was determined in Section \ref{sec:TheNearStationaryCase}). Therefore,
 \begin{equation}
\begin{split}
\int_{(A_{1})^{*}}^{(R_{2})^{*}}&\bigg[E(H)+\frac{1}{2}y_{\text{exp}}'\left|u'\right|^{2}+\frac{5}{9}y_{\text{exp}}'\omega_{0}^{2}\left|u\right|^{2}\bigg]+\int_{(R_{2})^{*}}^{B^{*}}\bigg[E(H)-V'|u|^{2}\bigg] \\
\ \ \leq & \bigg(\mathcal{J}_{1}^{y=1}[u]\bigg)(B^{*})-\bigg(\mathcal{J}_{1}^{y_{\text{exp}}}[u]\bigg)((A_{1})^{*}),
 \end{split}
 \label{na3ineq}
\end{equation}
for all $B\geq R_{2}$. Note that although the multiplier $y_{\text{exp}}$ that we constructed above is not continuous at $r=R_{2}$, the difference of  the corresponding boundary terms at $r=R_{2}$ has the right sign (provided we take $R_{2}$ sufficiently large).

We now look at the region $\left\{M\leq r\leq A_{1}\right\}$, where 
\begin{equation*}
\begin{split}
&0\leq V\leq C(\lambda_{1})\Delta, \\
&0\leq V'\leq C(\lambda_{1})\Delta^{\frac{3}{2}}.
\end{split}
\end{equation*}
Then,
\begin{equation*}
\begin{split}
\omega^{2}y'-(yV)'=\omega^{2}y'-yV'-y'V\geq \left[\omega^{2}_{0}-C(\lambda_{1})\cdot\Delta\right]y'-y\cdot C(\lambda_{1})\Delta^{\frac{3}{2}}.
\end{split}
\end{equation*}
We now choose $A_{1}$ such that for all $M\leq r\leq A_{1}$ we have $\omega_{0}^{2}-C(\lambda_{1})\Delta\geq \frac{\omega_{0}^{2}}{2}$. Since $\omega_{0}$ has been determined, $A_{1}$ depends only on $M$ and $\lambda_{1}$. In particular, $A_{1}<r_{1}$, where $r_{1}$ is as defined in Section \ref{sec:TheNearStationaryCase}. Hence, for all $M\leq r\leq A_{1}$, the term $ \left[\omega^{2}_{0}-C(\lambda_{1})\cdot\Delta\right]y'$ is positive as long as $y$ is increasing.

Moreover, there exists $A_{2}<A_{1}$ such that if $y(A_{2})=0$ and $y$ is linear with respect to $r^{*}$ then $y'>\frac{2}{\omega_{0}^{2}}C(\lambda_{1})\Delta^{\frac{3}{2}}$ in $[A_{2},A_{1}]$ (clearly $y\leq 1/2$ in this region).

We now consider  $r_{2}<A_{2}$ such that $y(r_{2})=-1$ and $y$ is linear with respect to $r^{*}$ in $[r_{2},A_{2}]$. Then $y'>0$ and $y\leq 0$. Finally, in region $\left\{M\leq r\leq r_{2}\right\}$ we simply choose $y=-1$. Then $y'=0$.  The positivity of the current follows  in view of the fact that $V'\geq 0$ in these regions.

  Therefore, by integrating $\left(\mathcal{J}^{y_{\text{exp}}}_{1}[u]\right)'$ and using \eqref{na3ineq} we obtain 
 \begin{equation}
\begin{split}
b(\lambda_{1},\omega_{1})\int_{\mathcal{M}}\bigg[|u'|^{2}+(1+\omega^{2}+\Lambda)|u|^{2}\bigg]dr^{*}\leq & \left(\mathcal{J}^{y=1}_{1}[u]\right)(B)-\left(\mathcal{J}^{y=-1}_{1}[u]\right)(A)\\ 
&+\int_{A}^{B}E(H)dr^{*}.\\
 \end{split}
 \label{ef12}
\end{equation}
The above estimate holds for all $A\leq r_{2}$ and $B\geq R_{2}$.

\subsection{The Unbounded Frequency Range $\mathcal{F}_{2}$}
\label{sec:TheUnboundedFrequencyRangeMathcalF2}

In view of the unboundedness of some of the coefficients in $\Pi$ defined by \eqref{pi}, we will have to be careful  and bound the ``correct'' expressions.

\subsubsection{The Trapped Frequencies $\mathcal{F}_{2,1}$}
\label{sec:TheTrappedFrequenciesMathcalF21}

In this frequency range we have both $\Lambda$ and $\omega$ being unbounded but such that $\Lambda\sim\omega^{2}$. Therefore, we need to bound the quantity 
 \begin{equation*}
\begin{split}
\Pi_{2,1}=|u'|^{2}+\Lambda|u|^{2}.
 \end{split}
\end{equation*}
In view of \eqref{V}, it suffices to estimate 
 \begin{equation*}
\begin{split}
\Pi'_{2}=|u'|^{2}+V|u|^{2}.
 \end{split}
\end{equation*}
However, note that (in view of \eqref{e3iswsntouu}) $V$ always appears in the expression $V-\omega^{2}$ and this expression does not have a sign in the frequency range $\mathcal{F}_{2,1}$. Therefore, we can not expect to bound $\Pi_{2}'$. Recall now that $\Lambda$ appears in $V'$, which, however, degenerates at $r=(1+\sqrt{2})M$. This implies that the best we could expect is to derive an estimate for $\Pi$ which degenerates exactly at $r=(1+\sqrt{2})M$. Recall that the derivative of the current $\mathcal{J}_{3}^{f}$ does not involve $V-\omega^{2}$. Indeed,
 \begin{equation*}
\begin{split}
\left(\mathcal{J}_{3}^{f}[u]\right)'=2f'|u'|^{2}+\left(-fV'-\frac{1}{2}f'''\right)|u|^{2}+E(H).
 \end{split}
\end{equation*}

We consider the function $f_{\text{trap}}\in C^{3}[M,+\infty]$ such that $-1\leq f_{\text{trap}}\leq 1$ and $f_{\text{trap}}'\geq 0$ and moreover
\begin{enumerate}
	\item in region $\mathcal{A}_{1}=\left\{M\leq r\leq \frac{r_{e}+M}{2}\right\}$ we have $f_{\text{trap}}=-1$,
	\item in region $\mathcal{A}_{2}=\left\{\frac{r_{e}+M}{2}\leq r\leq r_{e}\right\}$ we have $f_{\text{trap}}\leq -\frac{1}{2}$,
	\item in region $\mathcal{M}=\left\{r_{e}\leq r\leq R_{e}\right\}$ we have $f_{\text{trap}}(M+\sqrt{2}M)=0$, $f_{\text{trap}}'\geq c_{1}>0$ and $-f_{\text{trap}}'''\geq c_{2}>0$ (for the last condition we simply take $f_{\text{trap}}'$ to be stricly positive and strictly concave),
	\item in region $\mathcal{B}_{1,1}=\left\{R_{e}\leq r\leq R_{e}+1\right\}$ we have $f_{\text{trap}}\geq \frac{1}{2}$,
	\item in  region $\mathcal{B}_{1,2}=\left\{R_{e}+1\leq r\right\}$ we have $f_{\text{trap}}=1$. 
\end{enumerate}
\begin{figure}[H]
	\centering
		\includegraphics[scale=0.14]{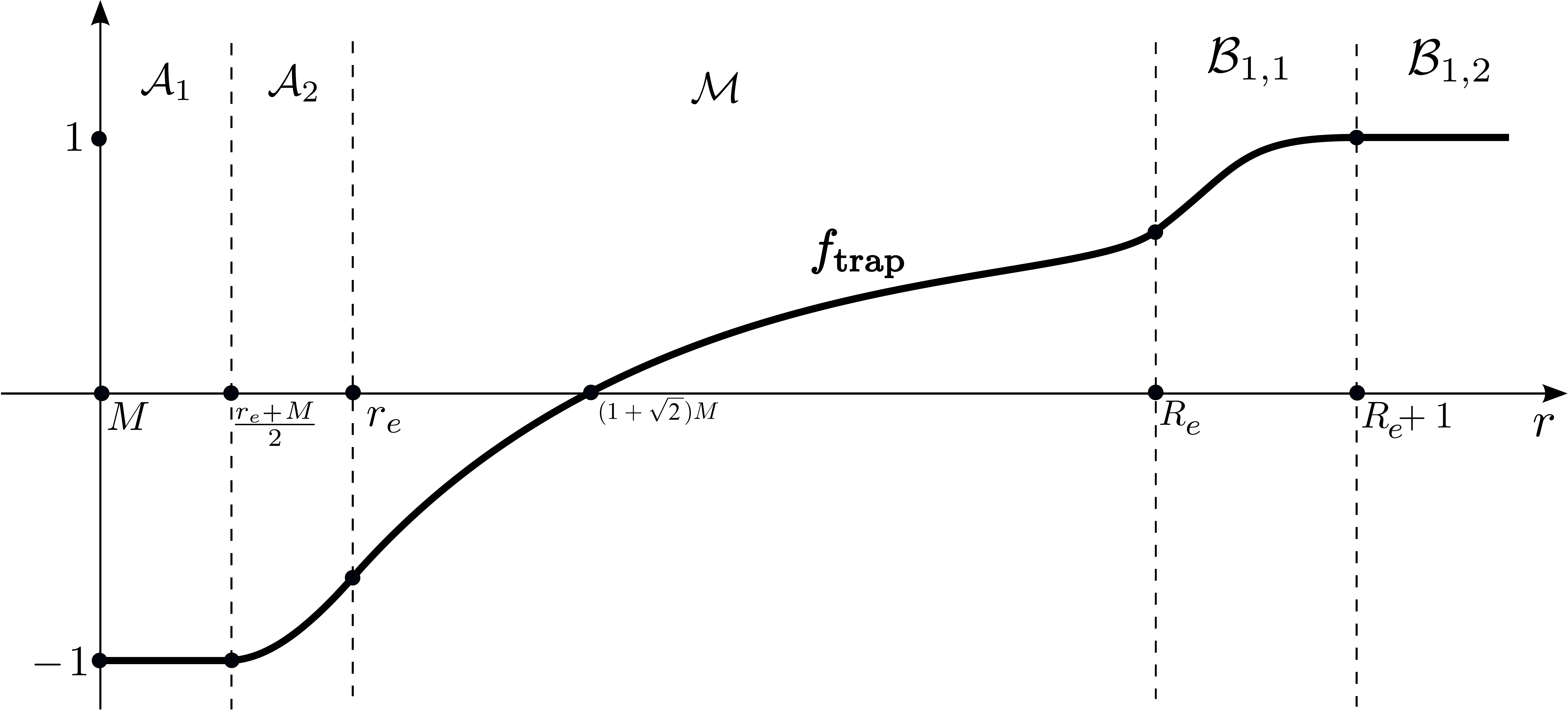}
	\label{fig:ftrap}
\end{figure}
Since $f'\geq 0$ and, in particular, $f'\geq c_{1}>0$ in $\mathcal{M}$ the term $2f'|u'|^{2}$ has always the right sign. It remains to understand the term $\left(-fV'-\frac{1}{2}f'''\right)$. In region $\mathcal{M}$ we have 
 \begin{equation*}
\begin{split}
\left(-fV'-\frac{1}{2}f'''\right)\sim \left(r-(1+\sqrt{2})M\right)^{2}\Lambda+1.
 \end{split}
\end{equation*}
In the unbounded region $\mathcal{B}_{1,2}$ we have $\left(-fV'-\frac{1}{2}f'''\right)=-V'\geq 0$. Finally, regarding the two intermediate regions $\mathcal{A}_{2}, \mathcal{B}_{1,1}$, in view of \eqref{vtonoskonta} and \eqref{vtonosmarkua} we have 
 \begin{equation*}
\begin{split}
\left(-fV'-\frac{1}{2}f'''\right)\sim \Lambda-\frac{1}{2}f'''. 
 \end{split}
\end{equation*}
Since these two regions are compact, we can consider $\Lambda$ (and thus $\omega_{1}$) large enough such that $\Lambda-\frac{1}{2}f'''\sim\Lambda$.

By integrating $\left(\mathcal{J}_{3}^{f_{\text{trap}}}[u]\right)'$ we obtain
\begin{equation}
\begin{split}
b\!\int_{\mathcal{M}}\left[|u'|^{2}+|u|^{2}+\left(r-(1+\sqrt{2})M\right)^{2}\left[\Lambda+\omega^{2}\right]|u|^{2}\right]\leq &
\left(\mathcal{J}_{1}^{y=1}[u]\right)\!(B)-\left(\mathcal{J}_{1}^{y=-1}[u]\right)\!(A)\\
\\&\ \ \ +\int_{A}^{B}E(H)dr^{*}.
\label{trapping}
\end{split}
\end{equation}
The above estimate holds for all $A\leq (\frac{r_{e}+M}{2})$ and $B\geq (R_{e}+1)$.

\begin{remark}
As we have shown above, the trapping properties in the high frequency limit of the wave equation are very closely related to the zeros of $V'$ in the region $r>r_{\hh}$. We have the following situation:

1. For $|a|\ll M$ and no symmetry restriction, we have that $V'$ has a unique zero in this frequency limit, which, however, depends on $\omega, m, \Lambda$. For fixed $\omega, m$ as $\Lambda\rightarrow+\infty$ these roots converge to a fixed value of $r$. Trapping appears for all values of $r$ which correspond to limit points of the collection of these roots.

2. For $|a|<M$ and $m=0$, $V'$ has a unique root which depends only on $\Lambda$. Trapping then appears only on the (unique) limit of these roots as $\Lambda\rightarrow+\infty$.  

3. For $|a|<M$ and no symmetry restriction, the situation is similar as in 1. However, in this case, it suffices to further restrict to non-superradiant frequencies, since as is shown in \cite{megalaa} the superradiant frequencies are non-trapped. Note that in \cite{mikraa}, one did not need to decompose in superradiant and non-superradiant frequencies. 

4. For $|a|=M$ and $m=0$, which corresponds to the case of this paper, we have that $V'$ has a unique root and this root does not depend on any frequency parameter (not even on $\Lambda$). Note that this behaviour of trapping on extreme black holes was reflected in \cite{aretakis2} using physical space methods. 
\end{remark}

\subsubsection{The Angular Dominated Frequencies $\mathcal{F}_{2,2}$}
\label{sec:TheAngularDominatedFrequencyRangeF2}

Recall that this frequency range is defined by 
		\[\mathcal{F}_{2,2}=\left\{(\omega, \Lambda):\left|\omega\right|\leq \omega_{1}, \Lambda >\lambda_{1} \right\}\cup\left\{(\omega, \Lambda):\left|\omega\right|> \omega_{1}, \Lambda>\lambda_{2}\omega^{2} \right\}.\]

In view of the  dominance of $\Lambda$, we need to derive an estimate for the quantity 

 \begin{equation*}
\begin{split}
\Pi_{2}=|u'|^{2}+\Lambda|u|^{2}.
 \end{split}
\end{equation*}
It is only then that we can also estimate \eqref{pi}. In fact, in view of \eqref{V}, $V$ behaves like $\Lambda$ in $\mathcal{M}$ and so  it suffices to estimate $\Pi'_{2}=|u'|^{2}+V|u|^{2}$ in this region.

Previously we derived estimates for any $\lambda_{1}>0$, which means that we can here allow $\lambda_{1}$ to be as large as we want. This is convenient, since in view of \eqref{V}, \eqref{vtonoskonta} and \eqref{vtonosmarkua} we can make $V$ and $V'$  as large as we want in regions where they do not degenerate. We will also consider $\lambda_{2}$ to be sufficiently large. No restriction will be imposed on $\omega_{1}$ for this range. 

As we have already mentioned, $V$ only appears in the expression $V-\omega^{2}$. In the frequency ranges under consideration, however, we expect $V$ to dominate $\omega^{2}$ in region $\mathcal{M}$. The current which contains $V-\omega^{2}$ in its  derivative is $\mathcal{J}_{2}^{h}[u]$, and therefore, one possible choice would be to simply take $h=1$. However, in order to estimate the boundary terms we need  to have $h=0$ for say $r\geq R+1$. In other words, we need to cut off this $h$. This cut-off will create error terms (originating from $-h''$) which need to be estimated by coupling $\mathcal{J}_{2}^{h}[u]$ with another current\footnote{In Section \ref{sec:TheNearStationaryCase} we had a similar situation, and there we used a smallness parameter $q$ to make $h''$ small; here, however, in view of the largeness of $\Lambda$ we may avoid that construction.}. We can not use $\mathcal{J}_{1}^{y}[u]$, because $V-\omega^{2}$ appears with the wrong sign. Therefore, the best way to estimate these terms is by using the current $\mathcal{J}_{3}^{f}[u]$. In fact, we can take $f=f_{\text{trap}}$ defined in Section \ref{sec:TheTrappedFrequenciesMathcalF21}. So we consider the  current $\mathcal{J}^{\text{ang}}[u]=\mathcal{J}_{2}^{h_{\text{cut}}}[u]+\mathcal{J}_{3}^{f_{\text{trap}}}[u]$ where $h_{\text{cut}}$ is such that $0\leq h_{\text{cut}}\leq 1$, $h_{\text{cut}}=0$ for $r\leq \frac{r_{e}+M}{2}$, $h_{\text{cut}}=1$ in region $\mathcal{M}$ and $h_{\text{cut}}=0$ for $r\geq R+1$.  

\begin{figure}[H]
	\centering
		\includegraphics[scale=0.14]{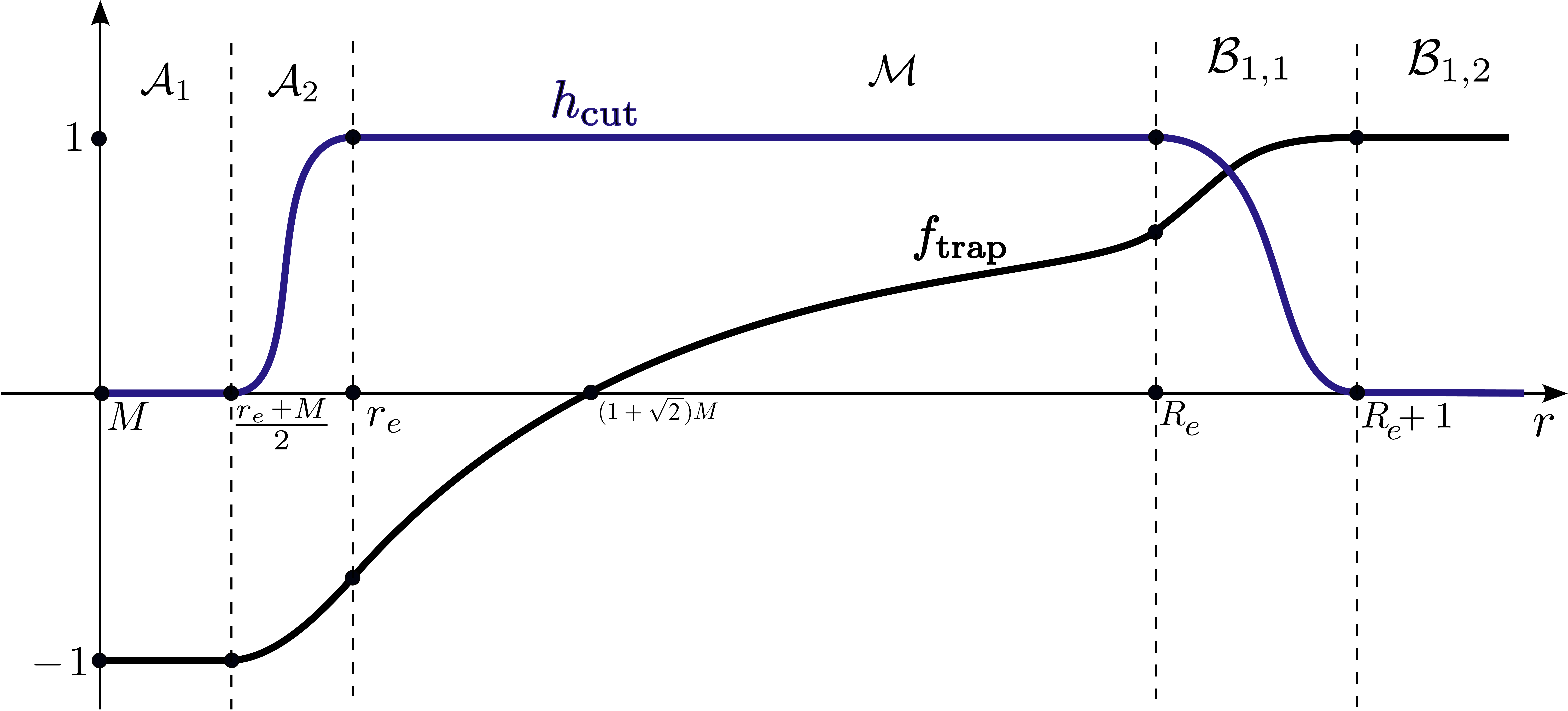}
	\label{fig:hcutftrap}
\end{figure}
Recall now that 
 \begin{equation*}
\begin{split}
\left(\mathcal{J}_{2}^{h}[u]+\mathcal{J}_{3}^{f}[u]\right)'=(2f'+h)|u'|^{2}+\left[h\cdot(V-\omega^{2})-\frac{1}{2}h''-fV'-\frac{1}{2}f'''\right]|u|^{2}+E(H),
 \end{split}
\end{equation*}
First note that $2f_{\text{trap}}'+h_{\text{cut}}$ is non-negative and, moreover, greater than 1 in $\mathcal{M}$. Furthermore, in the same region, if $\lambda_{1}, \lambda_{2}$ are sufficiently large, then $V-\omega^{2}\sim V\sim\Lambda$. By construction, we also have $-f_{\text{trap}}V'-\frac{1}{2}f_{\text{trap}}'''>0$ in $\mathcal{M}$. 

In regions $\mathcal{A}_{2}, \mathcal{B}_{1,2}$ depicted in the figure above, for sufficiently large $\lambda_{1}, \lambda_{2}$ we have $V-\omega^{2}>0$ and 
\begin{equation*}
-f_{\text{trap}}V'-\frac{1}{2}f_{\text{trap}}'''-\frac{1}{2}h_{\text{cut}}''\sim \Lambda-\frac{1}{2}f_{\text{trap}}'''-\frac{1}{2}h_{\text{cut}}''\geq 0.
\end{equation*}

Finally, clearly this current is non-negative definite in 
the regions  $\mathcal{A}_{1},\mathcal{B}_{1,2}$. 

Note that this construction does not explicitly use the behaviour of $V'$ in region $\mathcal{M}$. The constants $\lambda_{1},\lambda_{2}$ can now be chosen. Hence the constants $r_{2},R_{2}$ of Section \ref{sec:TheNonStationaryCase} are now determined. 

With these choices by integrating $\left(\mathcal{J}^{h_{\text{cut}}}_{2}[u]+\mathcal{J}_{3}^{f_{\text{trap}}}[u]\right)'$ we  obtain 
\begin{equation}
\begin{split}
b\int_{\mathcal{M}}\left[|u'|^{2}+(1+\omega^{2}+\Lambda)|u|^{2}\right]\leq &\left(\mathcal{J}_{1}^{y=1}[u]\right)\!(B)-\left(\mathcal{J}_{1}^{y=-1}[u]\right)\!(A)+ \int_{A}^{B}E(H)dr^{*}.
\label{angular}
\end{split}
\end{equation}
The above estimate holds for all $A\leq (\frac{r_{e}+M}{2})$ and $B\geq (R_{e}+1)$.

\subsubsection{The Time Dominated Frequencies $\mathcal{F}_{2,3}$}
\label{sec:TheTimeDominatedFrequencyRangeF4}

Recall that this frequency range is defined by 
		\[\mathcal{F}_{2,3}=\left\{(\omega, \Lambda):\left|\omega\right|> \omega_{1}, \Lambda \leq 2M^{2}\omega^{2} \right\}.\]

In view of the  dominance of $\omega^{2}$, we need to derive an estimate for the quantity 

 \begin{equation*}
\begin{split}
\Pi_{3}=|u'|^{2}+\omega^{2}|u|^{2}.
 \end{split}
\end{equation*}
It is only then that we can also estimate \eqref{pi}.

Previously we derived estimates for  sufficiently large  $\omega_{1}>0$ (we needed $\omega_{1}$ to be large in Section \ref{sec:TheTrappedFrequenciesMathcalF21}), which means that we can assume here that $\omega_{1}$ is as large as we want. 

We are looking for a current which contains in its derivative the term $+\omega^{2}$ (with this sign!). A quick inspection shows that we need to work with the current $\mathcal{J}_{1}^{y}[u]$. Indeed, recall that 
 \begin{equation*}
\begin{split}
\left(\mathcal{J}_{1}^{y}[u]\right)'=y'|u|^{2}+y'\cdot(\omega^{2}-V)|u|^{2}-yV'|u|^{2}+E(H).
 \end{split}
\end{equation*}

Since, we are clearly looking for a $y=y_{\text{time}}$ such that $y_{\text{time}}\geq 0$ (with $y_{\text{time}}\geq c_{1}>0$ in $\mathcal{M}$) and $-y_{\text{time}}V'\geq 0$ for all $r\geq M$, we can take $y_{\text{time}}=f_{\text{trap}}$, where $f_{\text{trap}}$ is as defined in Section \ref{sec:TheTrappedFrequenciesMathcalF21}. Therefore, it remains to estimate the term $\omega^{2}-V$. Using Lemma \ref{vmax} we have
 \begin{equation*}
\begin{split}
 \omega^{2}-V\geq \omega^{2}-V_{\text{max}}\geq \omega^{2}-\frac{1}{4M^{2}}\Lambda-\frac{5}{M^{4}}\geq \omega^{2}-\frac{\omega^{2}}{2}-\frac{5}{M^{4}}\geq \frac{\omega^{2}}{2}-\frac{5}{M^{4}}.
 \end{split}
\end{equation*}
Therefore, if we consider $\omega_{1}>\frac{2\sqrt{5}}{M^{2}}$ then $\omega^{2}-V\geq \frac{\omega^{2}}{4}$, for all $r\geq M$ in the frequency range $\mathcal{F}_{2,3}$.

With these choices, by integrating $\left(\mathcal{J}^{f_{\text{trap}}}_{1}[u]\right)'$ we obtain
\begin{equation}
\begin{split}
b\int_{\mathcal{M}}\left[|u'|^{2}+|u|^{2}+\Lambda|u|^{2}+\omega^{2}|u|^{2}\right]\leq &  
\left(\mathcal{J}_{1}^{y=1}[u]\right)\!(B)-\left(\mathcal{J}_{1}^{y=-1}[u]\right)\!(A)+\int_{A}^{B}E(H)dr^{*}.
\label{time}
\end{split}
\end{equation}
The above estimate holds for all $A\leq (\frac{M+r_{e}}{2})$ and $B\geq (R_{e}+1)$.

\subsection{Auxilliary Currents}
\label{sec:AuxilliaryCurrents}

Let now all the frequency parameters be determined. Let also $r_{0}=\text{min}\left\{r_{1}, r_{2}, \frac{r_{e}+M}{2}\right\}$ and $R_{0}=\text{max}\left\{R_{1},R_{2}, R_{e}+1\right\}$. Note that $r_{0}$ and $R_{0}$ are now determined constants and depend only on $M$, $r_{e}$ and $R_{e}$.

We next look at the error terms $E(H)$ which arise from the cut-off. In view of \eqref{derivatives}, $E(H)$ contains terms of the form $h\text{Re}(u\overline{H}), 2y\text{Re}(u'\overline{H})$ and $2f\text{Re}(u'\overline{H})+f'\text{Re}(u\overline{H})$. 

Since  $r\in[A,B]$ is bounded we will not worry about its powers and therefore, all the constants will depend on $M$ and $A, B$.  Moreover, note that \[\text{max}_{\left\{A\leq r\leq B\right\}}\left\{\left|y\right|,|h|,|h_{\text{cut}}|,\left|y_{\text{exp}}\right|,\left|f_{\text{trap}}\right|,\left|f'_{\text{trap}}\right|\right\}=C(A,B).\]
 Then, for every fixed pair $\omega^{2},\Lambda$ we have
 \begin{equation}
\begin{split}
E(H)\leq |E(H)|&\leq C(A,B)\left(\left|\text{Re}(u\overline{H})\right|+\left|\text{Re}(u'\overline{H})\right|\right)\leq C(A,B)\left(\left|\text{Re}(u\overline{F})\right|+\left|\text{Re}(u'\overline{F})\right|\right)\\&\leq C(A,B)\left(\epsilon\left|u\right|^{2}+\epsilon\left|u'\right|^{2}+\frac{1}{\epsilon}\left|F\right|^{2}\right),
\end{split}
\label{eh}
\end{equation}
for any $\epsilon>0$,  where we have used \eqref{hf}. Therefore, we need to recover $\epsilon$-parts\footnote{Note that this $\epsilon$ should only depend on $M$, $A$ and $B$ and not, in particular,  on the frequencies $(\omega^{2},\Lambda)$.} of $|u|^{2}$ and $|u'|^{2}$ in the whole interval $\left[A,B\right]$. Clearly, in view of \eqref{ef11}, \eqref{ef12}, \eqref{trapping}, \eqref{angular} and \eqref{time} this has already been done in region $\mathcal{M}$, and therefore,  it suffices to do so in the regions $\mathcal{R}_{1}=\left\{A\leq r\leq r_{e}\right\}$ and $\mathcal{R}_{2}=\left\{R_{e}\leq r\leq B\right\}$ (for $A\leq r_{0}$ and $B\geq R_{0}$). 

By revisiting the constructions of Section \ref{sec:FourierLocalisedEstimates}, one sees immediately  that   we have in fact obtained for all frequencies an estimate of the form
\begin{equation}
\begin{split}
&b(A,B)\int_{\left\{A\leq r\leq r_{e}\right\}}V'|u|^{2}dr^{*}+b(A,B)\int_{\left\{R_{e}\leq r\leq B\right\}}-V'|u|^{2}dr^{*}\\ &\ \ \ \ \ \ \ \ \ \leq   
\left(\mathcal{J}_{1}^{y=1}[u]\right)\!(B)-\left(\mathcal{J}_{1}^{y=-1}[u]\right)\!(A)+\int_{A}^{B}E(H)dr^{*}.
\end{split}
\label{finalfourier22}
\end{equation}
Since $V'$ does not degenerate in regions $\mathcal{R}_{1},\mathcal{R}_{2}$ and in view of \eqref{vtonoskonta} and \eqref{vtonosmarkua}, we can take for all frequencies an $\epsilon$-part of $|u|^{2}$ without any loss (where $\epsilon$ depends only on  $M$, $A$ and $B$).

However, we do not have a similar estimate for $|u'|^{2}$. That is why we introduce an auxiliary current which will allow us to obtain an $\epsilon$-part of $|u'|^{2}$ by borrowing from the remaining good terms.

This current is independent of the frequency range and is of the form $\mathcal{J}_{1}^{y}[u]$, for an appropriate $y=y_{\text{aux}}$. 
Recall that 
\begin{equation*}
\left(\mathcal{J}_{1}^{y}[u]\right)'=y'|u'|^{2}+\Big[y'(\omega^{2}-V)-yV'\Big]|u|^{2}+E^{y}[H].
\end{equation*}
We will not worry about the term $E^{y}[H]$, since we can apply again \eqref{eh}.

The prototype for $y_{\text{aux}}$ will be the piecewise linear with respect to $r^{*}$ function illustrated below 
\begin{figure}[H]
	\centering
		\includegraphics[scale=0.135]{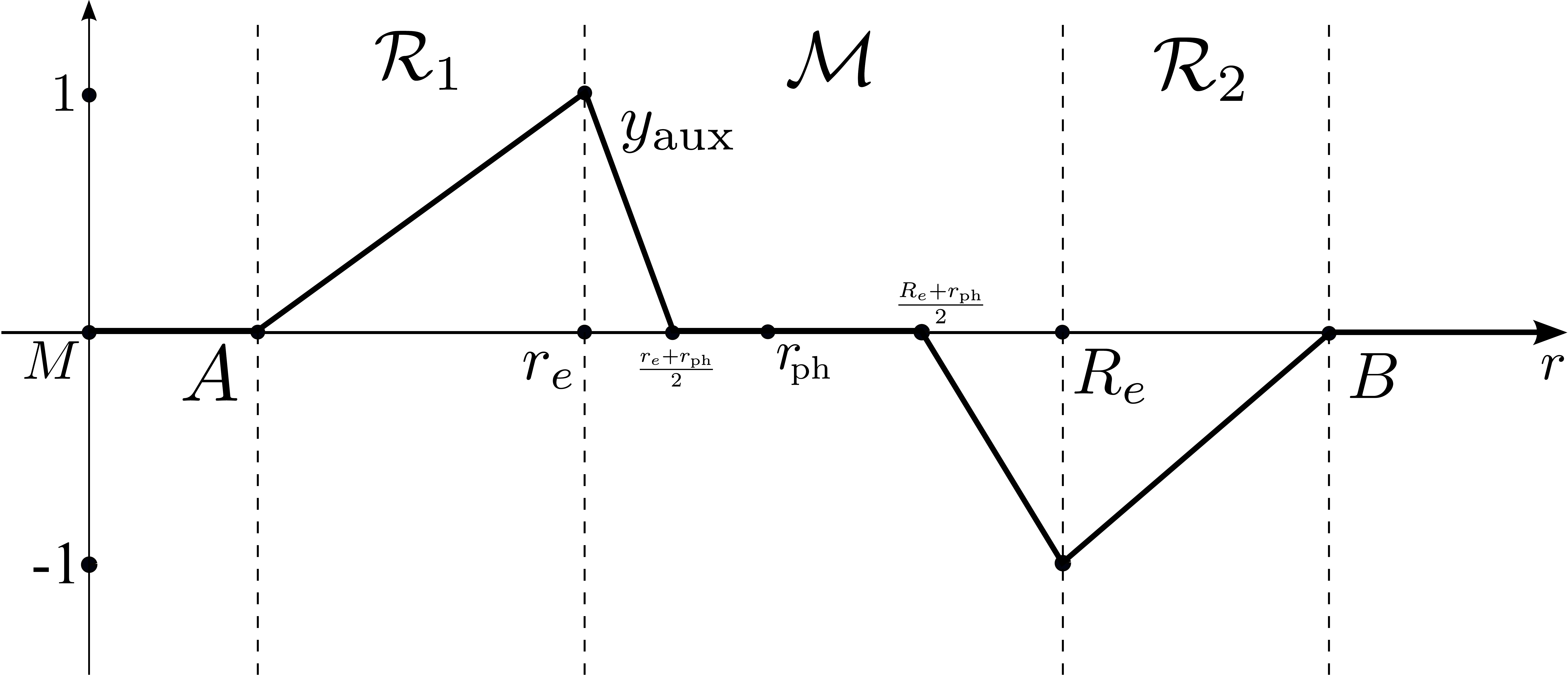}
	\label{fig:yaux}
\end{figure}
Note that $r_{\text{ph}}=(1+\sqrt{2})M$. Clearly $y_{\text{aux}}'>0$ in regions $\mathcal{R}_{1},\mathcal{R}_{2}$. Note that these regions are located far away from the event horizon and the photon sphere and  $r$ is bounded, and therefore, for all frequencies we have  $V\leq b_{1}|V'|$ (where $b_{1}$ depends only on $M, A, B$). Hence
\begin{equation*}
\begin{split}
y_{\text{aux}}'(\omega^{2}-V)-y_{\text{aux}}V' &\geq -y_{\text{aux}}'V-y_{\text{aux}}V' \\ & \geq -y_{\text{aux}}'\cdot b_{1}|V'|-y_{\text{aux}}V'=-(y_{\text{aux}}'\cdot b_{1}\cdot \text{sign}(V')+y_{\text{aux}})V'
\end{split}
\end{equation*}
in $\mathcal{R}_{1}\cup\mathcal{R}_{2}$.

Rescale now $y_{\text{aux}}$ such that $-(y_{\text{aux}}'\cdot b_{1}\cdot \text{sign}(V')+y_{\text{aux}})V'>-\frac{1}{2}b(A,B)|V'|$, where $b(A,B)$ is the constant on the left hand side of \eqref{finalfourier22}.  Rescale it further, so in the regions $\left\{r_{e}\leq r\leq \frac{r_{e}+r_{\text{ph}}}{2}\right\}$ and $\left\{\frac{R_{e}+r_{\text{ph}}}{2}\leq r\leq R_{e}\right\}$ the terms $y_{\text{aux}}|u'|^{2}$ and $\Big[y'_{\text{aux}}(\omega^{2}-V)-y_{\text{aux}}V'\Big]|u|^{2}$ can be controlled by the left hand side of \eqref{ef11}, \eqref{ef12}, \eqref{trapping}, \eqref{angular} and \eqref{time}. Clearly, the ``final'' $y_{\text{aux}}$ does not depend on the frequencies $(\omega^{2},\Lambda)$. Therefore, we obtain the following:

There exists a constant $b(A,B)$  such that for all frequencies $(\omega^{2},\Lambda)$ we have
\begin{equation}
\begin{split}
b(A,B)&\int_{\left\{A\leq r\leq r_{e}\right\}\cup\left\{R_{e}\leq r\leq  B\right\}}
\Big[|u'|^{2}+\left|u\right|^{2}\Big]dr^{*}\\ &\ \ \ \ \ \ \ \ \ \leq   
\left(\mathcal{J}_{1}^{y=1}[u]\right)\!(B)-\left(\mathcal{J}_{1}^{y=-1}[u]\right)\!(A)+\int_{A}^{B}E(H)dr^{*}. 
\end{split}
\label{eq:finalaux}
\end{equation}
Note that $E^{y_{\text{aux}}}(H)$ is now included in the $E(H)$. The above estimate holds for all $A\leq r_{0}$ and $B\geq R_{0}$.

\subsection{Micro-Local Integrated Decay Estimate}
\label{sec:MicroLocalIntegratedDecayEstimate}

We can now show the following

\begin{proposition}
Let  $u$  satisfy \eqref{e3iswsntouu}. Let $r_{e}>M$ and $r_{0},R_{0}$ be as defined in Section \ref{sec:AuxilliaryCurrents}. Let $A\leq r_{0}$ and $B\geq R_{0}$.   Then there exists a constant $b>0$ which depends only on $M$, $r_{e}$ and $R_{e}$ and a constant $\epsilon>0$ which depends on $M, r_{e}, R_{e}, A$ and  $B$  such that  \textbf{for all frequencies} $(\omega^{2},\Lambda)$ we have
\begin{equation}
\begin{split}
&b\int_{\mathcal{M}}\left[|u'|^{2}+|u|^{2}+\left(r-(1+\sqrt{2})M\right)^{2}\left[\Lambda+\omega^{2}\right]|u|^{2}\right]dr^{*}\\ &\ \ \ \ \ \ \ \ \ \leq   
\left(\mathcal{J}_{1}^{y=1}[u]\right)\!(B)-\left(\mathcal{J}_{1}^{y=-1}[u]\right)\!(A)+\frac{1}{\epsilon}\int_{A}^{B}|F|^{2}dr^{*}.
\end{split}
\label{finalfourier1}
\end{equation}
\label{ffourier}
\end{proposition}
\begin{proof}
Let $I_{\text{main}}[\psi]=\left(\partial_{r^{*}}\psi\right)^{2}+\psi^{2}+\left(r-(1+\sqrt{2})M\right)^{2}\left[\left|\nabb\psi\right|^{2}+\left(T\psi\right)^{2}\right]$. Then from \eqref{ef11}, \eqref{ef12}, \eqref{trapping}, \eqref{angular}, \eqref{time}, \eqref{eq:finalaux} and \eqref{eh} we have that there exists a constant $b=b(M,r_{e},R_{e})$ such that
\begin{equation*}
\begin{split}
b\int_{\mathcal{M}}&I_{\text{main}}[\psi]dr^{*}+b(A,B)\int_{\left\{A\leq r\leq r_{e}\right\}\cup\left\{R_{e}\leq r\leq  B\right\}}
\Big[|u'|^{2}+\left|u\right|^{2}\Big]dr^{*}\\ &\leq    
\left(\mathcal{J}_{1}^{y=1}[u]\right)\!(B)-\left(\mathcal{J}_{1}^{y=-1}[u]\right)\!(A)+\int_{A}^{B}E(H)dr^{*}
\\& \leq \left(\mathcal{J}_{1}^{y=1}[u]\right)\!(B)-\left(\mathcal{J}_{1}^{y=-1}[u]\right)\!(A)+C(A,B)\int_{A}^{B}\left[\epsilon|u'|^{2}+\epsilon|u|^{2}+\frac{1}{\epsilon}\left|F\right|^{2}\right]dr^{*}.
\end{split}
\end{equation*}
It suffices to choose $\epsilon$ such that $C(A,B)\epsilon<\text{min}\left\{b,b(A,B)\right\}$.
\end{proof}

\section{Physical Space Estimates}
\label{sec:PhysicalSpaceEstimates}

\subsection{The Main Estimate I}
\label{sec:MainEstimateI}

We next turn the above microlocal estimate into a physical space estimate.
\begin{proposition}
Let $r_{e},r_{0},R_{0},b,\epsilon$ be as in Proposition \ref{ffourier}. Then if $\psi_{\hbox{\Rightscissors}}$ is as defined in Section \ref{sec:SeparabilityOfTheWaveEquation} then
\begin{equation}
\begin{split}
& b\int_{\mathcal{M}}\left[\left(\partial_{r^{*}}\psi_{\hbox{\Rightscissors}}\right)^{2}+\left(\psi_{\hbox{\Rightscissors}}\right)^{2}+\left(r-(1+\sqrt{2})M\right)^{2}\left[\left|\nabb\psi_{\hbox{\Rightscissors}}\right|^{2}+\left(T\psi_{\hbox{\Rightscissors}}\right)^{2}\right]\right]\\
&\leq \int_{\left\{r=A\right\}}\textbf{\textit{J}}[\psi_{\hbox{\Rightscissors}}]dt\,dg_{\mathbb{S}^{2}}+ \int_{\left\{r=B\right\}}\textbf{\textit{J}}[\psi_{\hbox{\Rightscissors}}]dt\,dg_{\mathbb{S}^{2}} +\int_{\left\{A\leq r\leq B\right\}}\int_{\omega}\sum_{\ell}\frac{1}{\epsilon}\left|F\right|^{2}d\omega\, dr^{*}.\\
\end{split}
\label{eq:phspint}
\end{equation}
for all $A\leq r_{0}$ and $B\geq R_{0}$. Recall that  the expression $\textbf{\textit{J}}[\psi_{\hbox{\Rightscissors}}]$ is given by \eqref{defboundaryphys} and $dt\,dg_{\mathbb{S}^{2}}=\sin\theta \,d\theta\, d\phi\, dt$.
\label{phspint}
\end{proposition}
\begin{proof}
We derive the above estimate by summing \eqref{finalfourier1} over all $\Lambda$ ($\lambda_{m\ell}$), using the $L^{2}$-convergence on $\mathbb{S}^{2}$, then integrating over $\omega\in\mathbb{R}$, using the identities of Section \ref{sec:PhysicalSpaceFourierSpace} and integrating in $r^{*}$ from $r=A$ to $r=B$. As regards the boundary terms at the $r=A$ and $r=B$ hypersurfaces, recall that $\mathcal{J}_{1}^{y=1}[u]=|u'|^{2}+(\omega^{2}-V)|u|^{2}$. Then
\begin{equation*}
\begin{split}
&\ \ \ \ \ \int_{\omega}\sum_{\ell}\bigg(\mathcal{J}_{1}^{y=1}[u]\bigg)(r=c)d\omega=\\ & 
\int_{\left\{r=c\right\}}\bigg[\left(\partial_{r^{*}}\left(\sqrt{r^{2}+M^{2}}\psi_{\hbox{\Rightscissors}}\right)\right)^{2}+\Big(T\psi_{\hbox{\Rightscissors}}\Big)^{2}\cdot (r^{2}+M^{2})-\frac{(r-M)^{2}}{(r^{2}+M^{2})}\left[a^{2}\sin^{2}\theta(T\psi_{\hbox{\Rightscissors}})^{2}+\left|\nabb_{\mathbb{S}^{2}}\psi_{\hbox{\Rightscissors}}\right|^{2}\right]\bigg] \\
\ \ &\ \ \ \ \ \ \ \  -\bigg[\frac{(r-M)^{3}M}{(r^{2}+M^{2})^{3}}(2r^{2}+3rM-M^{2})(\psi_{\hbox{\Rightscissors}})^{2}\bigg]dtdg_{\mathbb{S}^{2}}.
\end{split}
\end{equation*}
One can easily see that the integrand expression equals to $\textbf{\textit{J}}[\psi_{\hbox{\Rightscissors}}]$. This calculation explains why it was so crucial that all the $y,f$ multipliers of Section \ref{sec:FourierLocalisedEstimates} were equal to $-1$ $(+1)$ close (away) to $\hh$. 

\end{proof}

 We now wish to derive our main estimate for $\psi$  itself. 
 \begin{proposition}\textbf{(Main Estimate I)}
 Let  $r_{e},r_{0},R_{0},b,\epsilon$ be as in Proposition \ref{ffourier}. Let also $\mathcal{S}_{\xi}$ be the region as defined in Section \ref{sec:TheCutOffXiVarepsilonTau}.  Then there exists a constant $C$ which depends only  on $M$, $r_{e}$ and $R_{e}$ such that for all axisymmetric solutions $\psi$ of the wave equation  we have 
 \begin{equation}
\begin{split}
& \ \ \ \ \ b\int_{\mathcal{M}}\left[\left(\partial_{r^{*}}\psi\right)^{2}+\psi^{2}+\left(r-(1+\sqrt{2})M\right)^{2}\left[\left|\nabb\psi\right|^{2}+\left(T\psi\right)^{2}\right]\right]\\
\leq & \int_{\left\{r=A\right\}}\textbf{\textit{J}}[\psi]dtdg_{\mathbb{S}^{2}}+ \int_{\left\{r=B\right\}}\textbf{\textit{J}}[\psi]dtdg_{\mathbb{S}^{2}}\\&+C(A)\int_{\left\{r=A\right\}\cap\mathcal{S}_{\xi}}\Big[E^{1}[\psi]\Big]dtdg_{\mathbb{S}^{2}}+C(B)\int_{\left\{r=B\right\}\cap\mathcal{S}_{\xi}}\Big[E^{1}[\psi]\Big]dtdg_{\mathbb{S}^{2}}\\& +
C\int_{\Sigma_{0}}J^{T}_{\mu}[\psi]n^{\mu}_{\Sigma_{0}}+\int_{\left\{A\leq r\leq B\right\}}\int_{\omega}\sum_{\ell}\frac{1}{\epsilon}\left|F\right|^{2}d\omega dr^{*}
\end{split}
\label{oeq:phspint}
\end{equation}
for all $A\leq r_{0}$ and $B\geq R_{0}$, where $E^{1}[\psi]=(\partial_{r^{*}}\psi)^{2}+(T\psi)^{2}+|\nabb\psi|^{2}+\psi^{2}$.
  \label{intpsi}
 \end{proposition}
 
\begin{proof} 
Let $I_{\text{main}}[\psi]=\left(\partial_{r^{*}}\psi\right)^{2}+\psi^{2}+\left(r-(1+\sqrt{2})M\right)^{2}\left[\left|\nabb\psi\right|^{2}+\left(T\psi\right)^{2}\right]$. Then, since $\psi=\psi_{\hbox{\Rightscissors}}$ in $\mathcal{R}-\mathcal{S}_{\xi}$ we obtain
$\int_{\mathcal{M}-\mathcal{S}_{\xi}}I_{\text{main}}[\psi]\leq\int_{\mathcal{M}}I_{\text{main}}\left[\psi_{\hbox{\Rightscissors}}\right]$, and therefore, 
\begin{equation*}
\begin{split}
\int_{\mathcal{M}}I_{\text{main}}[\psi] \leq &\int_{\mathcal{M}}I_{\text{main}}\left[\psi_{\hbox{\Rightscissors}}\right]+\int_{\mathcal{M}\cap\mathcal{S}_{\xi}}I_{\text{main}}[\psi]\\
\leq & \int_{\mathcal{M}}I_{\text{main}}\left[\psi_{\hbox{\Rightscissors}}\right]+C\int_{0}^{1}\left(\int_{\Sigma_{\tilde{\tau}}}J_{\mu}^{T}[\psi]n^{\mu}_{\Sigma_{\tilde{\tau}}}\right)d\tilde{\tau}+C\int_{\tau-1}^{\tau}\left(\int_{\Sigma_{\tilde{\tau}}}J_{\mu}^{T}[\psi]n^{\mu}_{\Sigma_{\tilde{\tau}}}\right)d\tilde{\tau}\\ \leq & \int_{\mathcal{M}}I_{\text{main}}\left[\psi_{\hbox{\Rightscissors}}\right]+
C\int_{\Sigma_{0}}J_{\mu}^{T}[\psi]n^{\mu}_{\Sigma_{0}}.
\end{split}
\end{equation*}
Regarding the boundary integrals we have 
\begin{equation*}
\begin{split}
\int_{\left\{r=c\right\}}\textbf{\textit{J}}[\psi_{\hbox{\Rightscissors}}]dtdg_{\mathbb{S}^{2}}\leq 
\int_{\left\{r=c\right\}}\textbf{\textit{J}}[\psi]dtdg_{\mathbb{S}^{2}}+C(c)\int_{\left\{r=c\right\}\cap\mathcal{S}_{\xi}}\Big[E^{1}[\psi]\Big]dtdg_{\mathbb{S}^{2}}.
\end{split}
\end{equation*}
 The result now follows from \eqref{eq:phspint}.

\end{proof}

It remains to estimate the  terms  on the right hand side of \eqref{oeq:phspint}.

\subsection{Estimates for the Boundary Integrals}
\label{sec:EstimatingErrorTerms}
 We have the following
\begin{proposition}
Let $c>M$ and \textbf{\textit{J}} be as defined in Section  \ref{sec:NonNegativeDefiniteCurrentsNearAndAwayHh}. Then there exists a constant $C$ that depends only on $M$ such that for all axisymmetric solutions $\psi$ of the wave equation we have
\begin{equation*}
\int_{\left\{r=c\right\}}\textbf{\textit{J}}[\psi]dtdg_{\mathbb{S}^{2}}\leq C\int_{\Sigma_{0}}J_{\mu}^{T}[\psi]n^{\mu}_{\Sigma_{0}}.
\end{equation*}
\label{bpropc}
\end{proposition}
\begin{proof}

Let $A<(1+\sqrt{2})M$.  If we use divergence identity for the current $J_{\mu}^{X,G}[\psi]=J_{\mu}^{X}[\psi]-2G\psi\nabla_{\mu}\psi-(\nabla_{\mu}G)\psi^{2}$, for $X=-2\partial_{r^{*}}$ and $G=-\frac{\Delta\cdot r}{(r^{2}+M^{2})^{2}}$ in the shaded region 
\begin{figure}[H]
	\centering
		\includegraphics[scale=0.125]{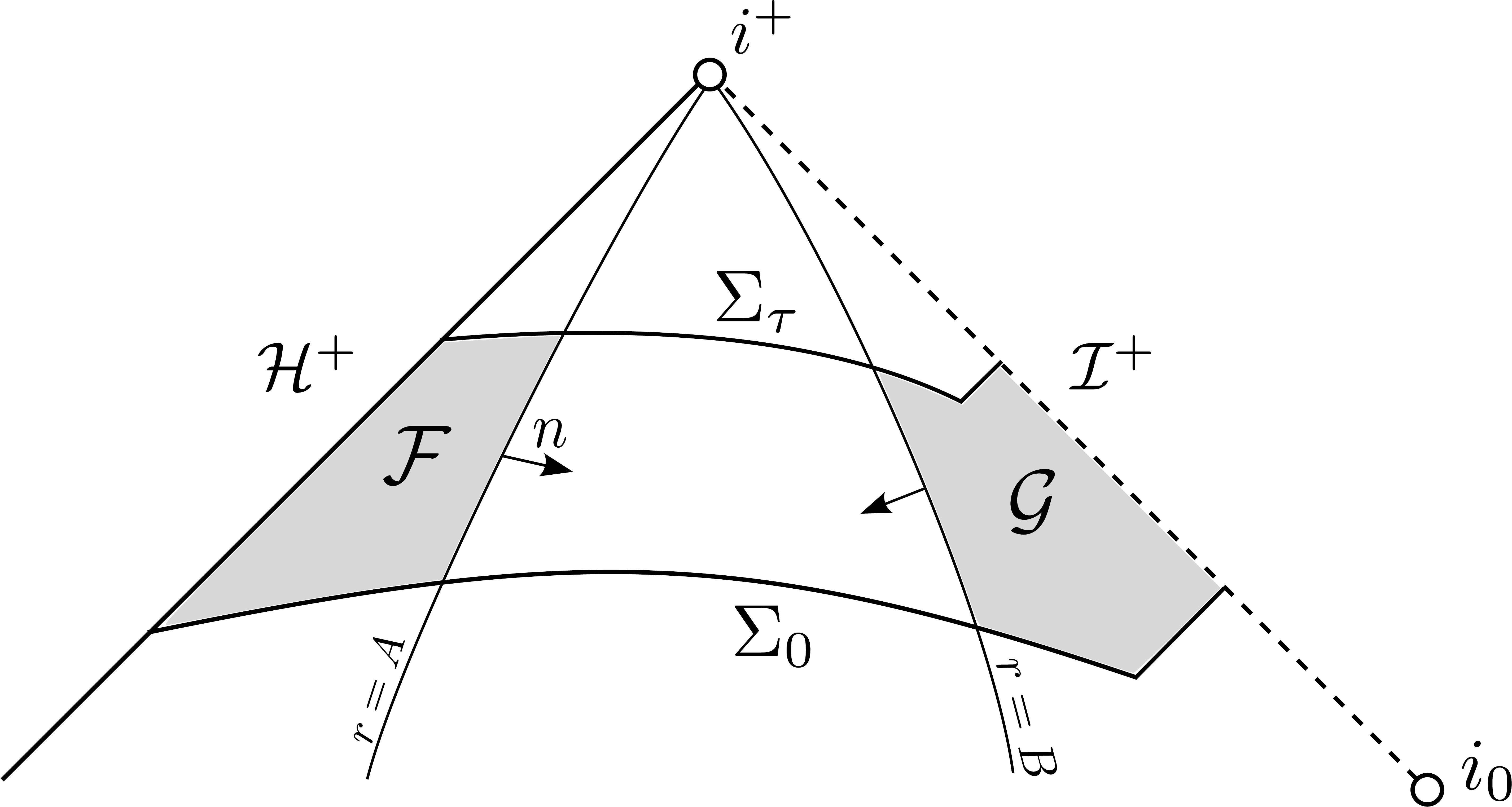}
	\label{fig:ern5}
\end{figure}\noindent then we obtain
\begin{equation*}
\int_{\Sigma_{\tau}}J_{\mu}^{X,G}[\psi]n^{\mu}_{\Sigma_{\tau}}+\int_{\mathcal{F}}K^{X,G}[\psi]+\int_{\hh\cup\mathcal{I}^{+}}J_{\mu}^{X,G}[\psi]n^{\mu}_{\hh}-\int_{\left\{r=A\right\}}J_{\mu}^{X,G}[\psi]n^{\mu}=\int_{\Sigma_{0}}J_{\mu}^{X,G}[\psi]n^{\mu}_{\Sigma_{0}}.
\end{equation*}
In view of the results of Section \ref{sec:NonNegativeDefiniteCurrentsNearAndAwayHh},  for these  choices of $X$ and $G$ we have $K^{X,G}[\psi]\geq 0$ and that 
\begin{equation*}
-\int_{\left\{r=A\right\}}J_{\mu}^{X,G}[\psi]n^{\mu}=\int_{\left\{r=A\right\}}\textbf{\textit{J}}[\psi]dtdg_{\mathbb{S}^{2}}.
\end{equation*}
Note that the left integral is with respect to the \textit{induced volume form}. Moreover, from the first Hardy inequality we have that the remaining boundary integrals over $\Sigma_{\tau}$ and $\hh$ are bounded by the conserved $T$-flux. 

Similarly, by taking $X=2\partial_{r^{*}}$ and $G=\frac{\Delta\cdot r}{(r^{2}+M^{2})^{2}}$ we obtain the same estimate for the hypersurface $r=B$ where $B>(1+\sqrt{2})M$.

\end{proof}

\subsection{Boundary Error Terms from the Cut-off; Averaging}
\label{sec:BoundaryErrorTermsFromTheCutOffAveraging}

We next estimate the boundary integrals which are supported on $\mathcal{S}_{\xi}$. Recall that all previous estimates hold for all $A\leq r_{0}$ and $B\geq R_{0}$. Let us next restrict $A,B$ such that
\begin{equation*}
\frac{M+r_{0}}{2}\leq A\leq r_{0}, \ \ \  R_{0}\leq B\leq R_{0}+1. 
\end{equation*}
Recall that $E^{1}[\psi]=(\partial_{r^{*}}\psi)^{2}+(T\psi)^{2}+|\nabb\psi|^{2}+\psi^{2}$.
 We have the following
\begin{proposition}
Let $\psi$ satisfy $\Box_{g}\psi=0$ and $\Phi\psi=0$. Then, there exist constants $A_{0}\in\left[\frac{M+r_{0}}{2}, r_{0}\right]$ and $B_{0}\in[R_{0}, R_{0}+1]$ which depend on $M, r_{0}, R_{0}$ and possibly on $\psi$ such that 
\begin{equation*}
\int_{\left\{r=c\right\}\cap\mathcal{S}_{\xi}}E^{1}[\psi]dtdg_{\mathbb{S}^{2}}\leq C(r_{0},R_{0})\int_{\Sigma_{0}}J_{\mu}^{T}[\psi]n^{\mu}_{\Sigma_{0}},
\end{equation*}
where $c\in \left\{A_{0}, B_{0}\right\}$ and $C(r_{0},R_{0})$ a constant which depends only on $M, r_{0}$ and  $R_{0}$.

\label{averaging}
\end{proposition}

\begin{proof}
\begin{figure}[H]
	\centering
		\includegraphics[scale=0.1]{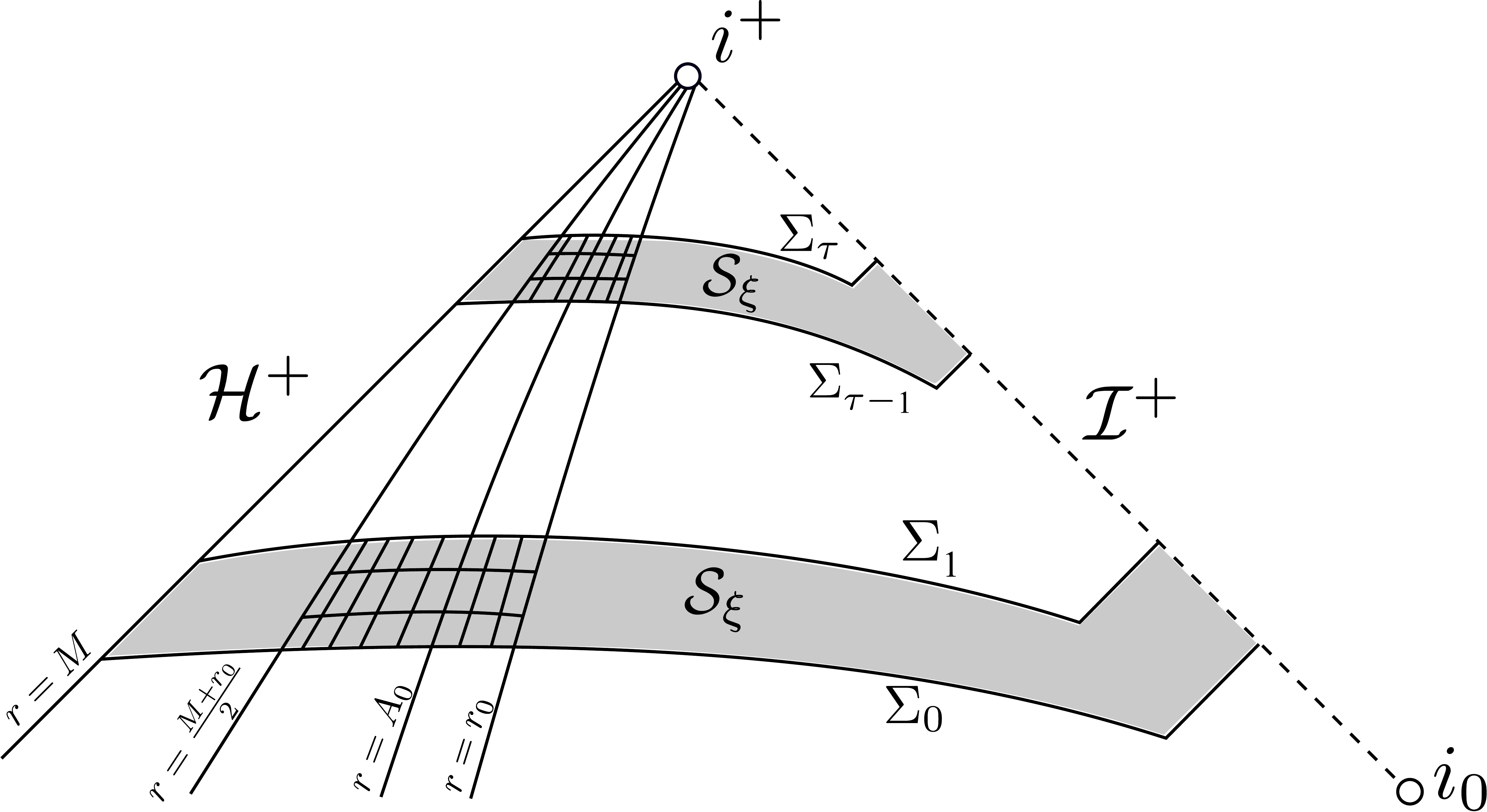}
	\label{fig:xiaver}
\end{figure}If we use the first Hardy inequality to estimate the zeroth order term, we obtain 
\begin{equation*}
\begin{split}
\int_{\frac{M+r_{0}}{2}}^{r_{0}}\left(\int_{\left\{r=c\right\}\cap\mathcal{S}_{\xi}}E^{1}[\psi]dtdg_{\mathbb{S}^{2}}\right)dr \leq & C(r_{0})\int_{\left\{\frac{M+r_{0}}{2}\leq r\leq r_{0}\right\}\cap\Big\{\left\{0\leq \tilde{\tau}\leq 1\right\}\cup\left\{\tau-1\leq \tilde{\tau}\leq\tau\right\}\Big\}}E^{1}[\psi]\\ \leq & C(r_{0}) \int_{\Big\{\left\{0\leq \tilde{\tau}\leq 1\right\}\cup\left\{\tau-1\leq \tilde{\tau}\leq \tau\right\}\Big\}}\!\!\left(\int_{\Sigma_{\tilde{\tau}}}J_{\mu}^{T}[\psi]n^{\mu}_{\Sigma_{\tilde{\tau}}}\right)d\tilde{\tau}\\
\leq & C(r_{0})\int_{\Sigma_{0}}J_{\mu}^{T}[\psi]n^{\mu}_{\Sigma_{0}}.
\end{split}
\end{equation*}
Therefore, by pigeonhole principle we deduce that there exists $A_{0}\in\left[\frac{M+r_{0}}{2},r_{0}\right]$ such that 
\begin{equation*}
\int_{\left\{r=A_{0}\right\}\cap\mathcal{S}_{\xi}}E^{1}[\psi]dtdg_{\mathbb{S}^{2}}\leq \frac{C(r_{0})}{(r_{0}-M)/2}\int_{\Sigma_{0}}J_{\mu}^{T}[\psi]n^{\mu}_{\Sigma_{0}}\leq C(r_{0})\int_{\Sigma_{0}}J_{\mu}^{T}[\psi]n^{\mu}_{\Sigma_{0}}. 
\end{equation*}
Similarly we argue for the existence of $B_{0}\in\left[R_{0},R_{0}+1\right]$. Note that although $A_{0}, B_{0}$ may possibly depend  on $\psi$, the constant $C(r_{0}, R_{0})$ depends only on $M$, $r_{0}$ and  $R_{0}$.  
\end{proof}
We next assume that $A_{0}, B_{0}$ are the constants of the previous proposition and we apply Proposition \ref{intpsi} for these choices for $A,B$. \textbf{Since $A_{0}, B_{0}$ lie in compact intervals which are completely determined by $M, r_{0}$ and $R_{0}$, all the involved constants will in fact depend only on  $M,r_{0}$ and $R_{0}$ (and not on the precise choice for $A_{0},B_{0}$).}

\subsection{Spacetime Error Terms from the Cut-off}
\label{sec:ErrorTermsFromTheCutOff}

 It remains to control the term that involves $F$. Recall that $\epsilon$ is the constant of Proposition \ref{ffourier}.
 We have 
\begin{equation}
\begin{split}
C(r_{0},R_{0})\int_{A_{0}}^{B_{0}}\int_{\omega}\sum_{\ell}\frac{1}{\epsilon}|F|^{2}d\omega dr^{*}& \leq C(r_{0},R_{0})\int_{\left\{A_{0}\leq r\leq B_{0}\right\}}\frac{1}{\epsilon}|F|^{2}\\
&= C(r_{0},R_{0})\int_{\left\{A_{0}\leq r\leq B_{0}\right\}\cap\mathcal{S}_{\xi}}\frac{1}{\epsilon}|F|^{2}\\
&\leq C(r_{0},R_{0},\epsilon)\int_{\left\{A_{0}\leq r\leq B_{0}\right\}\cap\mathcal{S}_{\xi}}E^{1}[\psi]\\
&\leq C(r_{0},R_{0},\epsilon)\int_{\Sigma_{0}}J_{\mu}^{T}[\psi]n^{\mu}_{\Sigma_{0}},
\end{split}
\label{cutoff1}
\end{equation}
where we  used the properties of the inhomogeneity term $F$ (see Section \ref{sec:SeparabilityOfTheWaveEquation}), that $dg_{\text{vol}}\sim\frac{\Delta}{r^{2}}dr^{*}dtdg_{\mathbb{S}^{2}}$ (again all powers of $r$ and $(r-M)$ are incorporated in the constant $C(r_{0},R_{0})$), the coarea formula and the first Hardy inequality.

\subsection{The Main Estimate II}
\label{sec:TheMainEstimateII}

We finally obtain
 \begin{proposition}\textbf{(Main Estimate II)}
 Let $R_{e}>r_{e}>M$ and $\mathcal{M}=\left\{r_{e}\leq r\leq R_{e}\right\}$. Then there exists a constant $C(r_{e},R_{e})$ which depends only on $r_{e},R_{e}$ and $M$ such that for all axisymmetric solutions $\psi$ of the wave equation  we have 
 \begin{equation}
\begin{split}
&\!\!\!\!\!\!\!\!  \int_{\mathcal{M}}\left[\left(\partial_{r^{*}}\psi\right)^{2}+\psi^{2}+\left(r-(1+\sqrt{2})M\right)^{2}\left[\left|\nabb\psi\right|^{2}+\left(T\psi\right)^{2}\right]\right]\\
&\leq
C(r_{e},R_{e})\int_{\Sigma_{0}}J_{\mu}^{T}[\psi]n^{\mu}_{\Sigma_{0}}.
\end{split}
\label{main2}
\end{equation}
  \label{main2p}
 \end{proposition}
\begin{proof}
We use Proposition \ref{ffourier} and the results of Sections \ref{sec:EstimatingErrorTerms}, \ref{sec:BoundaryErrorTermsFromTheCutOffAveraging} and \ref{sec:ErrorTermsFromTheCutOff}. Note that all constants involved depend on $M, r_{0}$ and $R_{0}$.  Recall that $r_{0},R_{0}$ are completely determined by $M, r_{e}$ and $R_{e}$, and therefore, all constants depend only on $M, r_{e}$ and $R_{e}$. 
\end{proof}
The degeneracy on the photon sphere can be removed at the expense of commuting with the Killing field $T$ (thus losing a derivative) and using an appropriate scalar multiple of the Lagrangian current $\psi\nabla_{\mu}\psi$. See \cite{aretakis2} for the details of such commutations.

Using \eqref{main2} and Proposition \ref{larger} completes the proof of Theorem \ref{t2}.

\section{Energy Estimates}
\label{sec:EnergyEstimates}

\subsection{Uniform Boundedness of Non-Degenerate Energy}
\label{sec:UniformBoundednessOfNonDegenerateEnergy}

The crucial ingredient to showing boundedness of the non-degenerate energy are the Propositions \ref{nprop} and \ref{main2p}. First note that we can extend (for $r\geq r_{e}$) the vector field $N$ so that it remains globally translation-invariant timelike and $N=T$ for $r\geq r_{e}+\frac{M}{23}$. Indeed,   in view of \eqref{main2}, we can bound all the error terms in the intermediate region $\left\{r_{e}\leq r\leq r_{e}+\frac{M}{23}\right\}$. We also introduce a smooth cut-off function $\delta:[M,+\infty]\rightarrow\mathbb{R}$ such that  $\delta\left(r\right)=1,r\in\left[M,r_{e}\right]$ and $\delta\left(r\right)=0,r\in\left[\left.r_{e}+\frac{M}{23},
+\infty\right.\right)$ and  consider the currents
\begin{equation}
\begin{split}
&J_{\mu}^{N,\delta,-\frac{1}{2}}\overset{.}{=}J_{\mu}^{N}-\frac{1}{2}\delta\psi\nabla_{\mu}\psi,\ \ K^{N,\delta,-\frac{1}{2}}\overset{.}{=}\nabla^{\mu}J_{\mu}^{N,\delta,-\frac{1}{2}}.\\
\end{split}
\label{fcurrent}
\end{equation}
\begin{proof}[Proof of Theorem \ref{t1}; Uniform Boundedness of Energy]
Stokes' theorem for the current $J_{\mu}^{N,\delta, -\frac{1}{2}}$ in region $\mathcal{R}(0,\tau)$ gives us
\begin{equation*}
\int_{\Sigma_{\tau}}{J_{\mu}^{N,\delta, -\frac{1}{2}}n^{\mu}_{\Sigma_{\tau}}}+\int_{\mathcal{R}}{K^{N,\delta, -\frac{1}{2}}}+\int_{\mathcal{H}^{+}}{J_{\mu}^{N,\delta, -\frac{1}{2}}n^{\mu}_{\mathcal{H}^{+}}}+\int_{\mathcal{I}^{+}}{J_{\mu}^{N,\delta, -\frac{1}{2}}n^{\mu}_{\mathcal{I}^{+}}}= \int_{\Sigma_{0}}{J_{\mu}^{N,\delta, -\frac{1}{2}}n^{\mu}_{\Sigma_{0}}}.
\end{equation*}
First observe that for the spacelike boundary terms we have the estimates
\begin{equation*}
\int_{\Sigma_{0}}{J_{\mu}^{N,\delta ,-\frac{1}{2}}[\psi]n^{\mu}}\leq C\int_{\Sigma_{0}}{J_{\mu}^{N}[\psi]n^{\mu}}
\end{equation*}
and 
\begin{equation*}
\int_{\Sigma_{\tau}}{J_{\mu}^{N}[\psi]n^{\mu}}\leq 2\int_{\Sigma_{\tau}}{J_{\mu}^{N,\delta,-\frac{1}{2}}[\psi]n^{\mu}}+C\int_{\Sigma_{\tau}}{J_{\mu}^{T}[\psi]n^{\mu}},
\end{equation*}
which are both applications of the first Hardy inequality. For the integral over $\hh$ recall that $n_{\hh}=T+\frac{1}{2M}\Phi$, $\Phi\psi=0$ and since $\delta=1$ on $\hh$ we have
that $J_{\mu}^{N,\delta,-\frac{1}{2}}n^{\mu}_{\mathcal{H}^{+}}=J_{\mu}^{N}n^{\mu}_{\mathcal{H}^{+}}-\frac{1}{2}\psi  T\psi$. 
However,
\begin{equation*}
\begin{split}
\int_{\mathcal{H}^{+}}{-2\psi T\psi}&=\int_{\mathcal{H}^{+}}{-T\psi^{2}}=\int_{\mathcal{H}^{+}\cap\Sigma_{0}}{\psi^{2}}-\int_{\mathcal{H}^{+}\cap\Sigma_{\tau}}{\psi^{2}}.
\end{split}
\end{equation*}
From the first and second Hardy inequality  we have
\begin{equation*}
\int_{\mathcal{H}^{+}\cap\Sigma}{\psi^{2}}\leq C_{\epsilon}\int_{\Sigma}{J_{\mu}^{T}n^{\mu}_{\Sigma}}+\epsilon\int_{\Sigma}{(T\psi)^{2}+(Y\psi)^{2}}\leq C_{\epsilon}\int_{\Sigma}{J_{\mu}^{T}n^{\mu}_{\Sigma}}+\epsilon\int_{\Sigma}{J_{\mu}^{N}n^{\mu}_{\Sigma}}
\end{equation*}
and hence
\begin{equation*}
\int_{\mathcal{H}^{+}}{J_{\mu}^{N,\delta,-\frac{1}{2}}[\psi]n^{\mu}_{\mathcal{H}^{+}}}\geq \int_{\mathcal{H}^{+}}{J_{\mu}^{N}[\psi]n^{\mu}_{\mathcal{H}^{+}}}\ -C_{\epsilon}\int_{\Sigma_{\tau}}{J_{\mu}^{T}[\psi]n^{\mu}_{\Sigma_{\tau}}}\ -\epsilon\int_{\Sigma_{\tau}}{J_{\mu}^{N}[\psi]n^{\mu}_{\Sigma_{\tau}}},
\end{equation*}
for any $\epsilon >0$.
The spacetime term is non-negative (and thus has the right sign) in the  region $\left\{r\leq r_{e}\right\}$, vanishes far away from $\hh$ and can be estimated in the intermediate spatially compact region (which does not contain the photon sphere) by Theorem \ref{t2}. The result follows from the boundedness of $T$-flux through $\Sigma_{\tau}$ (Proposition \ref{tprop}).
\end{proof}

\subsection{The Trapping Effect on $\hh$}
\label{sec:TheTrappingEffectOnHh}

By revisiting the proof of the previous section we obtain the following
\begin{proposition}
There exists a  constant $C>0$ which depends  on $M$ and $\Sigma_{0}$ such that for all solutions $\psi$ of the wave equation 
\begin{equation}
\int_{\left\{r\leq r_{e}\right\}}{K^{N,-\frac{1}{2}}[\psi]}\leq C\int_{\Sigma_{0}}{J_{\mu}^{N}[\psi]n^{\mu}_{\Sigma_{0}}}.
\label{nk}\end{equation}
\label{nkcor}
\end{proposition}

The above estimate  gives us a spacetime integral where the only weight that locally degenerates (to first order) is that of the derivative tranversal to $\mathcal{H}^{+}$. Recall that in the  subextreme Kerr case there is no such degeneration. This phenomenon allows us conclude that \textbf{trapping} takes place on the event horizon of extreme Kerr.

One application of the above theorem and the third Hardy inequality is the following  Morawetz estimate which does not degenerate at $\mathcal{H}^{+}$.
\begin{corollary}
There exists a  constant $C>0$ which depends  on $M$ and $\Sigma_{0}$ such that for all solutions $\psi$ of the wave equation 
\begin{equation}
\int_{\left\{r\leq r_{e}\right\}}{\psi^{2}}\leq C\int_{\Sigma_{0}}{J_{\mu}^{N}[\psi]n^{\mu}_{\Sigma_{0}}}.
\label{nodmoraw}
\end{equation}
\label{moranondeg}
\end{corollary}

This completes the proof of Theorem \ref{t2}.

\section{Energy Decay}
\label{sec:Energydecay}

The main goal of this section is the decay of the degenerate energy, which will be crucial for obtaining pointwise decay (see Section \ref{sec:PointwiseDecay}). Clearly, the first step is to obtain a bound on the 4-integral of the degenerate energy density integrated over the domain of dependence of $\Sigma_{0}$. In fact, having shown Proposition \ref{main2p}, we may restrict to regions $\mathcal{A}$ and $\mathcal{B}$ which are neighbourhoods of $\hh$ and $\mathcal{I}^{+}$, respectively. We first consider $\mathcal{A}$.

\subsection{The Vector Field $P$}
\label{sec:TheVectorFieldP}

We derive a hierarchy of (degenerate) energy estimates in a neighbourhood of $\mathcal{H}^{+}$, the crucial ingredient of which is a vector field $P$ to be constructed. This vector field is  timelike in the domain of outer communications and becomes null on the horizon  ``linearly". This linearity allows $P$ to capture  the degenerate redshift in $\mathcal{A}$  in a weaker way than $N$ but in stronger way than $T$.
 In this section, we use the $(v,r)$ coordinates. 
\begin{proposition}
There exists a $\phi_{\tau}^{T}$-invariant causal future directed vector field $P$ and a constant $C$ which depends only on $M$ such that for all axisymmetric functions $\psi$ we have
\begin{equation*}
J_{\mu}^{T}[\psi]n^{\mu}_{\Sigma}\leq CK^{P}[\psi],  \ \ J_{\mu}^{P}[\psi]n^{\mu}_{\Sigma}\leq CK^{N,-\frac{1}{2}}[\psi],  
\end{equation*}
in the region $\mathcal{A}=\left\{M\leq r\leq r_{e}\right\}$ for some $r_{e}>M$.
\label{propp}
\end{proposition}
\begin{proof}
Recall that $Y=\partial_{r}$ and observe that
\begin{equation*}
J_{\mu}^{-Y}[\psi]n^{\mu}_{\Sigma}\sim (Y\psi)^{2}+\left|\nabb\psi\right|^{2}.
\end{equation*}
Therefore, if $P=P^{T}(r)T+P^{Y}(r)Y$ with $P^{T}>0$ and $P^{Y}\sim -\sqrt{\Delta}$, then we have 
\begin{equation*}
J_{\mu}^{P}[\psi]n^{\mu}_{\Sigma}\sim (T\psi)^{2}+\sqrt{\Delta}(Y\psi)^{2}+\left|\nabb\psi\right|^{2}\sim K^{N,-\frac{1}{2}}[\psi]
\end{equation*}
close to $\hh$. Regarding $K^{P}$ recall that 
\begin{equation*}
K^{P}[\psi]=F_{TT}(T\psi)^{2}+F_{YY}(Y\psi)^{2}+F_{TY}(T\psi)(Y\psi)+F_{\scriptsize\nabb}\left|\nabb\psi\right|^{2},
\end{equation*}
where 
\begin{equation*}
\begin{split}
&F_{TT}=\frac{r^{2}+a^{2}}{\rho^{2}}\frac{d}{dr}\!\!\left(P^{T}\right)-\frac{a^{2}\sin^{2}\theta}{2\rho^{2}}\frac{d}{dr}\!\!\left(P^{Y}\right),\\
&F_{YY}=\frac{\Delta}{2\rho^{2}}\frac{d}{dr}\!\!\left(P^{Y}\right)-\frac{\frac{d\Delta}{dr}}{2\rho^{2}}P^{Y},\\
&F_{\scriptsize\nabb}=-\frac{1}{2}\frac{d}{dr}\!\!\left(P^{Y}\right),\\
&F_{TY}=\frac{\Delta}{\rho^{2}}\frac{d}{dr}\!\!\left(P^{T}\right)-\frac{2r}{\rho^{2}}\left(P^{Y}\right).
\end{split}
\end{equation*}
If we now assume that $P^{Y}=-\sqrt{\Delta}$ then the coefficient $F_{TY}$ vanishes to first order on $\hh$ and $F_{YY}\sim \Delta$. Therefore, if we take $\frac{d}{dr}\!\!\left(P^{T}\right)$ to be sufficiently large then we have
\begin{equation*}
\begin{split}
F_{TY}=\frac{\Delta}{\rho^{2}}\frac{d}{dr}\!\!\left(P^{T}\right)-\frac{2r}{\rho^{2}}\left(P^{Y}\right)=\frac{\sqrt{\Delta}}{\rho}\left[\frac{\sqrt{\Delta}}{\rho}\frac{d}{dr}\!\!\left(P^{T}\right)+\frac{2r}{\rho}\right]\leq \epsilon \frac{\Delta}{\rho^{2}}+ \frac{1}{\epsilon}\left[\frac{\sqrt{\Delta}}{\rho}\frac{d}{dr}\!\!\left(P^{T}\right)+\frac{2r}{\rho}\right]^{2}.
\end{split}
\end{equation*}
If we take $\epsilon$ sufficiently small and $P^{T}$ such that $\frac{1}{\epsilon}\left[\frac{\sqrt{\Delta}}{\rho}\frac{d}{dr}\!\!\left(P^{T}\right)+\frac{2r}{\rho}\right]^{2}\leq \frac{d}{dr}\!\!\left(P^{T}\right)$ (which is always possible in view of the degeneracy of $\sqrt{\Delta}$ at $\hh$), then there exists a $r_{e}>M$ such that 
\begin{equation*}
\begin{split}
K^{P}[\psi]\sim (T\psi)^{2}+\frac{\Delta}{\rho^{2}}(Y\psi)^{2}+\left|\nabb\psi\right|^{2}\sim J_{\mu}^{T}[\psi]n^{\mu}_{\Sigma}
\end{split}
\end{equation*}
in $\mathcal{A}=\left\{M\leq r\leq r_{e}\right\}$. 

\end{proof}

\subsubsection{Uniform Boundedness of $P$-Energy}
\label{sec:UniformBoundednessOfPEnergy}
We next show that the $P$-flux is uniformly bounded. 
\begin{proposition}
There exists a constant $C$ that depends  on $M$  and $\Sigma_{0}$  such that for all axisymmetric solutions $\psi$ of the wave equation we have
\begin{equation}
\int_{\Sigma_{\tau}}{J_{\mu}^{P}[\psi]n^{\mu}_{\Sigma_{\tau}}}\leq C\int_{\Sigma_{0}}{J_{\mu}^{P}[\psi]n^{\mu}_{\Sigma_{0}}}.
\label{pboun}
\end{equation}
\label{pbound}
\end{proposition}
\begin{proof}
Stokes' theorem for the current $J_{\mu}^{P}$ gives us
\begin{equation*}
\int_{\Sigma_{\tau}}{J_{\mu}^{P}n^{\mu}}+\int_{\mathcal{H}^{+}}{J_{\mu}^{P}n^{\mu}}+\int_{\mathcal{I}^{+}}{J_{\mu}^{P}n^{\mu}}+\int_{\mathcal{R}}{K^{P}}=\int_{\Sigma_{0}}{J_{\mu}^{P}n^{\mu}}.
\end{equation*}
Note that since $P$ is a future-directed causal vector field, the boundary integrals over $\mathcal{H}^{+}$ and $\mathcal{I}^{+}$ are non-negative. The same also holds for $K^{P}$ in region $\mathcal{A}$, whereas $K^{P}$ vanishes far away from the horizon.  In the intermediate region this spacetime integral can be bounded using Proposition \ref{main2p}. Note how crucial it was to bound the left hand side of \eqref{main2} by the initial $T$-flux.  The result now follows from $J_{\mu}^{T}n^{\mu} \leq CJ_{\mu}^{P}n^{\mu}.$
\end{proof}

\subsection{Hierarchy of Estimates close to $\hh$}
\label{sec:HierarchyOfEstimatesCloseToHh}

We are now in a position to derive integrated decay for the $T$-energy in the region $\mathcal{A}$.
\begin{proposition}
There exists a constant $C$ that depends  on $M$ and $\Sigma_{0}$   such that for all axisymmetric solutions $\psi$ of the wave equation   we have
\begin{equation*}
\int_{\tau_{1}}^{\tau_{2}}{\left(\int_{\mathcal{A}\cap\Sigma_{\tau}}{J_{\mu}^{T}[\psi]n^{\mu}_{\Sigma_{\tau}}}\right)d\tau}\leq C\int_{\Sigma_{\tau_{1}}}{J_{\mu}^{P}[\psi]n^{\mu}_{\Sigma_{\tau_{1}}}}
\end{equation*}
and
\begin{equation*}
\int_{\tau_{1}}^{\tau_{2}}{\left(\int_{\mathcal{A}\cap\Sigma_{\tau}}{J_{\mu}^{P}[\psi]n^{\mu}_{\Sigma_{\tau}}}\right)d\tau}\leq C\int_{\Sigma_{\tau_{1}}}{J_{\mu}^{N}[\psi]n^{\mu}_{\Sigma_{\tau_{1}}}}
\end{equation*}
in an appropriate $\varphi_{\tau}$-invariant neighbourhood $\mathcal{A}$ of $\mathcal{H}^{+}$.
\label{hierhh}
\end{proposition}
\begin{proof}
From the divergence identity for the current $J_{\mu}^{P}$ and the boundedness of $P$-energy we have
\begin{equation*}
\int_{\mathcal{A}}{K^{P}}\leq C\int_{\Sigma_{\tau_{1}}}{J_{\mu}^{P}[\psi]n^{\mu}_{\Sigma_{\tau_{1}}}}
\end{equation*}
for a uniform constant $C$. Thus the first estimate follows from Proposition \ref{propp} and the coarea formula. Likewise, the second estimate follows from Theorem \ref{t1}, the boundedness of the non-degenerate $N$-energy and Proposition \ref{propp}.
\end{proof}

\subsection{Dafermos-Rodnianski Method}
\label{sec:DafermosRodnianskiMethod}

Dafermos and Rodnianski have recently presented a method \cite{new} that allows us to obtain estimates in a neighbourhood of $\mathcal{I}^{+}$ such as the ones of Proposition \ref{hierhh}. Recall the vector field $\partial_{p}$ defined in Section \ref{sec:TheInitialHypersurfaceSigma0}. By introducing the function $\phi=r\psi$ and applying an appropriate multiplier this method gives us the following
\begin{proposition}
There exists a constant $C$ that depends  on $M$ and $\Sigma_{0}$  such that if $\psi$ is an axisymmetric solution to the wave equation and $\tilde{N}_{\tau}=\Sigma_{\tau}\cap\left\{r\geq R_{e}\right\}$  with $R_{e}$ sufficiently large, then 
\begin{equation*}
\int_{\tau_{1}}^{\tau_{2}}{\left(\int_{\tilde{N}_{\tau}}{J^{T}_{\mu}[\psi]n^{\mu}_{\tilde{N}_{\tau}}}\right)d\tau} \,\leq\, C\int_{\Sigma_{\tau_{1}}}{J_{\mu}^{T}\left[\psi\right]n^{\mu}_{\Sigma_{\tau_{1}}}}+C\int_{\tilde{N}_{\tau_{1}}}{r^{-1}\left(\partial_{p}\phi\right)^{2}}
\end{equation*}
and 
\begin{equation*}
\int_{\tau_{1}}^{\tau_{2}}{\left(\int_{\tilde{N}_{\tau}}{r^{-1}(\partial_{p}\phi)^{2}}\right)d\tau}\ \leq C\int_{\Sigma_{\tau_{1}}}{J_{\mu}^{T}\left[\psi\right]n^{\mu}_{\Sigma_{\tau_{1}}}}+C\int_{\tilde{N}_{\tau_{1}}}{\left(\partial_{p}\phi\right)^{2}}.
\end{equation*}
\label{hierii}
\end{proposition}

\subsection{Decay of Energy}
\label{sec:DecayOfEnergy}

In view of Propositions \ref{hierhh}, \ref{hierii} and \ref{nkcor} we have
\begin{equation}
\int_{\tau_{1}}^{\tau_{2}}{\left(\int_{\Sigma_{\tau}}{J^{T}_{\mu}[\psi]n^{\mu}_{\Sigma_{\tau}}}\right)d\tau}\, \leq\,  CI^{T}_{\Sigma_{\tau_{1}}}[\psi],
\label{Integt1}
\end{equation}
where 
\begin{equation*}
\begin{split}
I^{T}_{\Sigma_{\tau}}[\psi]=&\int_{\Sigma_{\tau}}{J^{P}_{\mu}[\psi]n^{\mu}_{\Sigma_{\tau}}}+
\int_{\Sigma_{\tau}}{J^{T}_{\mu}[T\psi]n^{\mu}_{\Sigma_{\tau}}}+\int_{\tilde{N}_{\tau}}{r^{-1}\left(\partial_{p}\phi\right)^{2}}.
\end{split}
\end{equation*}
Moreover we have
\begin{equation}
\int_{\tau_{1}}^{\tau_{2}}{I_{\Sigma_{\tau}}^{T}[\psi]d\tau}\leq CI_{\Sigma_{\tau_{1}}}^{T}[T\psi]+C\int_{\Sigma_{\tau_{1}}}{J_{\mu}^{N}[\psi]n^{\mu}}+C\int_{\tilde{N}_{\tau_{1}}}{(\partial_{p}\phi)^{2}}.
\label{integt2}
\end{equation}
Note that the properties of the vector field $P$ are crucial for obtaining estimate \eqref{Integt1}. The above estimates are everything we need in order to show Theorem \ref{t3}. See \cite{aretakis2} for the details of our method, a summary of which is the following: We show $\frac{1}{\tau}$ decay for the energy using \eqref{Integt1}.  The estimate \eqref{integt2} allows us to conclude that $I_{\Sigma_{\tau}}^{T}[\psi]$ decays along a dyadic sequence. Then, applying again \eqref{Integt1} gives us faster decay for $\displaystyle\int_{\Sigma_{\tau}}{J^{T}_{\mu}[\psi]n^{\mu}_{\Sigma_{\tau}}}$. The uniform boundedness of the $T$-flux allows us to drop the restriction of the dyadic sequence. 

\section{Pointwise Estimates}
\label{sec:PointwiseEstimates}

In order to derive pointwise estimates one needs to bound higher order energies and apply Sobolev embeddings. In view of our discussion in Section \ref{sec:EnergyAndPointwiseDecay}, we need to commute with the symmetry operators of extreme Kerr. The symmetry operators of up to second order of Kerr  are the following
\begin{equation*}
\mathbb{S}_{0}=\left\{id\right\}, \ \ \mathbb{S}_{1}=\left\{T, \Phi\right\}\ \ \mathbb{S}_{2}=\left\{T^{2}, \Phi^{2}, T\Phi, Q\right\},
\end{equation*}
where $Q$ is the Carter operator (see Section \ref{sec:TheCarterOperatorAndSymmetries}). Recall that for axisymmetric functions $\psi$ we have $Q\psi=a^{2}\sin^{2}\theta (TT\psi)+\lapp_{\mathbb{S}^{2}}\psi$. The following lemma (see also \cite{bluekerr}) shows that one can obtain pointwise bounds using only symmetry operators of Kerr.
\begin{lemma}
There exists a constant $C$ which depends only on $M$ such that for all sufficiently regular functions $\psi$ we have
\begin{equation}
\left|\psi\right|^{2}\leq C\sum_{|k|\leq 2}\int_{\mathbb{S}^{2}}\big|S^{k}\psi\big|^{2},
\label{lemma}
\end{equation}
where  $\left|S^{k}\psi\right|^{2}=\sum_{S\in \mathbb{S}_{k}}\left|S\psi\right|^{2}$.
\label{lemkerr}
\end{lemma}
\begin{proof}
We immediately see
\begin{equation*}
\left|\lapp\psi\right|^{2}\leq C\left[(Q\psi)^{2}+(TT\psi)^{2}+(\Phi\Phi\psi)^{2}\right].
\end{equation*}
The lemma now follows from the following spherical Sobolev inequality
\begin{equation*}
\left|\psi\right|^{2}\leq C\int_{\mathbb{S}^{2}}\left|\psi\right|^{2}+\left|\lapp\psi\right|^{2}.
\end{equation*}
Note that the constant $C$ depends only on ($M$ and) the sphere $\mathbb{S}^{2}$.
\end{proof}

For previous use of the differential operator $Q$ see the discussion in Section \ref{sec:TheCarterOperatorAndSymmetries}.

\subsection{Uniform Pointwise Boundedness}
\label{sec:UniformPointwiseBoundedness}

We next show that all axisymmetric solutions $\psi$ to the wave equation remain uniformly bounded in $\mathcal{R}$.  We work with the foliation $\Sigma_{\tau}$  and  the induced coordinate system $(\rho,\omega)$ defined in Section \ref{sec:TheInitialHypersurfaceSigma0}. For $r\geq M$   we have 
\begin{equation*}
\begin{split}
\psi^{2}\left(r,\omega\right)=\left(\int_{r}^{+\infty}{\left(\partial_{p}\psi\right)dp}\right)^{2}\leq\left(\int_{r}^{+\infty}{\left(\partial_{p}\psi\right)^{2}p^{2}dp}\right)\left(\int_{r}^{+\infty}{\frac{1}{p^{2}}dp}\right)
=\frac{1}{r}\left(\int_{r}^{+\infty}{\left(\partial_{p}\psi\right)^{2}p^{2}dp}\right).
\end{split}
\end{equation*}
Therefore,
\begin{equation}
\begin{split}
\int_{\mathbb{S}^{2}}{\psi^{2}(r_{0},\omega)d\omega}&\leq\frac{1}{r_{0}}\int_{\mathbb{S}^{2}}{\int_{r_{0}}^{+\infty}{\left(\partial_{p}\psi\right)^{2}p^{2}dp d\omega}}\leq \frac{C}{r_{0}}\int_{\Sigma_{\tau}\cap\left\{r\geq r_{0}\right\}}{J_{\mu}^{N}[\psi]n^{\mu}_{\Sigma_{\tau}}},
\end{split}
\label{1pointwise}
\end{equation}
where $C$ is a constant that depends only on $M$ and $\Sigma_{0}$. 
\begin{theorem}
There exists a constant $C$ which depends on $M$ and $\Sigma_{0}$ such that for all axisymmetric solutions $\psi$ of the wave equation we have
\begin{equation}
\begin{split}
\left|\psi\right|^{2}\leq C\cdot E_{2}[\psi]\frac{1}{r},
\end{split}
\label{2pointwise}
\end{equation}
where $E_{2}[\psi]=\sum_{\left|k\right|\leq 2}{\displaystyle\int_{\Sigma_{0}}{J_{\mu}^{N}[S^{k}\psi]n^{\mu}_{\Sigma_{0}}}}.$
\label{unponpsi}
\end{theorem}
\begin{proof}
The theorem follows from Lemma \ref{lemkerr}, the uniform boundedness of the $N$-flux and \eqref{2pointwise}.
\end{proof}

\subsection{Pointwise Decay}
\label{sec:PointwiseDecay}

\subsubsection{Decay away from $\mathcal{H}^{+}$}
\label{sec:DecayAwayMathcalH}

We consider the region $\left\{r\geq R_{1}\right\}$, where $R_{1}>M$. From now on, $C$ will be a constant depending only on $M$, $R_{1}$ and $\Sigma_{0}$.
\begin{figure}[H]
	\centering
		\includegraphics[scale=0.1]{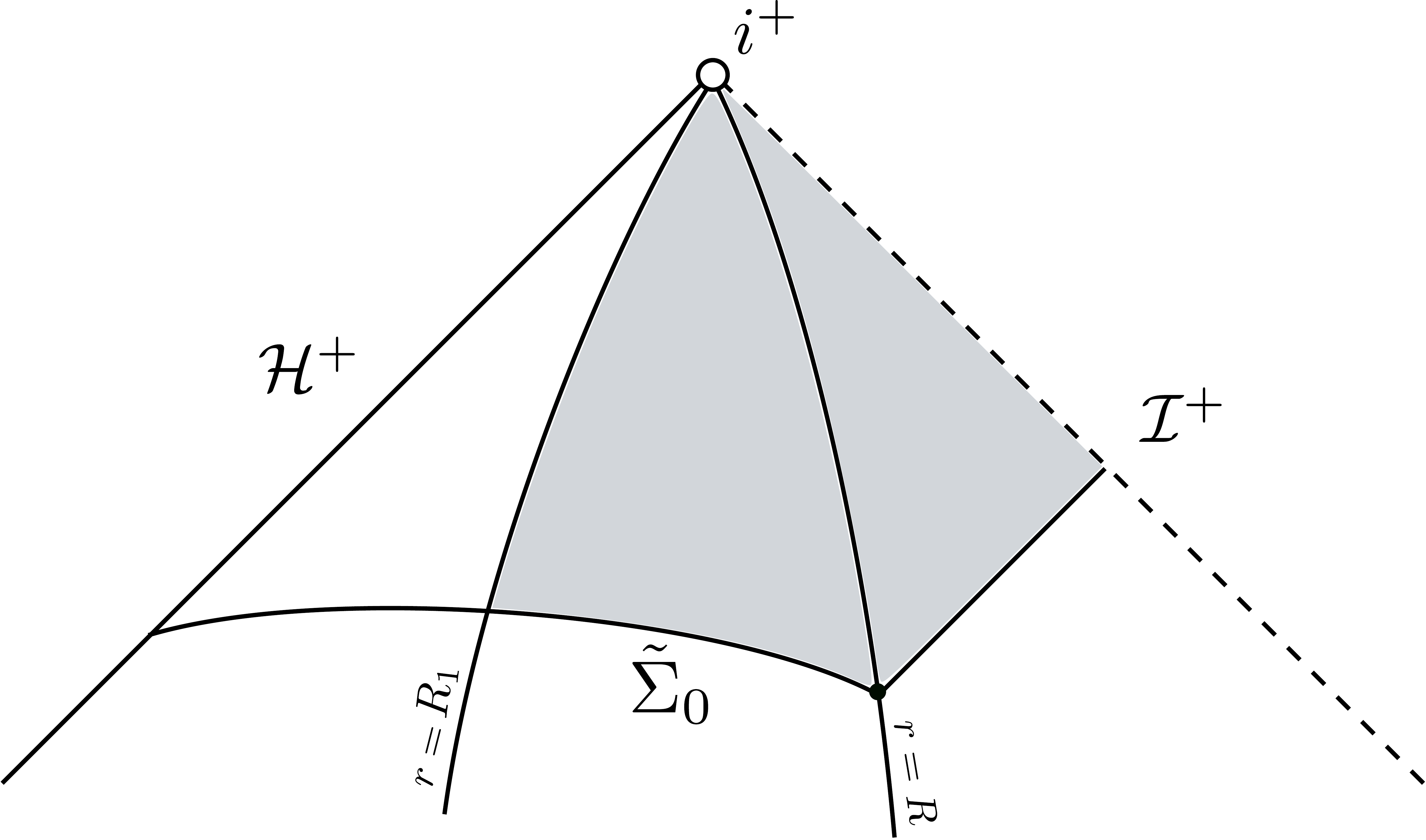}
	\label{fig:1pd}
\end{figure}
Clearly, in this region we have $J_{\mu}^{N}n^{\mu}_{\Sigma}\sim  J_{\mu}^{T}n_{\Sigma}^{\mu}$ and $\sim$ depends on $R_{1}$. Therefore, from \eqref{1pointwise} we have that for any $r\geq R_{1}$  
\begin{equation*}
\int_{\mathbb{S}^{2}}{\psi^{2}(r,\omega)d\omega}\leq \frac{C}{r}\int_{\Sigma_{\tau}}{J_{\mu}^{T}[\psi]n^{\mu}_{\Sigma_{\tau}}}\leq C\cdot E_{1}[\psi]\frac{1}{r\cdot \tau^{2}}.
\end{equation*}
Commuting with the symmetry operators $S^{k}$ for $|k|\leq 2$ and using Lemma \ref{lemkerr} yields
\begin{equation*}
\left|\psi\right|^{2}\leq C E_{3}\frac{1}{r\cdot \tau^{2}},
\end{equation*}
where  $E_{3}=\sum_{\left|k\right|\leq 2}{E_{1}\left[S^{k}\psi\right]}.$ Next we improve the decay towards the null infinity $\mathcal{I}^{+}$. For all $r\geq R_{1}$ we have
\begin{equation*}
\begin{split}
\int_{\mathbb{S}^{2}}{(r\psi)^{2}(r,\omega)d\omega}&=\int_{\mathbb{S}^{2}}{(R_{1}\psi)^{2}(R_{1},\omega)d\omega}+2\int_{\mathbb{S}^{2}}\int_{R_{1}}^{r}{\frac{\psi}{\rho}\partial_{\rho}(\rho\psi)\rho^{2}d\rho d\omega}\\
&\leq CE_{1}[\psi]\frac{1}{\tau^{2}}+C\sqrt{\int_{\Sigma_{\tau}\cap\left\{r\geq R_{1}\right\}}\frac{1}{\rho^{2}}\psi^{2}\int_{\Sigma_{\tau}\cap\left\{r\geq R_{1}\right\}}{\left(\partial_{\rho}(\rho \psi)\right)^{2}}}.
\end{split}
\end{equation*}
From  the first Hardy inequality  we have
\begin{equation*}
\int_{\Sigma_{\tau}}{\frac{1}{\rho^{2}}\psi^{2}}\leq C\int_{\Sigma_{\tau}}{J_{\mu}^{T}[\psi]n^{\mu}_{\Sigma_{\tau}}}\leq CE_{1}[\psi]\frac{1}{\tau^{2}}.
\end{equation*}
Moreover, from the Dafermos-Rodnianski method one obtains (see also \cite{aretakis2})
\begin{equation*}
\begin{split}
\int_{\Sigma_{\tau}\cap\left\{r\geq R_{1}\right\}}{\left(\partial_{\rho}(\rho \psi)\right)^{2}}\leq  C\int_{\Sigma_{0}}{J_{\mu}^{T}[\psi]n_{\Sigma_{0}}^{\mu}}+\int_{\tilde{N}_{0}}{(\partial_{\rho}(\rho \psi))^{2}}.
\end{split}
\end{equation*}
 Hence, by recalling the expression for $E_{1}[\psi]$ we obtain
 \begin{equation*}
\begin{split}
r^{2}\int_{\mathbb{S}^{2}}{\psi^{2}(r,\omega)d\omega}\leq C E_{1}[\psi]\frac{1}{\tau}
\end{split}
\end{equation*} 
and therefore, 
\begin{equation*}
\begin{split}
\psi^{2}\leq CE_{3}[\psi]\frac{1}{r^{2}\cdot\tau}.
\end{split}
\end{equation*}

\subsubsection{Decay near $\hh$}
\label{sec:DecayNearHh}

We now derive decay for $\psi$ on neighbourhoods of $\mathcal{H}^{+}$. We first need the following
\begin{lemma}
There exists a constant $C$ which depends only on $M$ such that for all $r$ with $M<r$ and all axisymmetric solutions $\psi$ of the wave equation we have
\begin{equation*}
\begin{split}
\int_{\mathbb{S}^{2}}{\psi^{2}(r_{1},\omega)d\omega}\leq\frac{C}{(r_{1}-M)^{2}}\frac{E_{1}[\psi]}{\tau^{2}}.
\end{split}
\end{equation*}
\label{1lemmadecay}
\end{lemma}

\begin{proof}
Using \eqref{1pointwise} we obtain 
\begin{equation*}
\begin{split}
\int_{\mathbb{S}^{2}}{\psi^{2}(r_{1},\omega)d\omega}&\leq \frac{C}{r_{1}}\int_{\Sigma_{\tau}\cap\left\{r\geq r_{1}\right\}}{J_{\mu}^{N}[\psi]n^{\mu}_{\Sigma_{\tau}}}=\frac{C}{r_{1}}\int_{\Sigma_{\tau}\cap\left\{r\geq r_{1}\right\}}{
\frac{D(\rho)}{D(\rho)}J_{\mu}^{N}[\psi]n^{\mu}_{\Sigma_{\tau}}}\\ &\leq
\frac{C}{(r_{1}-M)^{2}}\int_{\Sigma_{\tau}}{J_{\mu}^{T}[\psi]n^{\mu}_{\Sigma_{\tau}}}\leq\frac{C}{(r_{1}-M)^{2}}\frac{E_{1}}{\tau^{2}}.
\end{split}
\end{equation*}
\end{proof}

Let $r_{0}\in[M,R_{1}]$ and consider  the hypersurface $\gamma=\left\{r=r_{0}+\tau^{-\frac{1}{2}}\right\}$.
\begin{figure}[H]
	\centering
		\includegraphics[scale=0.11]{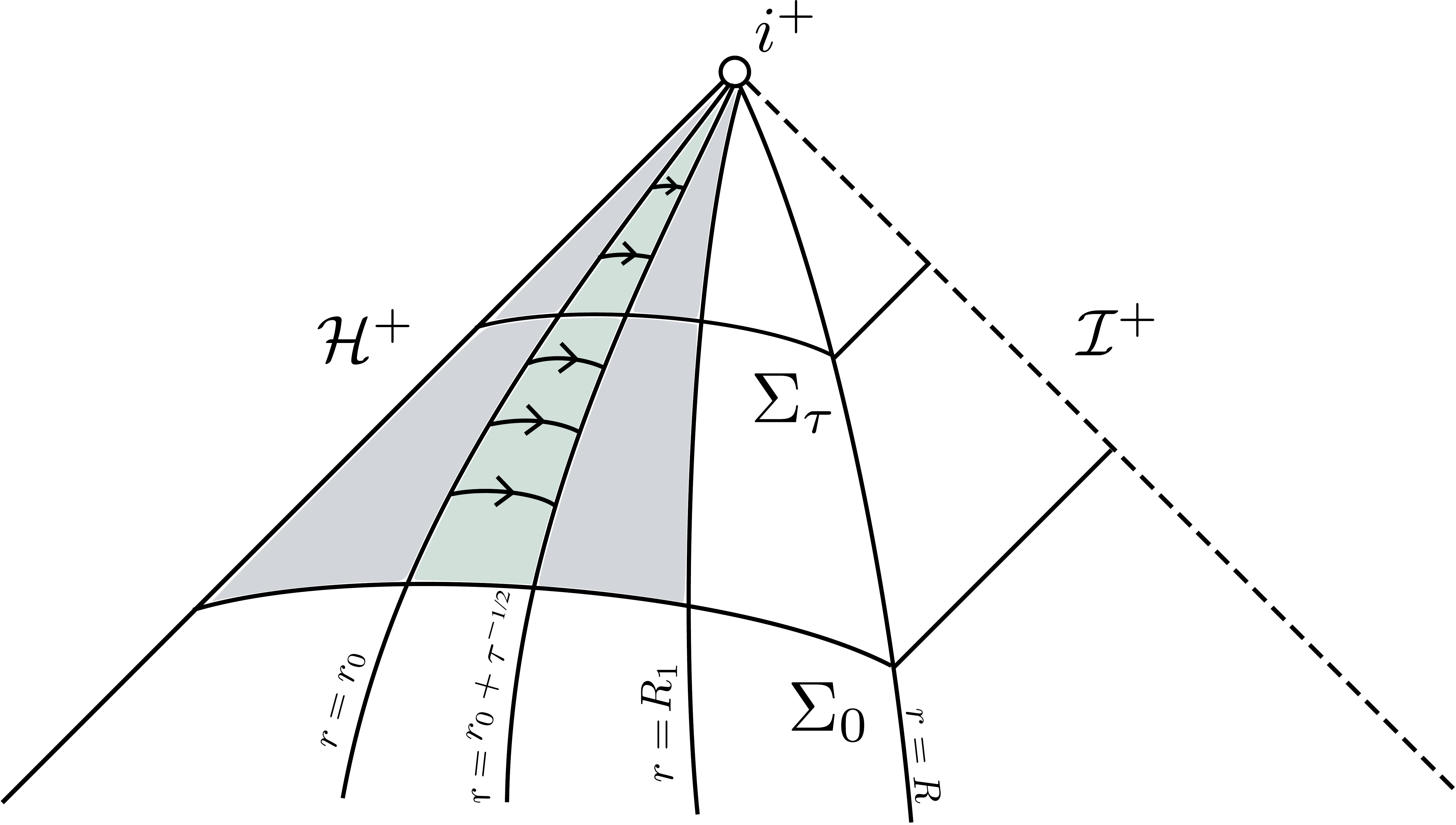}
		\label{fig:2decaypoint}
\end{figure}
By applying Stokes' theorem for the hypersurfaces shown in the figure above we obtain
\begin{equation*}
\begin{split}
\int_{\mathbb{S}^{2}}{\psi^{2}(r_{0},\omega)}\leq\int_{\mathbb{S}^{2}}{\psi^{2}(r_{0}+\tau^{-\frac{1}{2}},\omega)}+C\int_{\Sigma_{\tau}\cap\left\{r_{0}\leq r\leq r_{0}+\tau^{-\frac{1}{2}}\right\}}{\psi(\partial_{\rho}\psi)}.
\end{split}
\end{equation*}
Applying \ref{1lemmadecay} for the first term on the right hand side (note that $M<r_{0}+\tau^{-\frac{1}{2}}$, Cauchy-Schwarz  for the second term and the first Hardy inequality and Theorem   \ref{t3} yields
\begin{equation*}
\begin{split}
\int_{\mathbb{S}^{2}}{\psi^{2}(r_{0},\omega)d\omega}&\leq CE_{1}[\psi]\frac{1}{\tau}+C\sqrt{E_{1}[\psi]}\frac{1}{\tau}\sqrt{\int_{\Sigma_{\tau}\cap\left\{r_{0}\leq r\leq r_{0}+\tau^{-\frac{1}{2}}\right\}}{\!\!(\partial_{\rho}\psi)^{2}}}.
\end{split}
\end{equation*}
Using the uniform boundedness of the $N$-flux and the expression for $E_{1}[\psi]$ we obtain
 \begin{equation*}
\begin{split}
\int_{\mathbb{S}^{2}}{\psi^{2}(r,\omega)d\omega}\leq C E_{1}[\psi]\frac{1}{\tau}
\end{split}
\end{equation*} 
and therefore, 
\begin{equation*}
\begin{split}
\psi^{2}\leq CE_{3}[\psi]\frac{1}{\tau}.
\end{split}
\end{equation*}
This completes the proof of Theorem \ref{t5}.

\section{Acknowledgements}
\label{sec:Acknowledgements}

I would like to thank Mihalis Dafermos for introducing to me the problem and for his teaching and advice. I would also like to thank Willie Wong for several helpful discussions. I am supported by a Bodossaki Grant and a Grant from the European Research Council.

\end{document}